\documentclass[acmsmall,screen]{acmart}

\usepackage{wrapfig}
\usepackage{tikz}
\usepackage{subcaption}
\usepackage{thmtools} 
\usepackage{thm-restate}
\usepackage{enumerate}
\usepackage[inline]{enumitem}

\newcommand{\figlabel}[1]{\label{fig:#1}}
\newcommand{\figref}[1]{Figure~\ref{fig:#1}}

\newcommand{\seclabel}[1]{\label{sec:#1}}
\newcommand{\secref}[1]{Section~\ref{sec:#1}}

\newcommand{\deflabel}[1]{\label{def:#1}}
\newcommand{\defref}[1]{Definition~\ref{def:#1}}

\newcommand{\thmlabel}[1]{\label{thm:#1}}
\newcommand{\thmref}[1]{Theorem~\ref{thm:#1}}
\newcommand{\proplabel}[1]{\label{prop:#1}}
\newcommand{\propref}[1]{Proposition~\ref{prop:#1}}
\newcommand{\lemlabel}[1]{\label{lem:#1}}
\newcommand{\lemref}[1]{Lemma~\ref{lem:#1}}
\newcommand{\corlabel}[1]{\label{cor:#1}}
\newcommand{\corref}[1]{Corollary~\ref{cor:#1}}
\newcommand{\claimlabel}[1]{\label{claim:#1}}
\newcommand{\claimref}[1]{Claim~\ref{claim:#1}}

\newcommand{\applabel}[1]{\label{app:#1}}

\theoremstyle{definition}

\newtheorem{theorem}{Theorem}[section]
\newtheorem{lemma}[theorem]{Lemma}
\newtheorem{proposition}[theorem]{Proposition}
\newtheorem{claim}[theorem]{Claim}
\newtheorem{corollary}[theorem]{Corollary}

\newtheorem{definition}{Definition}[section]

\newenvironment{myparagraph}[1]{\vspace{0.1in}\noindent{\bf #1.}}{}

\DeclareMathAlphabet{\mathpzc}{OT1}{pzc}{m}{it}

\newcommand{\nats}{\mathbb{N}}
\newcommand{\natsp}{\mathbb{N}_{>0}}
\newcommand{\set}[1]{\{#1\}}
\newcommand{\setpred}[2]{\{#1 \,|\, #2\}}
\newcommand{\proj}[2]{#1|_{#2}}
\renewcommand{\emptyset}{\varnothing}

\newcommand{\partialto}{\hookrightarrow}
\newcommand{\vect}[1]{\overline{#1}}

\newcommand{\powset}[1]{\mathcal{P}(#1)}

\newcommand{\ev}[1]{\langle #1\rangle}
\newcommand{\mems}{\mathcal{X}}

\newcommand{\threads}{\mathcal{T}}

\newcommand{\ThreadOf}[1]{\mathsf{thr}(#1)}
\newcommand{\OpOf}[1]{\mathsf{op}(#1)}
\newcommand{\MemOf}[1]{\mathsf{mem}(#1)}

\newcommand{\events}[1]{\mathsf{Events}_{#1}}
\newcommand{\tr}{\sigma}
\newcommand{\rf}[1]{\mathsf{rf}_{#1}}
\newcommand{\po}[1]{\mathsf{po}_{#1}}
\newcommand{\reads}[1]{\mathsf{Reads}_{#1}}
\newcommand{\writes}[1]{\mathsf{Writes}_{#1}}

\newcommand{\maz}{\mathcal{M}}
\newcommand{\mazeq}{\equiv_{\maz}}

\newcommand{\indrel}{\mathbb{I}}

\newcommand{\labs}{\Sigma}

\newcommand{\rfeq}{\equiv_{\rf{}}}

\newcommand{\slcsymb}{{\rightsquigarrow_{s}}}
\newcommand{\slice}{\protect\slcsymb}
\newcommand{\sliceto}[2]{{#1} \slice {#2}}
\newcommand{\kslice}[1]{\overset{\scalebox{0.6}{(${#1}$)}}{\rightsquigarrow}_{s}}
\newcommand{\ksliceto}[3]{{#2} \kslice{#1} {#3}}
\newcommand{\mslice}{\kslice{\infty}}

\newcommand{\sliceq}{\rightsquigarrow^*_{s}}
\newcommand{\sliceqto}[2]{{#1} \sliceq {#2}}

\newcommand{\notkslice}[1]{\overset{\scalebox{0.6}{(${#1}$)}}{\not\rightsquigarrow}_{s}}
\newcommand{\notksliceto}[3]{{#2} \notkslice{#1} {#3}}

\newcommand{\sheight}[2]{\sf{h}_{s}(#1, #2)}

\newcommand{\kmaz}[1]{\overset{\scalebox{0.6}{(${#1}$)}}{\equiv}_{\maz}}
\newcommand{\Inv}{\textsf{inversions}}

\newcommand{\pre}[2]{\textsf{Pre}_{#2}(#1)}
\newcommand{\post}[2]{\textsf{Post}_{#2}(#1)}

\newcommand{\aut}{\mathcal{A}}
\newcommand{\states}{Q}
\newcommand{\trns}{\delta}
\newcommand{\accpt}{F}
\newcommand{\init}{q_0}

\newcommand{\opcl}[2]{{#1}^{\kslice{#2}}}
\newcommand{\autcl}[1]{\opcl{\aut}{#1}}

\newcommand{\cnst}{\sf cnst}
\newcommand{\opcnst}[1]{{#1}^{\cnst}}
\newcommand{\autcnst}{\opcnst{\aut}}
\newcommand{\statescnst}{\opcnst{\states}}

\newcommand{\memb}{\sf memb}
\newcommand{\opmemb}[1]{{#1}^{\memb}}
\newcommand{\autmemb}{\opmemb{\aut}}
\newcommand{\statesmemb}{\opmemb{\states}}
\newcommand{\trnsmemb}{\opmemb{\trns}}
\newcommand{\accptmemb}{\opmemb{\accpt}}
\newcommand{\initmemb}{\opmemb{\init}}

\newcommand{\subsq}[1]{\mathbf{i}}

\newcommand{\grains}{\mathcal{G}} 
\newcommand{\graineq}{\equiv_{\grains}}

\newcommand{\opfont}[1]{\texttt{#1}}
\newcommand{\acq}{\opfont{acq}}
\newcommand{\rel}{\opfont{rel}}
\newcommand{\rd}{\opfont{r}}
\newcommand{\wt}{\opfont{w}}

\newcommand{\btx}{\opfont{bgn}} 
\newcommand{\etx}{\opfont{end}} 

\newcommand{\lk}{\ell}

\newcommand{\ord}[2]{\leq^{#1}_{\mathsf{#2}}}

\newcommand{\trord}[1]{\ord{#1}{}}

\newcommand*\squaresmall[1]{\tikz[baseline=(char.base)]{
  \node[shape=rectangle,draw=black,fill=green!80!black!20!white,inner sep=1.5pt, solid] (char) {\textcolor{black}{\tt\footnotesize{#1}}};}}

\def\marrow{{\marginpar[\hfill\(\longrightarrow\)]{\(\longleftarrow\)}}}
\def\umang#1{({\sc Umang says: }{{\color{brown}\marrow\sf #1})}}

\newcommand{\executionfull}[9]{
\scalebox{#3}{
  \begin{tikzpicture}
    \foreach \x in {1,...,#1}
    \node[right] at (#4*\x+#6, #8) {$T_{\x}$}; 
    \draw (#7,0) -- (#1*#4+#7,0); 
    \pgfmathsetmacro{\y}{1};
    #2 
    \draw (#7,0) -- (#7,-#5*\y); 
    \draw (#1*#4+#7,0) -- (#1*#4+#7,-#5*\y); 
    \draw (#7,-#5*\y) -- (#1*#4+#7,-#5*\y); 
    \ifthenelse{#9 = 1}{
      \foreach \x in {2,...,#1}
      \draw[dashed] (#4*\x+#7-#4,0) -- (#4*\x+#7-#4,-#5*\y); 
    }{}
  \end{tikzpicture}
}
}

\newcommand{\figevfull}[9]{
\ifthenelse{#7 = 1}{
  \ifthenelse{#8 = -1}{
    \node [left] at (#5,(-#4*\y))  {\pgfmathprintnumber{\y}};%
  }{
    \node [left] at (#5,(-#4*\y))  {#9};%
  }
}{}
\node at (#1*#3 + #6,(-#4*\y)) {$ #2 $};%
\pgfmathsetmacro{\y}{\y+1};
}

\begin{document}

\title{Parametrizing Reads-From Equivalence for Predictive Monitoring}

\author{Azadeh Farzan}
\orcid{0000-0001-9005-2653}
\affiliation{%
  \institution{University of Toronto}
  \city{Toronto}
  \country{Canada}
}
\email{azadeh@cs.toronto.edu}

\author{Umang Mathur}
\orcid{0000-0002-7610-0660}
\affiliation{%
  \institution{National University of Singapore}
  \city{Singapore}
  \country{Singapore}
}
\email{umathur@nus.edu.sg}


\begin{abstract}
Predictive runtime monitoring asks whether 
a given execution $\tr$ of a concurrent program can be used
to \emph{soundly predict} the existence of another execution $\rho$
(obtained by reordering $\tr$ without re-executing the program)
that satisfies a property $\varphi$.
Such techniques enhance the coverage of traditional runtime monitoring
and mitigate the effects of scheduling non-determinism.

The effectiveness and efficiency of predictive monitoring are governed by 
two, often conflicting, factors:
(a) the complexity of the specification $\varphi$, and 
(b) the expressive power of the space of reorderings that must be explored.
When one considers the largest space of reorderings, 
namely those induced by \emph{reads-from equivalence},
the predictive monitoring problem becomes intractable,
even for very simple specifications such as data races. 
At the other extreme, restricting reasoning to commutativity-based reorderings 
in the style of Mazurkiewicz's trace equivalence
yields fast and space-efficient
algorithms for simple properties.
However, under trace equivalence, predictive monitoring 
remains intractable for the full class of regular language specifications,
despite the significantly reduced predictive power arising from the smaller
space of reorderings.

In this work, we address this fundamental tradeoff through an orthogonal approach
based on \emph{parametrization}.
We introduce a notion of \emph{sliced reorderings}, along with its parametric
generalization, \emph{$k$-sliced reorderings}.
Informally, an execution $\rho$ is a $k$-sliced reordering of an execution $\tr$
if $\tr$ can be partitioned into $k+1$ ordered subsequences such that
concatenating these subsequences yields $\rho$,
while preserving program order and reads-from constraints.

Our main results are twofold.
First, we show that $k$-sliced reorderings form a strictly 
increasing hierarchy of expressive power that converges to reads-from equivalence as $k$ increases,
establishing completeness of our parametrization in the limit.
Second, for any fixed $k$, the predictive monitoring
problem modulo $k$-sliced reorderings against any regular specification 
can be solved using a constant space
streaming algorithm.
Together, these results position $k$-sliced reorderings as an effective alternative
to existing equivalence relations on concurrent executions, 
yielding a uniform parametrized approach to predictive monitoring,
whose expressive power can be
systematically traded off against computational resources.
\end{abstract}


\maketitle


\section{Introduction}
\seclabel{intro}

Runtime verification has emerged as a promising and
practical class of techniques for ensuring reliability
of software.
The core algorithmic question one asks here is the standard {\em monitoring problem} 
that asks whether a given program run 
belongs to a specification language (often representing erroneous program runs) 
presented as a {\em monitor}. 
If the specification is {\em regular} then the {\em monitor} 
can be expressed as a finite state machine, 
and the problem can be solved in {\em constant space}, 
that is, if we assume the size of the alphabet 
to be constant and only consider the run as an input to this 
algorithm (and not the {\em monitor} itself). 

This work is motivated by the {\em predictive monitoring problem} that arises
in the context of concurrent programs. 
Here, even when the run (of a concurrent program) under consideration is not 
erroneous according to the specification, 
we ask if one can {\em predict} the existence of another erroneous program run.  
At a high level, thus, predictive monitoring offers the promise of
enhancing the coverage of vanilla (non-predictive) monitoring techniques.
The predictive monitoring problem is defined in terms of three components: 
(1) a concurrent program run $\sigma$, 
(2) a specification language $S$, which defines a set of erroneous runs of a concurrent program, 
and (3) a sound {\em predictor}, which using $\sigma$,
{\em soundly} reasons about a set of other runs that 
are guaranteed to be valid executions of the same program. 
A predictor is typically defined based on a 
{\em sound equivalence} relation $\equiv$.  
{\em Soundness} of the relation $\equiv$ guarantees the following: 
$\sigma \equiv \rho$ implies that, 
$\sigma$  is a feasible run of the program iff $\rho$ is a feasible run of the program.
%
Together, the {\em predictive monitoring problem},
for a specification $S$ and a sound equivalence relation $\equiv$,
can be formally stated as:
\begin{quote}
Given a program run $\sigma$ as input, check if there exists a run
$\rho$ such that $\rho \equiv \sigma$ {\bf and} $\rho \in S$.
\end{quote}

Ideally, one would like to have a similar algorithmic setup 
for the predictive monitoring problem as the original vanilla monitoring problem: 
a {\em constant space streaming} algorithm, that is,
a streaming algorithm whose memory usage
does not depend on the length of the input run, 
but only depends on how sophisticated the specification is. 
It must be noted that the task at hand depends on both 
the choice of (1) the equivalence relation as well as on (2) the specification. 
Indeed, each of these factors can orthogonally 
impact the complexity of predictive monitoring.

The {\em largest sound} equivalence relation \cite{serbanuta2013maximal,Abdulla2019} 
that can be used as a predictor is the {\em reads-from equivalence} 
(denoted $\rfeq$). 
It declares two concurrent program executions equivalent 
if (1) they have the same set of events, 
(2) order events of a given thread in the same way and ensure that 
(3) reads observe the same writes (see~\secref{prelim} for a formal definition).
Recently, it was shown ~\cite{FarzanMathur2024} that 
if one fixes the specification at the level of a {\em very simple} 
regular specification called {\em causal concurrency} 
(which asks if two events can be reordered in an equivalent run)
then the predictive monitoring problem modulo $\rfeq$ 
cannot be solved in a constant space streaming fashion. 

Another prominent notion of equivalence is
that of Mazurkiewicz's commutativity-based trace equivalence~\cite{Mazurkiewicz87}
which deems two executions equivalent if 
they can be transformed into each other through repeated swaps 
of neighbouring non-conflicting events (see~\secref{prelim} for a formal definition).
Trace equivalence has remained a popular choice for addressing 
predictive monitoring problems against a specific
class of specifications such as causal concurrency,  data races, 
and conflict serializability~\cite{Farzan2006,ang2023predictive,FM08CAV,Tunc2023,AngMathurPrefixes2024,Elmas07}.
This popularity stems largely from the fact that, against these select few specifications,
predictive monitoring (modulo trace equivalence) can be performed efficiently,
using a constant space streaming algorithm.
Nevertheless, it suffers from two key limitations.
First, aside from a limited class of regular specifications (including those above),
one cannot design constant space predictive monitoring against
arbitrary regular specifications~\cite{Ochmanski85}. 
Second, trace equivalence is known to be a strict refinement of $\rfeq$ and thus has
limited expressive power.
The challenge of going beyond the expressive power of trace equivalence
using bespoke algorithms for data races and deadlocks
has gained a lot of traction owing to its
practical implications~\cite{Huang14,Smaragdakis12,Kini17,Mathur21,OSR2024,Pavlogiannis2020,Mathur2020b,Roemer18,Kalhauge2018,Tunc2023}.
In line with these works, here, we ask the following high-level question:
\begin{quote}
{\em Is there a sound and sufficiently expressive predictor 
for which predictive monitoring
against arbitrary regular specifications can be solved
efficiently (i.e., using a constant space streaming algorithm)?}
\end{quote}

Recently, Farzan and Mathur proposed two new notions of equivalence
based on {\em grain commutativity}~\cite{FarzanMathur2024}, 
which are strictly more expressive than trace equivalence 
and strictly less expressive than $\rfeq$. 
They showed that, modulo these new equivalences,
the {\em causal concurrency} specification can be predictively monitored 
using a constant space streaming algorithm, despite the larger expressive power than trace equivalence. 
Nevertheless, this proposal also suffers from the same problems: 
these equivalences continue to be strictly less expressive than $\rfeq$ 
{\em and} they do not yield tractable predictive monitoring algorithms
against {\em arbitrary} regular specifications. 
In fact, as we show in \thmref{beyond-trace},
any predictor that subsumes the expressive power of 
trace equivalence also suffers from the same tractability problem as trace equivalence. 
Hence, as such, the solution to this problem does not lie in a magic point 
in the space of predictors that are in between trace equivalence and $\rfeq$. 

Ideally, one wants an equivalence as expressive as 
$\rfeq$ that can be used to predictively monitor any 
regular specification in constant space. Since, this is theoretically impossible, 
this article puts forward a proposal for the next best feasible approach. 
We propose a novel {\em constant space} solution 
to the predictive monitoring problem that 
\squaresmall{G1} works with  {\em any regular} specification, and 
\squaresmall{G2} presents a {\em pay-as-you-go} strategy of an 
increasing measure of expressiveness for the predictor, where 
\squaresmall{G3} this expressiveness can {\em provably} reach the {\em ideal limit} (i.e $\rfeq$).

As such, achieving the combination of goals \squaresmall{G1} and 
\squaresmall{G2} alone is not that hard. 
As an example, consider the 
{\em pay-as-you-go} model for trace equivalence,
where one bounds the number of swaps (say by some number $k \in \nats$).
The resulting parametric version of trace equivalence (where the parameter is this bound $k$),
is a sound predictor that predicts an execution that is at most $k$ swaps away from
an initial execution.
Such a parametrization of trace equivalence directly achieves \squaresmall{G2}
because you can enhance its expressive power by successively increasing
the value of this parameter $k$.
It also helps meet goal \squaresmall{G1} in that, for a fixed value of this parameter,
predictive monitoring can be performed modulo this predictor using a constant space streaming
algorithm against arbitrary regular specifications\footnote{This is easy enough to see for experts in the area, but it does not appear anywhere in the literature, hence we include the formal results in~\secref{kmaz}}.
However, the expressive power of this parametrized predictor, even in the limit, 
does not go beyond trace equivalence 
(which is strictly less expressive than $\rfeq$) and hence
it cannot meet goal \squaresmall{G3}.

Let us now turn our attention to our new parametric predictor
that seamlessly achieves all three goals \squaresmall{G1}, \squaresmall{G2} and \squaresmall{G3}.
It breaks away from the traditional style of commutativity-based reasoning
(for trace and grain equivalences) to an entirely new scheme based on the concept of {\em slices},
which we systematically investigate in this work.

\begin{wrapfigure}[12]{r}{0.4\textwidth}\vspace{-20pt}
\includegraphics[scale=0.4]{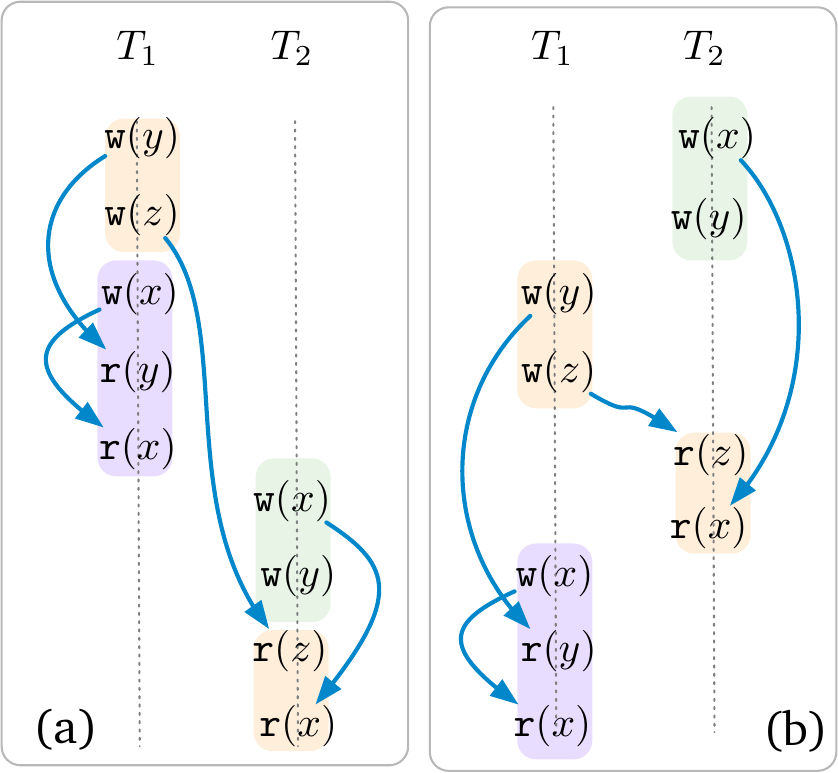}\vspace{-10pt}
\end{wrapfigure}
Let us illustrate how slices work.
Consider the two runs on the right. 
The reads-from relation is illustrated using (blue) arrows. 
It is easy to quickly verify that they are $\rfeq$-equivalent, 
but are not equivalent up to weaker equivalences such as trace or grain equivalence.
Our new predictor nevertheless deems them equivalent, 
and structures its reasoning as follows.
First, it identifies that the run (b) can be divided into
three successive contiguous sub-sequences
(we highlight these using a different color: green, followed by orange, followed by purple).
It then checks that the subsequences thus identified
satisfy two properties:
(i) each subsequence (appearing contiguously in (b))
appears in (a) also as a (possibly dispersed) subsequence, i.e.,
the order of events within the same subsequence 
is preserved across (a) and (b), and
(ii) the rearrangement from (b) to (a) 
does not break the program order or the reads-from relation.

One can rephrase the above reasoning to describe
the transformation from execution (a) to (b).
This transformation is obtained by identifying 
a collection of three subsequences (green, orange and purple) 
that are disjoint and cover all events of the execution (a),
together with an ordering amongst them: green < orange < purple,
and placing them contiguously one after the other to obtain (b), 
in the identified order (green < orange < purple), while ensuring that
condition (ii) holds. 
We call each subsequence a {\em slice} and 
call this transformation 
(from (a) to (b)) as a $(3-1)=2$-slice reordering,
and denote it as $\ksliceto{2}{(a)}{(b)}$.
In general, the number of slices we permit becomes the {\em parameter} 
in the definition of such a move. 
The formal definition of such a parametric move, with parameter $k \in \nats$,
is thus the following:
\begin{quote}
{\bf $k$-slice reordering}. Execution $\rho$ is a $k$-slice reordering of
execution $\sigma$ (denoted $\ksliceto{k}{\sigma}{\rho}$) if $\sigma \rfeq \rho$ and further,
$\sigma$ can be partitioned into an ordered sequence of $k+1$ 
subsequences $\sigma_1, \sigma_2, \ldots, \sigma_{k+1}$
such that $\rho = \sigma_1 \cdot \sigma_2 \cdots \sigma_{k+1}$
is obtained by concatenating these $k+1$ (contiguous) subsequences in order.
\end{quote}

The ask for condition (ii) to hold may appear too permissive,
given that this is no different from asking for reads-from equivalence.
The difference (from the full-fledged $\rfeq$) crucially lies in
bounding the number of subsequences we allow in such a move,
because doing so also restricts the number of reads-from and program-order that could
have been violated in an arbitrary ordered concatenation of the subsequences identified.
In this sense, this parameter controls the complexity of the move. 
In fact, the parameter $k$ also governs expressiveness:
the more slices we allow, 
larger the set of reorderings one can obtain from $\sigma$.
In \propref{k-slice-gradation}, we show a stronger version of 
this claim: $k$-slice reorderings are strictly more permissive
than $m$-slice reorderings when $k > m$. 
In other words, the increase in the parameter strictly increases 
the expressiveness of this parametric predictor, i.e., 
a predictor induced by our notion of reorderings meets goal \squaresmall{G2}.

The extreme values of the parameter are also noteworthy.
Indeed, for the smallest value $k=0$ for this parameter,
the only run that one can obtain from $\sigma$ is $\sigma$ itself,
and thus this class of reorderings coincides with the identity equivalence relation.
The more interesting extreme is the limit where we do not bound the
number of slices; we denote the reordering relation thus induced as $\mslice$.
We show that, in fact, $\mslice$ coincides with $\rfeq$, 
i.e., the expressive power of the predictor based on sliced reorderings
reaches that of $\rfeq$ in the limit, and thus it meets goal \squaresmall{G3}.


The key result in this work is that the 
predictive monitoring modulo $k$-slice reorderings can be performed 
in constant space in a streaming fashion, against arbitrary 
regular specifications (\thmref{regularity-of-forward-closure});
thus meeting goal \squaresmall{G1}. 
We establish this by showing that there is an automaton that accepts
the set of $k$-sliced reorderings of executions in the regular specification.

While $k$-slice reorderings offer a suitable way to parametrize
reads-from equivalence, in this work we also ask 
whether {\em stacking} $k$ slices 
in the way we define this class of reorderings, as a one-shot move,
is necessary for achieving the desired properties.
In particular, we investigate this question by considering the 
simplest $k$-slice reordering move, that is, 
when we set $k = 1$, as an atomic move. 
We call this simply a {\em slice reordering} and denote it as $\slice{}$.
Now, similar to how {\em swap} moves (commuting consecutive independent events)
can be combined in a sequence (a la trace equivalence),
slice reorderings can also be combined in sequence, 
giving us an alternative predictor defined as the 
transitive closure of these individual moves, denoted $\sliceq{}$. 
We give due consideration to this alternative predictor $\sliceq{}$ and
study its properties.
We show that $\sliceq{}$ is strictly more expressive than trace equivalence 
(\thmref{maz-subsumption}) and previous proposals in the literature 
for predictors strictly better than trace equivalence such as 
grains and scattered grains (\thmref{grain-subsumption}),
but comes with two key downsides.
First, $\sliceq{}$
is strictly less expressive than $\rfeq$-equivalence (\thmref{exp-slicestar}).
Second, predictive monitoring up to $\sliceq{}$ 
of slice reorderings admits a linear space lower bound 
for the simple specification of causal concurrency (\thmref{linspace-hardness-slice-star}). 
In other words, this alternative would neither meet goal 
\squaresmall{G2} nor goal \squaresmall{G3}.
This demonstrates why {\em stacking} slices in the manner of $k$-slice 
reorderings proposes is a carefully chosen viable solution in this space.

\noindent
{\bf Organization.}
After setting up notation and recalling prior results in this space in \secref{prelim},
we set out to first study the simpler version \emph{sliced reorderings}
and its transitive and symmetric closure in \secref{slice},
and compare these proposals with trace equivalence and its variants in~\secref{comparison}.
In \secref{stacking-slices}, we then undertake a thorough investigation of the
full proposal of $k$-slice reorderings, examining its expressiveness and computational questions around it.
In \secref{problem} and \secref{algo} we study the predictive monitoring
problem in the context of these reorderings, establishing monitorability
as well as tight lower bounds.
We discuss closely related work in \secref{related} and conclude in \secref{conclusion}.
To keep the presentation concise, some proofs
have been delegated to the appendix.


\section{Preliminaries}
\seclabel{prelim}

\subsection{Modeling executions of programs}

\myparagraph{Concurrent program runs and events}
In our work, we follow the tradition of modeling concurrent program executions
or \emph{runs} as sequences of events performed by different threads.
Each event is a tuple of the form\footnote{As such, each event also has a 
unique identifier $id$ and is more accurately represented as 
$e = [id, lab]$ with label $lab = \ev{t, op}$. 
We will omit this identifier in favor of conciseness of presentation.} 
$e = \ev{t, op(x)}$,
where $t = \ThreadOf{e} \in \threads$ is the identifier of the thread that performs
$e$, $op = \OpOf{e} \in \set{\wt, \rd}$ 
describes the (write or read) operation that was performed
at this event and
the subject of the operation is the memory location
 $x = \MemOf{e} \in \mems$.
For a run $\tr = e_1 \cdot e_2 \cdot \ldots\cdot e_n$, we use
$\events{\tr} = \set{e_1, \ldots, e_n}$ to represent the set of events of $\tr$.
Given the language theoretic treatment of dynamic analysis problems in our work,
it will be convenient to clearly demarcate the \emph{alphabet}
of runs as the set $\labs = \setpred{\ev{t, op(x)}}{t \in \threads, op \in \set{\wt, \rd}, x \in \mems}$.
As with prior works on dynamic analyses of concurrent programs~\cite{Smaragdakis12,Kini17,Mathur18,Mathur21,Pavlogiannis2020,OSR2024,FarzanMathur2024},
unless otherwise stated, 
we will assume that the size of this alphabet is constant.
Thus the number 
of events in executions will be used to determine the size of the input
for the complexity-theoretic and algorithmic treatment of dynamic analysis problems
we undertake in this work.

\myparagraph{Program order and reads-from mapping}
For our presentation, it will be helpful to denote some
semantic relations and functions.
First, for a run $\tr = e_1 \cdot e_2 \cdots e_n$,
we will use $\trord{\tr}$ to denote the
unique total order on $\events{\tr}$ such that for every $i<n$, $e_i \trord{\tr} e_{i+1}$.
Next, we use $\po{\tr}$ to denote the \emph{program order} of $\tr$,
which is defined to be the smallest partial order on $\events{\tr}$ such that
whenever we have $e \trord{\tr} e'$ and $\ThreadOf{e} = \ThreadOf{e'}$,
then we require that $(e, e') \in \po{\tr}$.
Finally, the \emph{reads-from} mapping $\rf{\tr}$ 
maps read events of $\tr$ to their corresponding write events in $\tr$.
Formally, let $\reads{\tr}, \writes{\tr} \subseteq \events{\tr}$
be the set of read
and write events of $\tr$.
Then, the reads-from of $\tr$ is a partial
mapping $\rf{\tr}: \reads{\tr} \partialto \writes{\tr}$ such that
for each $e_r \in \reads{\tr}$, the write $e_w = \rf{\tr}(e_r)$, if one exists, satisfies:
(1) $\MemOf{e_r} = \MemOf{e_w}$ (say $x$),
(2) $e_w \trord{\tr} e_r$, and 
(3) there is no event $e'_w \neq e_w \in \events{\tr}$ such that $\OpOf{e'_w} = \wt$, $\MemOf{e'_w} = x$ and $e_w \trord{\tr} e'_w \trord{\tr} e_r$.
Further, if $\rf{\tr}(e_r)$ is not defined, then we require that 
there is no write event $e_w$ such that $\MemOf{e_r} = \MemOf{e_w}$ and $e_w \trord{\tr} e_r$.

\subsection{Reorderings and equivalences on executions}
\seclabel{prelim-equivalence}

Program analysis techniques such as those that rely on
enumerating executions as in 
partial order reduction based model checking~\cite{Flanagan2005,Abdulla2019,SourceSetsDPORAbdulla2017,Agarwal2021,Kokologiannakis2022}, fuzz testing~\cite{RFFWolff2024}, 
randomized testing~\cite{sen2007effective,pos,tapct} 
as well as those that only infer the presence of bugs
from single executions, as in predictive analyses~\cite{Said11,huang2015gpredict6,Huang14,Tunc2023,Mathur21}
crucially leverage equivalences, and more generally, reorderings
on concurrent program executions to effectively reduce the search space
of program interleavings.
In such applications, one typically works with a
reordering relation $R \subseteq \labs^* \times \labs^*$
such that observations made on a run $\tr$ also generalize
to all other runs $\rho$ for which $(\tr, \rho) \in R$.
Such a generalization is possible if $R$ is \emph{sound},
as we discuss next.

\myparagraph{Soundness of reordering relations}
A reordering $R$ is said to be sound if, intuitively,
whenever $(\tr, \rho) \in R$, then $\rho$ can be generated by
every concurrent program that can generate $\tr$~\cite{serbanuta2013maximal}.
This can be ensured by, in turn, ensuring that 
$\rho$ preserves the control and data flow of the underlying program that 
$\tr$ was obtained from (no matter what the program is).
A direct consequence of using a sound equivalence relation $R$ 
in the context of partial-order reduction style model checking or other forms of exploration based
testing is that it suffices to explore
only one execution per equivalence class.
Dually, a predictive testing technique that
observes a single execution $\tr$ but reasons about 
the entire set $\setpred{\rho}{(\tr, \rho) \in R}$ of 
$R$-reorderings of $\tr$ is, by design, sound (i.e., does not report false positives)
when $R$ is sound.
In the predictive analysis literature, the notion of \emph{correct reorderings},
proposed by Smaragdakis et al~\cite{Smaragdakis12} has been widely adopted, and is known to be the
\emph{largest} sound relation on runs.
In the model checking literature, its analogue \emph{reads-from} equivalence,
denoted $\rfeq$, has emerged as a popular choice of equivalence, and is known to be the
\emph{largest sound} equivalence on runs~\cite{FarzanMathur2024}.
In general, reorderings can relate runs of different lengths,
as is the case with correct reorderings~\cite{Smaragdakis12}, allowing for
enhanced coverage through prefix reasoning in the context of predictive 
analysis~\cite{AngMathurPrefixes2024}, we will restrict our
focus on \emph{length-preserving} 
reorderings --- a reordering relation $R$ is length-preserving
if $(\tr, \rho) \in R \implies |\tr| = |\rho|$ ---
for a cleaner presentation and for wider applicability
such as in model checking applications that rely on this restriction.
All the results of this paper can be straightforwardly generalized
to the setting of non-length-preserving reordering relations.
With the restriction of length-preserving reordering relations, for the purpose of
our work, soundness can be defined purely
in terms of reads-from equivalence $\rfeq$ (which is length preserving), 
whose precise definition we present soon.

\begin{definition}[Soundness of a reordering relation]
\deflabel{sound-reordering}
A length-preserving reordering relation $R \subseteq \labs^* \times \labs^*$ 
is sound if $R \subseteq \rfeq$.
\end{definition}

In the following, we survey
different notions of reorderings that have emerged in the literature.

\myparagraph{Reads-from equivalence}
The reads-from equivalence $\rfeq$ is the smallest equivalence on $\labs^*$
such that for two runs $\tr$ and $\rho$, if $\po{\tr} = \po{\rho}$
and $\rf{\tr} = \rf{\rho}$, then $\tr \rfeq{} \rho$.
In our presentation, reads-from equivalence is sound by definition, and is in fact
the largest sound reordering relation.


\myparagraph{Trace equivalence}
Trace theory~\cite{Mazurkiewicz87} provides a classic example of a commutativity
based equivalence for systematically defining equivalences over strings.
Here, one fixes an irreflexive and symmetric independence relation $\indrel \subseteq \labs \times \labs$
to demarcate when two events must be considered independent or \emph{commuting},
and then deems two runs (i.e., strings over $\labs$)
equivalent if one can be obtained from another through repeated \emph{swaps} of neighbouring
independent events.
Formally, given a choice of independence relation $\indrel$,
the trace equivalence $\mazeq$ induced by $\indrel$ is
defined to be the smallest equivalence for which, whenever
$(e,f) \in \indrel$, then the runs $\tr = \tr_1 \cdot e \cdot f \cdot \tr_2$ and $\rho = \tr_1 \cdot f \cdot e  \cdot\tr_2$ are equivalent, i,e, $\tr \mazeq \rho$.
In our context, we will consider the usual independence relation 
of non-conflicting events, i.e.,
$\indrel = \setpred{(\ev{t_1, op_1(x_1)}, \ev{t_2, op_2(x_2)})}{t_1 \neq t_2 \land (x_1 = x_2 \implies op_1 = op_2 = \rd)}$, since it
induces the largest sound trace equivalence~\cite{FarzanMathur2024}.


\myparagraph{Grain and scattered grain commutativity}
Recently, Farzan and Mathur~\cite{FarzanMathur2024} proposed \emph{grain and scattered grain}
based reasoning to soundly approximate
reasoning based on reads-from equivalence, while offering higher predictive power
than trace equivalence, without compromising on the algorithmic benefits of trace equivalence based reasoning.
Formally, an execution $\rho$ is a grain-reordering of execution $\tr$
if there is a partition of $\tr$ into contiguous subsequences (called \emph{grains}) $\tr = g_1 \cdot g_2 \cdots g_k$ such that $\rho$ can be obtained from $\tr$ by repeated swaps of the grains $G = \set{g_1, \ldots, g_k}$
under the grain independence relation $\indrel_{G}$ which marks two grains $g, g'$ independent 
if (a) they do not share a thread and (b) are complete with respect to any common memory location $x$,
i.e., for any two accesses $(e, e') \in \rf{\tr}$ on $x$, either $\set{e, e'} \subseteq g$
(resp. $\set{e, e'} \subseteq g'$) or $\set{e, e'} \cap g = \emptyset$ (resp. $\set{e, e'} \cap g' = \emptyset$).
In summary, $\indrel_G$ treats a pair of grains as independent exactly when they can be swapped
in any surrounding context.
We use the notation $\rho \graineq \tr$ to denote that $\rho$ is a grain-reordering of $\tr$;
$\graineq$ is a symmetric reflexive relation.
The notion of scattered grain reorderings 
is a generalization of grain equivalence, where
grains are no longer required to be contiguous.
In the original work of Farzan and Mathur~\cite{FarzanMathur2024}, scattered grains were introduced
in the context of answering a causal concurrency question (whether two given events can be flipped);
a formal definition of the induced reordering relation was not provided, and we skip it here
since all our observations about grain reorderings carry over to scattered grains as well.


\section{Sliced Reordering}
\seclabel{slice}


In this section, we introduce a sound reordering relation 
that relates two executions when the second can be obtained from the first 
through a \emph{slicing operation} and study its properties.

\begin{definition}[Sliced Reordering]
\deflabel{sliced-reordering}
For a pair of concurrent program runs $\tr$ and $\rho$, we say that $\rho$ is a sliced reordering of $\tr$,
denoted $\sliceto{\tr}{\rho}$, if
$\tr \rfeq \rho$, and further,
there are disjoint subsequences $\tr_1$ and $\tr_2$
of $\tr$  
such that $\rho = \tr_1 \cdot \tr_2$.
\end{definition}



\noindent
Note that soundness is directly baked into the definition:

\begin{restatable}{proposition}{SliceSoundness}[Soundness of sliced reorderings]
\proplabel{sound-slices}
$\slice \subseteq \rfeq$
\end{restatable}

\begin{figure}[t]
\begin{center}
\includegraphics[scale=0.8]{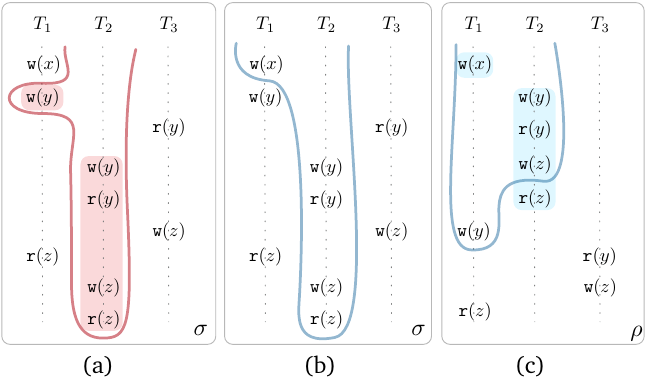}\vspace{-10pt}
\caption{The wrong (a) and the right (b) slice choice for $\sliceto{\tr}{\rho}$}
\figlabel{single-slice}
\end{center}\vspace{-10pt}
\end{figure}

Consider the example in \figref{single-slice}. $\rho$ is a sliced reordering of $\sigma$. 
Interestingly, the partitions of $\sigma$ can be pictorially depicted using
a single curve denoting how to \emph{slice} $\tr$, as illustrated by the (blue) curve. 
Also observe that, $\rho$ and $\tr$ have the same program order,
and every read event in $\rho$ observes the same last write event
as in $\tr$. 
Consider now the curve in \figref{single-slice}(a) marking a slicing of this execution, demarcating the reordering obtained by linearizing and then concatenating the two partitions 
(shaded followed by unshaded). Such a reordering would 
yield a different program order (because it flips
the relative order of events in $T_1$),
and thus would not
not be reads-from equivalent to $\tr$. 
Consequently, it is straightforward to observe that any curve marking a correct sliced reordering should intersect each thread at most once. 

We have been careful in not calling the relation $\slice$ an equivalence, because it indeed is not one!

\begin{restatable}{proposition}{SliceReflexiveTransitiveSymmetric}
\proplabel{slice-refl-symm-trans}
$\slice$ is reflexive but neither symmetric nor transitive.
\end{restatable}

\begin{proof}
First, $\slice$ is trivially reflexive --- $\tr$ is a sliced reordering of itself,
and this can be witnessed by the subsequences $\tr_1 = \tr$ (itself) 
and $\tr_2 = \epsilon$ (empty subsequence).
Now, let us understand why $\slice$ is not symmetric.
Consider again the runs $\tr$ and $\rho$ in
\figref{single-slice}.
We previously observed that $\sliceto{\tr}{\rho}$.
We will now argue that, in turn, $\rho$ is not a sliced reordering of $\tr$.
Assume on the contrary that indeed there are subsequences $\rho_1$ and $\rho_2$ of $\rho$
such that $\rho_1 \cap \rho_2 = \emptyset$ and $\tr = \rho_1 \cdot \rho_2$.
Consider the second and the sixth events of $\rho$:
$e_6 = \ev{T_1, \wt(y)}$ and $e_2 = \ev{T_2, \wt(y)}$.
Since their relative order gets flipped across $\rho$ and $\tr$
(i.e., $e_2 \trord{\rho} e_6$, but $e_6 \trord{\tr} e_2$),
we must have $e_6 \in \rho_1$ and $e_2 \in \rho_2$.
For the same reason, $e_7 = \ev{T_3, \rd(y)} \in \rho_1$, $e_8 = \ev{T_3, \wt(z)} \in \rho_1$
and $e_9 = \ev{T_1, \wt(z)} \in \rho_1$.
Also of course, $e_1 = \ev{T_1, \wt(x)} \in \rho_1$, 
and all other events must be in $\rho_2$.
But then, the events $e_1, e_6, e_7, e_8, e_9$ must appear contiguously in $\tr$
which is a contradiction.

Let us now argue why $\slice$ is not transitive.
Consider the pair of executions below, where the first execution (left) is $\rho$ from \figref{single-slice}.
\begin{center}
\includegraphics[scale=0.8]{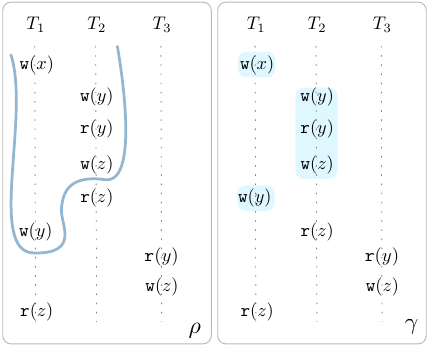}
\end{center}
Observe that the second execution $\gamma$ is a sliced reordering of $\rho$ (i.e., $\sliceto{\rho}{\gamma}$) as witnessed by the blue curve on $\rho$. Hence, we have $\sliceto{\tr}{\rho}$ and $\sliceto{\rho}{\gamma}$.
We argue that $\gamma$ is not a sliced reordering of $\tr$.
Assume on the contrary that there are subsequences $\tr'_1$ and $\tr'_2$ of $\tr$
such that $\tr'_1 \cap \tr'_2 = \emptyset$, $\tr'_1 \cdot \tr'_2 = \gamma$.
The second and eighth events of $\tr$, 
namely $f_2 = \ev{T_1, \wt(y)}$ and
$f_8 = \ev{T_2, \wt(z)}$, must belong to different subsequences since their order gets flipped,
and thus, we must have $f_2 \in \tr'_2$ and $f_8 \in \tr'_1$.
Now, since $f_3 = \ev{T_3, \rd(y)}$, $f_6 = \ev{T_3, \wt(z)}$
and $f_9 = \ev{T_2, \rd(z)}$ appear later than $f_2$ in $\gamma$, they must also all belong to $\tr'_2$.
But then, $f_2, f_3, f_8, f_9$ must appear in the exact same order in $\gamma$, which is a contradiction.

\end{proof}

\subsection{Sequencing Sliced Reorderings}

\noindent Even though $\slice$ is not transitive, we can consider
its reflexive transitive closure $\sliceq:$

\begin{definition}[Repeated sliced reordering]
\deflabel{repeated-sliced-reordering}
For concurrent program runs $\tr$ and $\rho$, we say that $\rho$ is a repeated sliced reordering of $\tr$,
denoted $\sliceqto{\tr}{\rho}$, if there exist $\gamma_1, \gamma_2, \ldots, \gamma_k$
such that 
\[\tr = \gamma_1 \sliceto{}{} \gamma_2  \sliceto{}{} \dots \sliceto{}{} \gamma_k = \rho\]
\end{definition}

$\sliceq$ is still not an equivalence relation, 
because it remains non-symmetric. 
We end this section by making two key observations about $\sliceq$,
in the broader context of evaluating whether $\sliceq$ is fit
to be a a predictor and if it can be used as an alternative to existing
predictors such as reads-from equivalence or other commutativity based equivalences.
First, it is strictly contained in $\rfeq$, 
even if we consider the equivalence $(\slice + \slice^{-1})^*$
obtained by closing it under symmetry. 
This is a vital point justifying the central contribution of this paper, 
which is the parametric definition presented in~\secref{stacking-slices}.

\begin{restatable}{theorem}{SliceStarExpressivity}[Expressivity of $\sliceq$]
\thmlabel{exp-slicestar}
$\sliceq$, closed under symmetry is strictly smaller than $\rfeq$. 
\end{restatable}

The visual proof is in \figref{lim}. The runs in (a) and (b) are {\sf rf}-equivalent, but there is not a single slice reordering move that is enabled in either (a) or (b), and hence the two runs sit in an equivalence class of size one if considered under the symmetry closure of $\sliceq$. It is tedious to enumerate all possibilities and the reason why they are invalid, but the high-level observation is that any slice reordering move would either break the reads-from relation between $\wt(y)$ and $\rd(y)$ in thread $T_1$, or $\wt(x)$ and $\rd(x)$ in thread $T_2$, or the cross thread pairs $\wt(z)$/$\rd(z)$ or $\wt(t)$/$\rd(t)$. 

\begin{figure}[t]
\begin{center}
\includegraphics[scale=0.4]{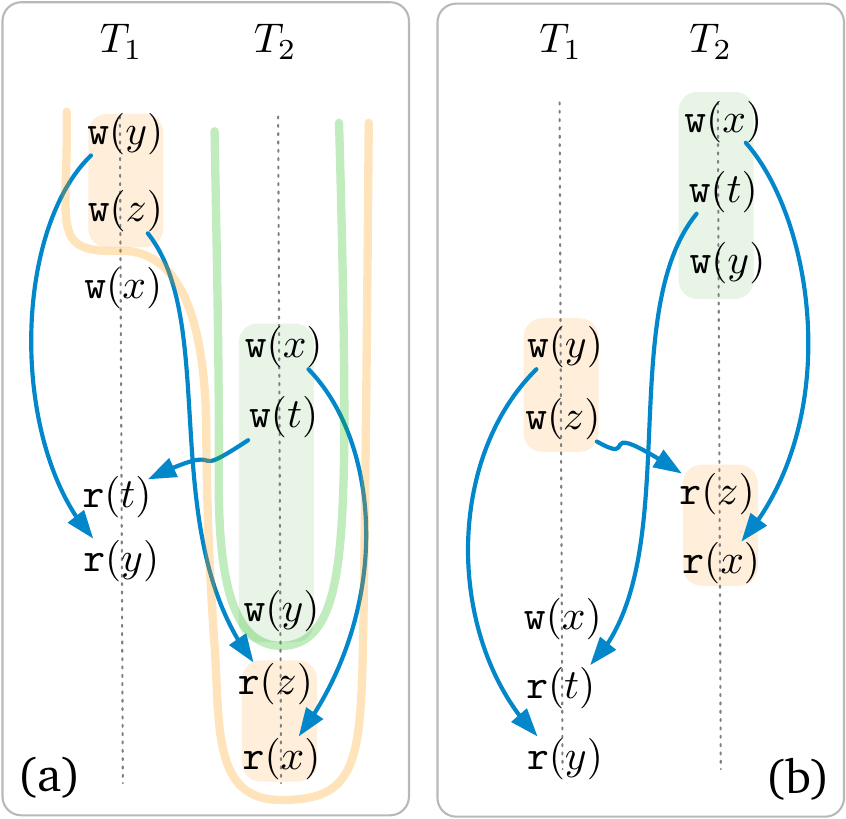}\vspace{-7pt}
\caption{$\sliceq$ is strictly weaker than $\rfeq$.\figlabel{lim}}\vspace{-10pt}
\end{center}
\end{figure}

Second, despite having limited expressivity, $\sliceq$ shares the disadvantage of {\sf rf}-equivalence in having a {\em linear space lower bound} when computing the closure of a regular language up to it. 
In particular, an analogue of the hardness result in \cite{FarzanMathur2024} for $\rfeq$ holds:

\begin{restatable}{theorem}{LinSpaceHardnessSliceStar}[Linear Space Hardness of Closure up to $\sliceq$]
\thmlabel{linspace-hardness-slice-star}
Given a concurrent program $\sigma$ and a pair of events $e$ and $f$ in it, any one-pass algorithm that checks if the order of $e$ and $f$ can be flipped in a run $\rho$ such that $\sliceqto{\sigma}{\rho}$ requires linear space.
Further, the time $T(n)$ and space $S(n)$ usage of any multi-pass algorithm for solving this problem must satisfy $S(n) \cdot T(n) \in \Omega(n^2)$, where $|\sigma| = n$.
\end{restatable}


\section{Comparison with Existing Sound Predictors}
\seclabel{comparison}

In this section, we compare the expressive power of $\slice$ (and consequently $\sliceq$) against commutativity-based equivalences $\mazeq$, and the notion of grains and scattered grains introduced in \cite{FarzanMathur2024}.

\subsection{Trace Equivalence}\label{sec:maz-comp}

Let us first observe that a single swap of two adjacent commuting events can be simulated using a single step of sliced reordering. We have a run $\rho$ such that 
$\rho = \sigma e f \sigma' \equiv_\maz \sigma f e \sigma'$
\begin{wrapfigure}[8]{r}{0.3\textwidth}
\vspace{-10pt}
\includegraphics[scale=1.0]{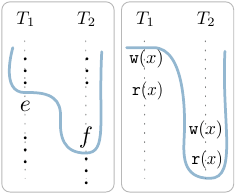}
\vspace{-20pt}
\end{wrapfigure}
If we let $\sigma_1 = \sigma f$ and $\sigma_2 = e \sigma'$ in Definition \ref{def:sliced-reordering}, 
then this swap is simulated through a sliced reordering (as illustrated on the right, but nominally only for two threads). 
The converse is, however, not true. First, a single slice reordering can simultaneously swap 
many pairs of events. Therefore, a slice reordering may require many swaps to simulate. Second, it can reorder pairs of events that can never be reordered under trace equivalence, for instance, swapping the order of the sequence $\wt(x)\rd(x)$ in thread $T_1$ against the same sequence in thread $T_2$ as illustrated on the right. 
Thus, we have:
\begin{restatable}{theorem}{SliceStarVsMazExpressivity}
\thmlabel{maz-subsumption}
For all $\sigma$ and $\rho$ such that $\sigma \mazeq \rho$, 
we have $\sliceqto{\sigma}{\rho}$. On the other hand, there exist $\sigma', \rho'$ such that $\sliceqto{\sigma'}{\rho'}$, but $\sigma' {\not \equiv}_{\maz} \rho'$. Hence, $\mazeq \subsetneq \sliceq$.  
\end{restatable}

In effect, $\sliceq$ and $\mazeq$, both have the flavour of converting one run to the other through a sequence of atomic steps: swaps in the case of $\mazeq$ and sliced reorderings in the case of $\sliceq$. The next natural question is, for a pair of runs $\sigma$ and $\rho$ that are related by both relations, is there a difference in the worst-case number of atomic steps it takes to go from $\sigma$ to $\rho$?

Here, an analogy with two classic sorting methods holds the key to the answer. One can think of swaps as the unit operation in {\em bubble sort} and the slices as being able to simulate an {\em insertion} from {\em insertion sort}. It is well-understood that sorting a reverse-sorted sequence requires {\em quadratically many} swaps in bubble sort. Hence, the same intuition is true for $\mazeq$. In contrast, one needs at most linearly many {\em insertions} to sort any list. Note, however, that in the context of concurrent program runs, once we have more events than the number of threads, we cannot have the entire set of events to be reverse-ordered to create this extreme adversarial situation. Therefore, making a precise argument on the lower bound of the number of required swaps/insertions requires a bit more care, captured by the following theorem:

\begin{restatable}{theorem}{SliceStarVsMazExpressivityMetric}
\thmlabel{maz-comp}
If $\sigma \mazeq \rho$, then the number of steps of sliced reordering that are required to go from $\sigma$ to $\rho$ is always less than or equal to the number of swaps of adjacent commutative actions. Moreover, there exists $\sigma$ and $\rho$ such that $\sigma \equiv_{\mathcal{M}} \rho$ and it takes $O(|\sigma|)$ number of slice reorderings to convert $\sigma$ to $\rho$, and $O(|\sigma|^2)$ number of swaps. 
\end{restatable}

\begin{proof}
Consider the execution in \figref{swap-lower-bound} (left) which is
the sequential composition of $k$ threads, each of which executes $n$ $\rd(x)$ events. 
We have $nk$ events in total. Since all $\rd(x)$'s commute under Mazurkiewicz commutativity, it is straightforward to see that the runs in the figure below, parts (a) and (b), are equivalent, up to both $\mazeq$ and (consequently by Theorem \ref{thm:maz-subsumption}) $\sliceq$. 
Now, imagine that we want to use a sequence of atomic slices to reorganize the run 
in (a) into the one in (b). 
We must start by taking the slice marked as $s_1$ to place the first event of $T_k$ in place. Then, we continue with $s_2$ to place the first event of $T_{k-1}$ in place. After $k-1$ slices $s_1, \dots, s_{k-1}$, we have the first round of the round-robin schedule of run (b) in place. We must continue with this process for another $n - 1$ rows, using $k-1$ slices in each row, until we are done. Hence, we need $n(k-1)$ slices to reorganize the $nk$ events.

\begin{figure}[t]
\begin{center}
\includegraphics[scale=0.8]{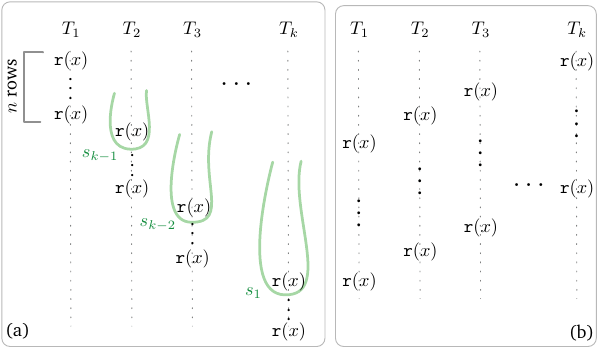}
\end{center}
\caption{Illustrative example demonstrating the number of swaps required to witness trace equivalence. 
}
\figlabel{swap-lower-bound}
\end{figure}


Note that one must use more swaps to go from (a) to (b). The first $\rd(x)$ of thread $T_k$ has to be swapped against $n(k-1)$ events that come before it in (a). The first $\rd(x)$ of thread $T_{k-1}$ has to be swapped against $n(k-2)$ events. Following the same line of reasoning, we have
\begin{align*}
\text{number of required swaps} &= n(k - 1) + n(k - 2) + \dots + n \\
								& + (n - 1) (k - 1) + (n - 1) (k - 2) + \dots + n - 1 \\
                                & + \dots \\
                                & + (k - 1) + (k - 2) + \dots + 1 \\
                                & = n (n+1) k (k - 1)/4
\end{align*}
which yields $(n+1)k/4$ times more swaps than slices needed to convert the above run (a) to run (b). Recall that a single swap can always be simulated by a single slice. So, one should never need more slices than swaps to go to a run equivalent up to $\mazeq$.

\end{proof}


\subsection{Grains and Scattered Grains Commutativity}\label{sec:grains-comp}

\begin{wrapfigure}[8]{r}{0.13\textwidth}
\vspace{-10pt}
\includegraphics[scale=0.8]{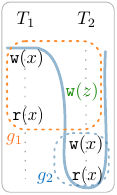}\vspace{-5pt}
\end{wrapfigure}Sliced reordering can simulate the swap of two consecutive commuting grains precisely the same way that it simulates the swap of two adjacent commuting events. Simply replace $e$ and $f$ with two grains, and the argument is the same. To observe that a grain swap cannot simulate a sliced reordering, consider the figure on the right. The $\wt(z)$ must be included in $g_1$ to satisfy the grain {\em contiguity} requirement, but then, as a result, $g_1$ and $g_2$ no longer commute since they share a thread. The marked slice, however, reorders the content of thread $T_2$ against that of thread $T_1$.

{\em Scattered grains} were proposed in~\cite{FarzanMathur2024} as a workaround for the above outlined contiguity problem, and can argue for the validity of the reordering in the above example. Unlike slice reordering, grains, and trace equivalence, scattered grains do not propose a step-by-step transformation of one run to another equivalent run. Nevertheless, we can argue that a single sliced reordering can transform a run in a way that scattered grains cannot.


\begin{figure}[t]
\begin{center}
\includegraphics[scale=0.75]{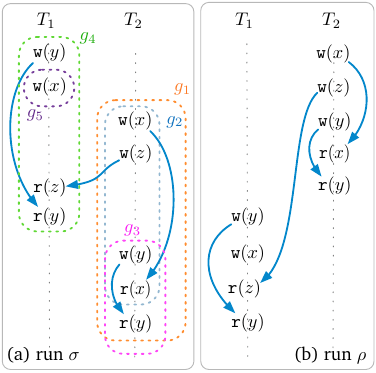}
\vspace{-10pt}
\caption{Slice Reordering vs Scattered Grains. 
\figlabel{glimit}}
\vspace{-10pt}
\end{center}
\end{figure}
Consider~\figref{glimit}. We have $\sigma \equiv_{\sf rf} \rho$: the arrows, connecting the matching write to each read, are preserved, and the program order has not changed. Observe further that this equivalence can be observed using a single step of sliced reordering: the slice that includes everything in $T_2$ in the first partition and everything in $T_1$ in the second partition does the job. Now, let us argue that $\sigma$ cannot be transformed to $\rho$ through any choices of scattered grains. All choices of scattered grains, that potentially have more commutativity than the set of events inside them, are marked in $\sigma$. $g_2$ and $g_5$ are {\em complete} wrt $x$, i.e. they include all the read events that read from the included $\wt(x)$. $g_4$ is complete wrt $y$. $g_1$ and  $g_4$ are complete wrt both $x$ and $y$. $g_3$ and $g_5$ are contiguous, and the rest of the grains are scattered. However, no pairs of (scattered) grains commute.


\begin{restatable}{remark}{SliceStarVsScatteredGrains}
$\rho$ is not equivalent to $\sigma$ (in fact, to anything other than itself!) using scattered grains.
\end{restatable}

\begin{proof}
No (scattered) pairs of grains commute:
\begin{itemize}
\item $(g_1,g_4)$ {\em overlap}. They do not commute due to the edge from $\wt(z)$ to $\rd(z)$. For them not to be {\em entangled}, the only remaining possibility is to use only Mazurkiewicz commutativity to argue that $\rho$ is equivalent to $\sigma$, but due to each having a $\wt(x)$, this cannot be done. 
\item $(g_2,g_4)$ are {\em entangled} for the exact same reason. 
\item $(g_3,g_4)$ do not commute because $g_3$ contains $\rd(x)$ without including its matching $\wt(x)$. 
\item $(g_3, g_5)$ do not commute for the same exact reason.
\item $(g_1,g_5)$ technically commute, but this fact can never be used because if $g_5$ forms a grain, the sequence of actions $\rd(z) \rd(y)$ at the end of thread $T_1$ is entangled in grain $g_1$ which forces $g_1$ to be strictly ordered after $g_5$.
\item $(g_2, g_5)$ follow the same exact pattern as above.
\end{itemize}
\end{proof}

This example highlights the key difference between sliced reorderings and grains. The reason behind {\em commutativity} of grains must be simple and syntactic, which is precisely why we cannot argue $\sigma$ is equivalent to $g_1 g_4 = \rho$.   However, sliced reorderings can argue for this using a semantic style of reasoning: all the reads remain matched with the same writes if we execute $g_1$ first. With scattered grains, the argument for the preservation of this matching is through the simple syntactic means that all the matched pairs remain inside the same grain and move as one unit together. Therefore, scattered grains are fundamentally incapable of arguing for the preservation of the $(\wt(z),\rd(z))$ matching, because they cannot be put inside a single (scattered) grain together without blocking everything else from moving.


\begin{wrapfigure}[8]{r}{0.33\textwidth}
\vspace{-15pt}
\includegraphics[scale=0.85]{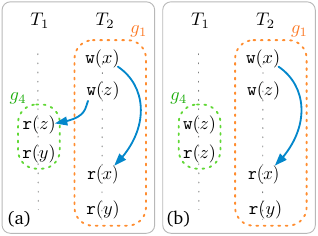}
\vspace{-15pt}
\end{wrapfigure}
To dig deeper into this, consider the example runs on the right. In (a), we can argue that the grain graph only contains an edge from $g_1$ to $g_4$ and therefore this run is equivalent, up to scattered grains, to $g_1 g_4$. The reasoning needed here is restricted to commutativity of read actions only and falls under trace equivalence-based reasoning. This means that, limited to these two grains only, our power of reasoning is as good as what trace equivalence offers. Now consider a slight variation in (b). Here, we can make the same argument that the run is equivalent to $g_4g_1$, but this time, we rely on the full commutativity of the two grains. This reasoning cannot be done by trace equivalence. Thus, we have:

\begin{restatable}{theorem}{SliceStarVsGrainsExpressivity}
\thmlabel{grain-subsumption}
For all $\sigma$ and $\rho$, $\sigma \graineq \rho \implies \sliceqto{\sigma}{\rho}$. On the other hand, there exist $\sigma', \rho'$ such that $\sliceto{\sigma'}{\rho'}$, but $\sigma' \not \graineq \rho'$. Hence, $\graineq \subsetneq \sliceq$. The same result holds for scattered grains \footnote{To be very precise, in~\cite{FarzanMathur2024}, an equivalence relation based on scattered grains is not defined. Hence, our result here is not making use of any such formal definition. Our argument uses a construction in which no choice of scattered grains can be used to argue for reordering any part of the program run, and hence we can make the claim without committing to any speculative definition.}. 
\end{restatable}


%
%
%
%
%


\section{Stacking Slices}
\seclabel{stacking-slices}

In Section \ref{sec:slice}, we introduced the transitive closure of sliced reorderings, as a way of composing multiple sliced reordering transformations, and formally argued that this notion is not expressive enough to simulate $\rfeq$. Here, we recall~\cite{FarzanMathur2024} that states a similar negative result for grains and scattered grains, and use the example to motivate a different way of composing individual sliced reordering moves that overcomes this limitation of expressivity shared by all three notions. 

Let us revisit the example in \figref{lim}. As argued after \thmref{exp-slicestar}, the runs illustrated in (a) and (b) are equivalent up to $\rfeq$ but not so under the symmetric closure of $\sliceq$. In particular, for example, the two slices marked by green and orange curves would transform the run in (a) to the one in (b), but they both correspond to invalid slice reorderings; the green one, for instance, would break the reads-from relation between $\wt(x)$ and $\rd(x)$ in thread $T_2$, marked by a blue arrow. This could be remedied if one could simultaneously reorder the events in the suffix to fix this problem, that is, apply the reordering marked by the orange curve at the same time. 

Unlike the classical way of combining atomic (commutativity) moves, where a program run is transformed step-by-step through a sequence of atomic moves and a sequence of intermediate equivalent runs, our proposal for many-slice reorderings {\em stacks} the moves together in one {\em compound} move that can do more than a sequence of such moves. 
In Figure \ref{fig:lim}, a compound move consisting of both the slicing suggested by the green curve {\em and} the one suggested by the orange curve {\em simultaneously},  one can transform run (a) to run(b). 


\subsection{$k$-slice Reorderings}
\seclabel{k-slices}

Recall the two aspects of a slice reordering --- to demonstrate
if a run $\tr$ can be reordered to another run $\rho$, one first identifies a
subsequence $\tr_1$ (with the residual subsequence being $\tr_2$), and then
arranges the two subsequences one after the other to get $\rho = \tr_1 \cdot \tr_2$
without changing the relative order of events within each of the subsequences.
Intuitively, $k$-slice reorderings can be obtained by generalizing this construction
to more than one slice.

\begin{definition}[k-sliced  reordering]
\deflabel{k-slice-reordering}
Let $\tr$ and $\rho$ be concurrent program runs, and let $k \in \natsp$.
We say that $\rho$ is a $k$-slice reordering of $\tr$,
denoted $\ksliceto{k}{\tr}{\rho}$, if
$\tr \rfeq \rho$, and further,
there are $k+1$ disjoint subsequences $\tr_1, \tr_2, \ldots, \tr_{k+1}$
of $\tr$ (i.e., for every $1 \leq i\neq j \leq k+1$, $\events{\tr_i} \cap \events{\tr_j} = \emptyset$) 
such that $\rho = \tr_1 \cdot \tr_2 \cdots \tr_{k+1}$.
\end{definition}

\noindent
As with sliced reorderings (\propref{sound-slices}),
soundness is baked into the definition of $\kslice{k}$:

\begin{restatable}{proposition}{kSliceSoundness}[Soundness of $k$-sliced reorderings]
\proplabel{sound-k-slices}
For every $k \in \natsp$, $\kslice{k} \subseteq \rfeq$.
\end{restatable}

To illustrate the reorderings of \defref{k-slice-reordering}, recall \figref{lim}.
We have $\ksliceto{2}{\tr}{\gamma}$, i.e., there is a 
`one-shot 2-slice move' that transforms $\tr$ to $\gamma$.
The green and the orange curves pictorially denote the slices and the 
resulting three subsequences $\tr_1, \tr_2, \tr_3$ that witness this move.
The subsequence $\tr_1$ comprises the events above the innermost 
(green) curve, i.e., the first two events of $T_1$ and the first $3$ events of $T_2$.
The subsequence $\tr_2$ comprises of next two events ($\ev{T_1, \wt(t)}$, $\ev{T_1, \rd(y)}$)
of thread $T_1$ and the next two events ($\ev{T_2, \wt(y)}$, $\ev{T_2, \rd(x)}$) of thread $T_2$.
while $\tr_3$ consists of the remaining events.
It is easy to see that $\gamma$ can be obtained by the concatenation $\tr_1 \cdot \tr_2 \cdot \tr_3$,
and further, as we have already noted, $\tr \rfeq \gamma$.
Thus, $\ksliceto{2}{\tr}{\gamma}$.

\subsection{Properties of $k$-slice Reorderings}

It is easy to observe that the relations $\slice$ and $\kslice{1}$ (i.e., $k=1$) coincide. The discussion around the example in \figref{lim} argued the point that $\kslice{2}$ is strictly more expressive than $\kslice{1}$. 
The most significant feature of $\kslice{k}$ is that this strict increase in expressivity scales with the parameter $k$:


\begin{restatable}{proposition}{kSliceGradedExpressivity}[Graded Expressivity]
\proplabel{k-slice-gradation}
For every $k \in \natsp$, $\kslice{k} \subsetneq \kslice{k+1}$.
\end{restatable}

\begin{figure}[t]
\includegraphics[scale=0.8]{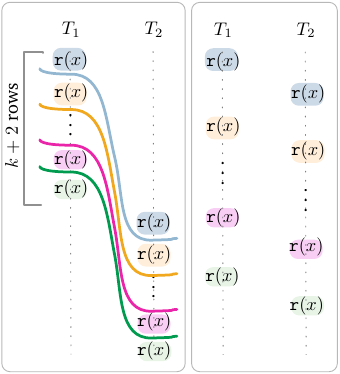}
\caption{$\tr^{\sf seq}$ (left) and $\tr^{\sf int}$ (right) \figlabel{k-slice-gradation}}
\end{figure}

\begin{proof}
The inclusion is straightforward since one can always choose the last partition to be an empty word.
The more interesting part of the statement of \propref{k-slice-gradation} is that
$\kslice{k+1}$ is strictly
more permissive than $\kslice{k}$, for every $k$.
To understand why, consider the sequential run $\tr^{\sf seq}$
and the interleaved run $\tr^{\sf int}$ of \figref{k-slice-gradation}, 
each consisting of $k+2$ events in both threads. Also note that $\tr^{\sf int} \rfeq \tr^{\sf seq}$.
As the figure demonstrates, indeed we have $\tr^{\sf seq} \kslice{k+1} \tr^{\sf int}$.
Now, let us argue that we cannot reorder $\tr^{\sf seq}$ into
$\tr^{\sf int}$ with less than $k+1$ slices.
Suppose on the contrary that $\ksliceto{k}{\tr^{\sf seq}}{\tr^{\sf int}}$.
Consider the subsequences $\tr^{\sf seq}_1, \tr^{\sf seq}_2 \ldots \tr^{\sf seq}_k$
that witnesses this reordering.
It is easy to see that the $i^{th}$ event of $T_2$ must appear
in an earlier subsequence than the $(i+1)^{th}$ event of $T_1$
because they appear in the inverse order in $\tr^{\sf int}$.
This means there must be at least  $k+1$ distinct subsequences, which is a contradiction.
That is, $\tr^{\sf seq} \notkslice{k} \tr^{\sf int}$.
\end{proof}

\propref{k-slice-gradation} implies, in a straightforward manner, that for any $k$, there are a pair of runs $\sigma$ and $\rho$ such that 
$$\sigma \kslice{k} \rho \land \sigma \notkslice{k - 1} \rho.$$ 
Below, we state this corollary and strengthen it by making the claim symmetric:

\begin{restatable}{proposition}{kSliceLowerLimitWidth}[Lower-bound on $m$ for $m$-slice relations]
\proplabel{lower-limit-width}
For any $m$, there exists a pair of runs $\sigma$ and $\rho$ such that $\sigma \kslice{m} \rho$ and $\rho \kslice{m} \sigma$ while $\sigma \notkslice{m - 1}$ and $\rho \notkslice{m - 1} \sigma$
\end{restatable}

\begin{proof}
Recall \figref{k-slice-gradation}. Imagine $\sigma$ is $\sigma^{\sf int}$ from the figure, and $\rho$ is the same style of round-robin execution, except it starts with thread $T_2$ in contrast to $\sigma^{\sf int}$ that starts with thread $T_1$. Following the same line of the argument as before, it is easy to see that one needs a slice for each round of the round-robin execution to reorder the events from threads $T_1$ and $T_2$ that are in the wrong order. As such, one needs $k+1$ slices to go from $\sigma$ to $\rho$ and the same number to go from $\rho$ to $\sigma$. Let $m = k+1$, and the \propref{lower-limit-width} is proved. 
\end{proof}

%
%



$\kslice{k}$ is reflexive because $\slice = \kslice{1}$ is reflexive
(\propref{slice-refl-symm-trans}) and \propref{k-slice-gradation} guarantees monotonicity.
On the other hand, symmetry and transitivity continue to remain absent for
$k$-sliced reorderings, for all values of $k$:

\begin{restatable}{proposition}{kSliceNotEquivalence}[Not An Equivalence Relation]
\proplabel{k-slice-refl-symm-trans}
For every $k \in \natsp$, $\kslice{k}$ is reflexive, but neither symmetric nor transitive.
\end{restatable}

\begin{proof}

The absence of symmetry can be explained through
the two executions $\tr^{\sf int}$ and $\tr^{\sf seq}$ in \figref{k-slice-gradation}.
Recall that $\ksliceto{k}{\tr^{\sf int}}{\tr^{\sf seq}}$ does not hold.
But in the other direction, observe that one can obtain
$\tr^{\sf seq}$ from $\tr^{\sf int}$ using a single slice, i.e.,
two subsequences $\tr^{\sf int}_1$ and $\tr^{\sf int}_2$.
The subsequence $\tr^{\sf int}_1$ consists of exactly all events of thread $T_1$,
while the subsequence $\tr^{\sf int}_2$ consists of exactly all events of thread $T_2$.
It is easy
 to see that $\tr^{\sf seq} = \tr^{\sf int}_1 \cdot \tr^{\sf int}_2$.
This means, 
$\ksliceto{1}{\tr^{\sf int}}{\tr^{\sf seq}}$ and thus 
$\ksliceto{k}{\tr^{\sf int}}{\tr^{\sf seq}}$ (\propref{k-slice-gradation}).
But, the converse is not true. In other words, $\kslice{k}$ is not symmetric.

\begin{figure}[t]
\includegraphics[scale=0.8]{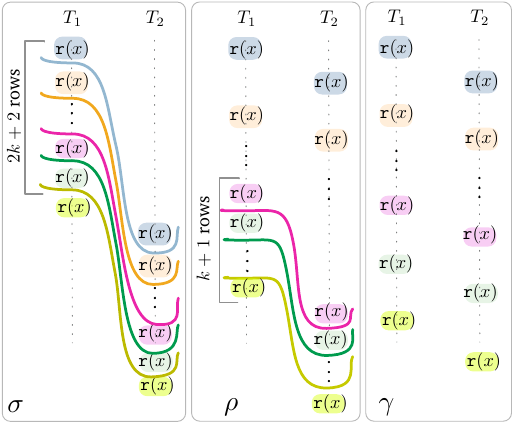}
\caption{$\kslice{k}$ is not transitive}
\figlabel{k-slice-not-transitive}
\end{figure}
The lack of transitivity also follows because of analogous reasoning.
Consider the three executions $\tr, \rho$ and $\gamma$ in \figref{k-slice-not-transitive}, each containing
$2k+2$ events in both threads.
We note that $\ksliceto{k}{\tr}{\rho}$ using an argument similar to the previous example.
Likewise, $\ksliceto{k}{\rho}{\gamma}$ for the same reason.
But, in order to transform $\tr$ to $\gamma$, we do need at least $2k$ slices and thus
$\notksliceto{k}{\tr}{\gamma}$.
\end{proof}

\subsection{Expressive Power of $k$-sliced Reorderings}
\seclabel{many-slices}

We start by briefly stating results analogous to those in \secref{comparison} for $k$-slice reorderings. 

\begin{restatable}{theorem}{kSliceExpressivity}[Expressivity of $k$-slice Reorderings]
\thmlabel{mslice-expressivity}
$k$-slice reorderings have an expressive power that is incomparable wrt Mazurkiewicz commutativity ($\mazeq$), grains and scattered grains. They are strictly weaker than reads-from equivalence ($\rfeq$). 
\end{restatable}

\begin{proof}
We can argue that $\kslice{k}$ for any $k$ is incomparable to $\mazeq$. First, recall that for any $k$, we have  $\slice \subseteq \kslice{k}$ and we already argued in Section \ref{sec:comparison} that $\slice$ can simulate a grain swap that is beyond the power of $\mazeq$. Recall the argument for \propref{lower-limit-width}. Observe that $\sigma \mazeq \rho$ since they only have read operations that fully commute. This shows that $\kslice{k}$ also does not subsume $\mazeq$. But, the latter is for an obvious reason. $\mazeq$ permits an arbitrary sequence of swaps whereas $\kslice{k}$ is a single shot move. As we argue in Section \ref{sec:comparison}, permitting a sequence of moves, even for $\kslice{1}$, would immediately yield a relation that subsumes $\mazeq$. 

The arguments for grains and scattered grains are similar. Since already a single slice can deliver different expressivity, we only need to argue that $k$-slice reorderings do not subsume grains. But, since they do not subsume $\mazeq$, they cannot subsume grains or scattered grains (which both subsume $\mazeq$) either. 

\propref{sound-k-slices} implies that $\rfeq$ subsumes $k$-slice orderings. The proof of the strictness of this subsumption is identical to that of \thmref{exp-slicestar}.

\end{proof}

While \propref{k-slice-gradation} states that higher values of $k$
lead to increasingly more reorderings using $k$-slices, this increase in expressiveness converges as $k$ approaches the length of the run.

\begin{restatable}{proposition}{kSliceBounded}
\proplabel{k-slice-bounded}
If $\sigma \kslice{k} \rho$ for some value $k$, then $\sigma \kslice{m} \rho$ for some $m \le |\sigma| - 1$.
\end{restatable}

\begin{proof}
The proof follows from the observation that a list of size $n$ 
can be sorted using at most $n - 1$ insertions since 
each partition in a slicing can be viewed as an insertion, as we previously discussed in Section \ref{sec:comparison}. In other words, one can select each event of $\tr$ as part
of a distinct slice as follows --- if event $e$ appears at the
$i^\text{th}$ position in the target reordering $\rho$, then the subsequence $\tr_i$
can be chosen to contain exactly the singleton set $\set{e}$.
With this choice of subsequences $\tr_1, \tr_2, \ldots, \tr_n$,
it is easy to see that $\rho = \tr_1 \cdot \tr_2 \cdots \tr_n$
and thus $\ksliceto{n-1}{\tr}{\rho}$. 
\end{proof}

In particular, \propref{k-slice-bounded} implies one of the two key defining features of $k$-slices as predictors: limited to runs of length (up to) $k$, it suffices to use at most $k-1$ slices to witness any reordering, even those that are $\rfeq$ to it:

\begin{restatable}{theorem}{kSliceReadsFromBounded}
\thmlabel{k-slice-rf-bounded}
If $\sigma \rfeq \rho$, then $\sigma \kslice{m} \rho$ for some $m \le |\sigma| - 1$.  
\end{restatable}

\begin{proof}
The proof follows from the observation that a list of size $n$ 
can be sorted using at most $n - 1$ insertions since 
each partition in a slicing can be viewed as an insertion, as we previously discussed in Section \ref{sec:comparison}. In other words, one can select each event of $\tr$ as part
of a distinct slice as follows --- if event $e$ appears at the
$i^\text{th}$ position in the target reordering $\rho$, then the subsequence $\tr_i$
can be chosen to contain exactly the singleton set $\set{e}$.
With this choice of subsequences $\tr_1, \tr_2, \ldots, \tr_n$,
it is easy to see that $\rho = \tr_1 \cdot \tr_2 \cdots \tr_n$
and thus $\ksliceto{n-1}{\tr}{\rho}$. 
\end{proof}




The proof follows from the observation that a list of size $n$ can be sorted using at most $n - 1$ insertions since each partition in a slicing can be viewed as an insertion, as we previously discussed in Section \ref{sec:comparison}. In other words, one can select each event of $\tr$ as part
of a distinct slice as follows --- if event $e$ appears at the
$i^\text{th}$ position in the target reordering $\rho$, then the subsequence $\tr_i$
can be chosen to contain exactly the singleton set $\set{e}$.
With this choice of subsequences $\tr_1, \tr_2, \ldots, \tr_k$,
it is easy to see that $\rho = \tr_1 \cdot \tr_2 \cdots \tr_k$
and thus $\ksliceto{k}{\tr}{\rho}$. 

In summary, the expressive power of $k$-sliced reorderings
converges to $\rfeq$ when restricted to runs of bounded length.
More importantly, one can frame this exact result more insightfully. 
Given that successively larger values of $k$ give strictly larger spaces of reorderings, it is imperative to ask --- what happens
in the limit? Let $\mslice = \bigcup_{k \geq 1} \kslice{k}$. Observe that it is the limit of the monotonically
increasing sequence $\kslice{1}, \kslice{2} \ldots$, which guarantees the soundness of  
$\mslice$. Hence an alternative formulation of \thmref{k-slice-rf-bounded} is:


\begin{corollary}
\corlabel{many-slice-rf-limit}
$\mslice = \rfeq$
\end{corollary}

That is, in the limit, $k$-sliced reorderings attain the expressive
power of reads-from equivalence. Immediately, from this theorem, we can conclude that $\mslice$ is reflexive, transitive, and symmetric, and it subsumes all existing sound predictors, since $\rfeq$ does.

\subsection{Checking $k$-sliceability}
\seclabel{stack-height}

A natural question in the context of an equivalence $E \subseteq \labs^* \times \labs^*$ is the
\emph{recognition problem}~\cite{blass1984equivalence} --- given
two runs $\tr$ and $\rho$, how does one determine computationally if $(\tr, \rho) \in E$.
In the case of trace equivalence, it is well understood that the partial order, 
being a canonical representation for a class, can be efficiently used to check if $(\tr, \rho) \in \mazeq$
in linear time.
Likewise, the recognition problem for reads-from equivalence can be solved
in linear time by simply constructing and comparing the program order and reads-from relations.
 Here, we ask an analogous question for $\kslice{k}$, and answer it
 in terms of \emph{slice height}:

\begin{definition}[Slice height]
The {\em slice height} of a pair of runs $\sigma$ and $\rho$, denoted by $\sheight{\tr}{\rho}$,  is the
minimum $k \in \nats$ such that $\tr \kslice{k} \rho$.
We say,  $\sheight{\tr}{\rho} = 0$ if $\tr = \rho$ and 
$\sheight{\tr}{\rho} = \infty$ if $\tr \not\rfeq \rho$.
\end{definition}

Observe that $(\tr, \rho) \in \kslice{k}$ iff
$\sheight{\tr}{\rho} \leq k$.
We show that both the problems of determining the slice height of two runs
as well as the recognition problem for them can be solved in linear time:


\begin{restatable}{theorem}{kSliceDepthLinearTime}[Checking $k$-sliceability]
\thmlabel{sheight}
The problem of computing $\sheight{\tr}{\rho}$ can be solved in linear time.
Thus, the recognition problem $(\tr, \rho) \in \kslice{k}$ can also be solved in linear time.
\end{restatable}

The intuition behind this result is based on a key observation that, 
when $\sigma \rfeq \rho$, then the number of slices can be uniquely 
determined by looking at the number of times a pair of consecutive 
events in $\sigma$ that belong to two different threads appear in inverted order from $\rho$. 
One can argue that the number of slices necessary is one more than this number. 
For instance, recall the two runs in \figref{lim}. The $\sheight{(a)}{(b)} = 2$, 
precisely because there are exactly two pairs of events: 
$(\ev{T_1, \wt(x)}, \ev{T_2, \wt(x)})$
and 
$(\ev{T_1, \rd(y)}, \ev{T_2, \wt(y)})$
 that are consecutive and appear as reordered in (b). 
This can be easily verified by observing that exactly two slices,
namely those demarcated by the two curves in \figref{lim},
will suffice to transform (a) to (b).  
Similarly, $\sheight{(b)}{(a)}=2$ because the following two consecutive pairs
that appear in inverted order as compared to (a):
$(\ev{T_2, \wt(y)}, \ev{T_1, \wt(y)})$
and 
$(\ev{T_2, \rd(x)}, \ev{T_1, \wt(x)})$;
again two slices (like the ones marked in (a), but with opposite threads)
suffice to transform (b) to (a).
More formally, we prove the following towards the proof of \thmref{sheight}:

\begin{restatable}{proposition}{depthDrops}
\proplabel{depth-drops}
Let $\tr$ and $\rho$ be such that $\tr \rfeq \rho$ with $|\tr| = |\rho| = n$.
Let $\pi : \set{1, \ldots, n} \to \set{1, \ldots, n}$ be the permutation
function such that the $i^\text{th}$ event in $\rho$
is the $\pi(i)^\text{th}$ event in $\tr$.
Let $D = \setpred{i}{1 \leq i\leq n-1, \pi(i) > \pi(i+1)}$ be the set of
drop positions in $\pi$.
Then, $\sheight{\tr}{\rho} = |D| + 1$.
\end{restatable}


\section{The New Problem of Predictive Monitoring}
\seclabel{problem}

Reads-from and trace equivalence relations have both been the basis of the classic predictive monitoring problem. $k$-slice reorderings, in contrast, are not symmetric nor transitive, and hence do not yield an equivalence relation. We revisit the formal definition of the {\em predictive monitoring} problem for such relations, in preparation for stating the key result of this paper in the next section. 

A (monitoring) specification (or its negation) is typically represented using a  
language $L \subseteq \labs^*$ denoting the set of buggy executions. Given a run $\tr \in \labs^*$, the predictive monitoring problem against $L$ modulo $R$ can be formalized as the validity of either of the following two sentences: 
\begin{align}
\exists \rho: \rho \in L \land (\sigma, \rho) \in R \\
\exists \rho: \rho \in L \land (\rho, \sigma) \in R 
\end{align}


Recall that,
a reordering relation $R$ is {\em sound} if and only if $(\sigma, \rho) \in R \implies \sigma \rfeq \rho$. 
This means that when $R$ is sound, the validity of either statement about $\tr$, in turn, guarantees the soundness of the predicted bug via the certifying execution $\rho$. 
In this sense, predictive monitoring helps enhance the coverage of an otherwise
vanilla monitoring problem (i.e., `is $\tr \in L$?').
We focus on regular language specifications since they 
can encode a wide class of concurrency bugs. 
We discuss some known examples of bugs in \secref{reg-spec}.


\subsection{Encoding bugs using regular specifications}
\seclabel{reg-spec}

Regular languages have been the de facto standard for
specifying properties in runtime monitoring~\cite{RuntimeVerification2009}.
Their popularity stems from their close computational connection to finite automata
as well as to other specification formalisms such as linear temporal logic (LTL) and
mondadic second order (MSO) logic.
Indeed, several common concurrency bugs can also be encoded as regular languages. We list some concrete instances to outline the breadth of applicability of the techniques proposed in this paper.

\myparagraph{Data Races} While many definitions of data races
have emerged, a prominent one in the predictive analysis literature demarcates an
execution to be racy if in this execution, two conflicting events performed by different
threads appear consecutively~\cite{Kini17,Smaragdakis12,Mathur18,genc2019,Huang14,Mathur21,OSR2024}.
The language corresponding to racy executions is then the following
and is easily regular: 
\[L_{\sf race} = \sum_{\footnotesize\begin{aligned}\begin{array}{c}x \in \mems, t_1 \neq t_2 \in \threads{}, (op_1, op_2) \neq (\rd, \rd) \end{array}\end{aligned}} \labs^* \cdot \ev{t_1, op_1(x)} \cdot \ev{t_2, op_2(x)} \cdot \labs^* \]

\myparagraph{Order Violations}
Order violations are a common source of
errors in low-level systems code and manifest as errors such as use-after-free~\cite{huang2018ufo} and null-pointer dereferences~\cite{Farzan12}.
To formally define them, one first fixes two events (more precisely, types of events)
$\alpha, \beta \in \labs$ whose desired order is $\alpha < \beta$, and thus
the set of executions with such a violation is simply the regular language:

\[L_{\sf OV}^{\alpha,\beta} = \labs^* \beta \cdot \labs^* \cdot \alpha \cdot \labs^*
\]

\myparagraph{Atomicity}
Atomicity is a correctness specification derived from database theory
and is a key correctness property that allows programmers to reason about 
their code easily, in a modular fashion, without the need to consider all possible 
interleavings.
Here, a programmer specifies their intent by denoting which parts of code are assumed atomic by marking them with begin ($\btx$)
and end ($\etx$) instructions.
Executions thus also include events corresponding to \emph{transaction boundaries}
denoted by begin and end instructions.
An execution is atomic if all transactions appear serially, i.e., without interference
from other threads.
Formally, let $\labs' = \labs \cup \setpred{\ev{t, \btx}, \ev{t, \etx}}{t \in \threads{}}$
denote the extended set of event labels.
The language of atomic runs is then the following:

\[
L_{\sf serial} = \big(\sum\limits_{t \in \threads} \ev{t, \btx} \cdot \labs_t^* \cdot \ev{t, \etx} \big)^*
\]
In the above, $\labs_t$ is the subset of $\labs$ performed by thread $t$.
In words, a serial run is a sequence of transactions, where each transaction
is performed by a single thread $t$ which begins with the $\btx$ event by thread $t$,
performs a bunch of non-begin and non-end events of thread $t$ and ends in $\etx$ event by $t$.
We remark that the problem of predictive monitoring against
the $L_{\sf serial}$ modulo trace equivalence precisely corresponds to
checking \emph{conflict serializability}~\cite{confser}, while
checking it modulo reads-from equivalence corresponds to \emph{view serializability}~\cite{confser}.

\myparagraph{Pattern Languages}
In \cite{ang2023predictive}, pattern languages were introduced to specify more expressive specifications
for finding bugs in concurrent programs, in line with small depth hypotheses~\cite{pct,chistikov2016}. 
A pattern language of dimension $d \in \nats$ is a regular language
of the following form ($a_1, a_2 \ldots, a_d \in \labs$):
\[\mathtt{Patt}_{a_1, a_2, \ldots, a_d} = \labs^* a_1 \labs^* \ldots \labs^* a_d \labs^*\]

\subsection{Predictive Monitoring as Image Computation}

The predictive membership problem can be equivalently stated as a membership in a pre-/post-image of the specification language under the reordering relation $R$:

\begin{definition}[Predictive Membership With Image Computation]\deflabel{new-problem}
Let $L \subseteq \labs^*$ be a language and let 
$R \subseteq \labs^* \times \labs^*$.
The pre-image and post-image of $L$ under $R$ are defined as follows:
\begin{align*}
\pre{L}{R} &= \setpred{\tr \in \labs^*}{\exists \rho \in L, (\tr, \rho) \in R} \\
\post{L}{R} &= \setpred{\rho \in \labs^*}{\exists \tr \in L, (\tr, \rho) \in R}
\end{align*}
The membership problem with pre-image (respectively post-image) computation asks if $\sigma \in \pre{L}{R}$ (respectively $\sigma \in \pre{L}{R}$). 
\end{definition}


This membership problem can be solved in constant space and linear time
iff the language $\pre{L}{R}$ (respectively $\post{L}{R}$) is regular.
Thus, for the setting of our work, a reordering relation $R$ is desirable
if the pre-/post-image of every regular language under $R$ is also a regular language.

When the reordering relation is an equivalence (say $\sim)$, 
then $\pre{L}{\sim}$ is also known as the \emph{closure} of $L$ under $\sim$
and denoted as $[L]_\sim$;
indeed, observe the closure property  $[[L]_\sim]_\sim = [L]_\sim$.

The closure $[L]_{\rfeq}$ of $L$ under $\rfeq$ is known to be
non-regular even for a very simple regular language that captures data races.
For Mazurkiewicz's trace equivalence $\mazeq$, 
the closure of data races or atomicity violation specifications 
is known to be regular~\cite{FarzanMathur2024,FM08CAV,ang2023predictive}.
However, for an arbitrary regular language, its closure under $\mazeq$ 
can be beyond context-free languages, and cases under which regularity
is preserved have been studied~\cite{Ochmanski85,Bouajjani2007APC,PartialCommutationRegular2008Gomez,ang2023predictive}.
We show that, the pre- and post-image of arbitrary regular languages under
grain and scattered grain commutativity relations, as well as $\sliceq$ are also
non-regular, and show that this happens for reasons similar to the case of trace equivalence:


\begin{restatable}{proposition}{nonRegularityOfClosures}
\proplabel{non-regular-closure}
There is a regular language $L$ such that each of 
$[L]_{\rfeq}$, $[L]_{\mazeq}$, $\pre{L}{\sliceq}$,
 $\post{L}{\sliceq}$, $\pre{L}{\graineq}$, $\post{L}{\graineq}$ 
is non-regular\footnote{We formally do not list scattered grains. The result holds simply because scattered grains can simulate Mazurkiewicz. But, since the definition of a prediction relation based on scattered grains does not appear in \cite{FarzanMathur2024}, we refrain from stating the result based on an undefined closure.}.
\end{restatable}

\begin{proof}
The proof follows from \thmref{beyond-trace} and the observations that
(1) a language is regular iff its membership problem can be solved in constant space, and
(2) $\mazeq$ is subsumed by each of ${\rfeq}, {\mazeq}, {\sliceq}, \graineq$
\end{proof}

In \secref{algo}, we show that the pre-image of any arbitrary regular
language under $\kslice{k}$ is regular (see \thmref{regularity-of-forward-closure})
and discuss how to compute the DFA representation of this regular language.
In sharp contrast, we show that the post-image of a regular language under $\kslice{k}$ is not necessarily regular (see \thmref{post-image-liear-space-hardness}). This is precisely why in \defref{new-problem}, we maintained a separation between the two modes of predication and did not combine them as a single predictor with more predictive power.

\section{Predictive monitoring modulo sliced reorderings}
\seclabel{algo}

In this section, we consider the predictive monitoring 
question modulo sliced reorderings.
Recall from \secref{problem} that we opt for a language-theoretic view
on predictive monitoring and study the image of regular specifications
under sliced reorderings.
More precisely, we show in \secref{pre-image-slices}
that the pre-image of a regular language
under $k$-sliced reorderings ($k \in \natsp$) is actually regular,
thus allowing us to solve the predictive monitoring problem efficiently 
in \emph{constant space and linear time}. 
We also consider the dual problem of determining the post-image
of a language $L$ in \secref{post-image-slices} and show that this may not be regular
even when $L$ is regular.

\subsection{Pre-image of regular languages under sliced reorderings}
\seclabel{pre-image-slices}

Our key result is that the pre-image of a regular language $L$ under
the $k$-sliced reordering relation (for a fixed $k \in \nats$)
is also a regular language (\thmref{regularity-of-forward-closure}).
In the following, we first give an overview of the proof of this result, and then give the details.

\begin{figure}[t]
\begin{subfigure}[t]{0.35\textwidth}
\centering
\includegraphics[scale=0.23]{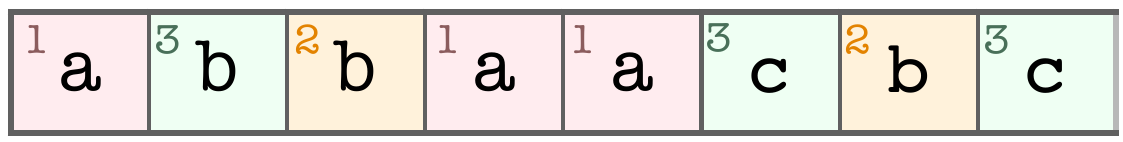}
\caption{Annotating $\texttt{abbaacbc}$. }
\figlabel{annotated-word}
\end{subfigure}
\begin{subfigure}[t]{0.18\textwidth}
\centering
\includegraphics[scale=0.15]{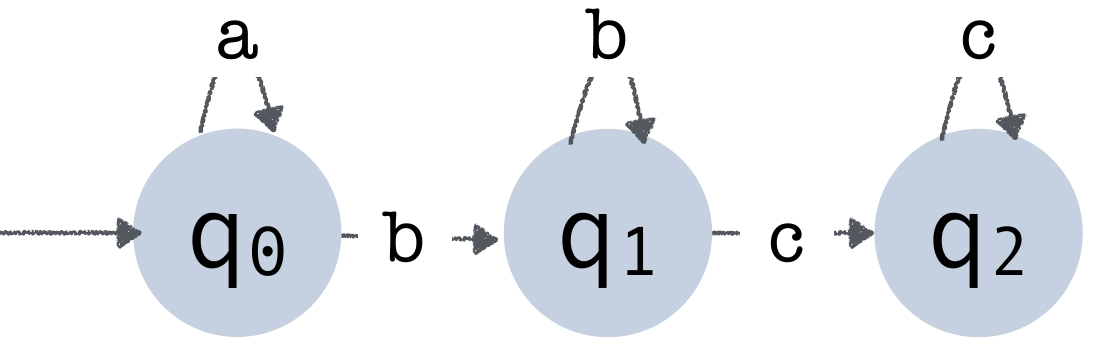}
\caption{Automaton $\aut$}
\figlabel{automaton}
\end{subfigure}
\hfill
\begin{subfigure}[h]{0.42\textwidth}
\centering
\includegraphics[scale=0.23]{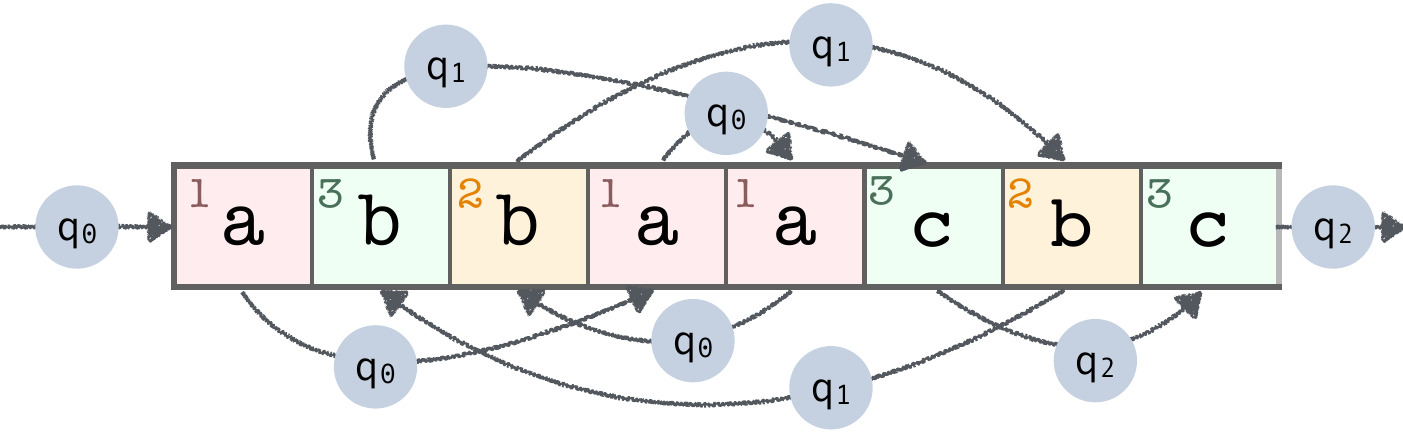}
\caption{Run of $\aut$ on reordered run.}
\figlabel{out-of-order-run}
\end{subfigure}
\vspace{-10pt}
\caption{Illustrating the construction of the automaton $\autcl{k}$
for pre-image ($k = 2$) of the regular language
$L(\aut) = {\texttt{a}^*\texttt{b}^+\texttt{c}^+}$ (middle, (b)).
The automaton $\autcl{k}$ first guesses an annotation (left, (a)) for every event in the execution
$\tr = \texttt{abbaacbc}$ denoting the reordering
$\rho = \texttt{aaabbbcc}$, and then simulates the 
original automaton $\aut$
on the reordered run $\rho$ (right, (c)). 
}
\figlabel{illustration-automaton}
\vspace{-0.1in}
\end{figure}

\myparagraph{Overview}
At a high level, we start with the NFA $\aut$ for the regular
language $L$ and derive the NFA $\autcl{k}$ 
that accepts $\pre{L}{\kslice{k}}$.
We refer readers to our running example in \figref{illustration-automaton}
where we work with the automaton that accepts the language $\texttt{a}^*\texttt{b}^+\texttt{c}^+$ 
(\figref{automaton}), where $\texttt{a}, \texttt{b}, \texttt{c} \in \labs$.
Recall that our new automaton $\autcl{k}$ must accept a run $\tr$ iff there are disjoint 
subsequences $\tr_1, \tr_2, \ldots, \tr_{k+1}$ such that
the reordering $\rho = \tr_1 \cdot \tr_2 \cdots \tr_{k+1}$
obtained by successive concatenation of these subsequences satisfies:
(\underline{\em consistency}) the concatenated string $\rho$ is $\rf{}$-equivalent to $\tr$, and 
(\underline{\em membership}) the concatenated string $\rho$ is a string that the automaton $\aut$ accepts.

Since these checks can be performed more conveniently over runs
which already demarcate the $k$ subsequences, we will work with 
the `annotated' alphabet, where letters are also identified with the index of the subsequence they belong to:
\[\hat{\labs} = \labs \times \set{1, 2, \ldots, k+1}\]
Consider for example the execution $\tr = \texttt{abbaaacbc}$
and a possible annotation of it in \figref{annotated-word};
the reordering corresponding to this annotation is  $\rho = \texttt{aaabbbcc}$.
For an annotated execution $\hat{\tr} \in \hat{\labs}^*$,
we will use the notation $\proj{\hat{\tr}}{i}$ to denote
the maximal subsequence of $\hat{\tr}$ each of whose events have annotation $i$.
Towards our main result, we will consider the following two 
languages over the alphabet $\hat{\labs}$, capturing the requirements of
\underline{\em consistency} and \underline{\em membership} outlined above:
\begin{align*}
\begin{array}{rcl}
\hat{L}_{\cnst} 
&=& 
\setpred{
\hat{\tr} \in \hat{\labs}^*
}{
h(\hat{\tr}) \rfeq h(\proj{\hat{\tr}}{1}) \cdot h(\proj{\hat{\tr}}{2}) \cdots h(\proj{\hat{\tr}}{k+1})
} \\
\hat{L}_{\memb} 
&=& 
\setpred{
\hat{\tr} \in \hat{\labs}^*
}{
h(\proj{\hat{\tr}}{1}) \cdot h(\proj{\hat{\tr}}{2}) \cdots h(\proj{\hat{\tr}}{k+1}) \in L
}
\end{array}
\end{align*}
Here, $h: \hat{\labs} \to \labs$ is the projection homomorphism
given by $h((a, i)) = a$ for every $a \in \labs$ and $i \in \set{1, 2, \ldots, k+1}$.
We will show that both $\hat{L}_{\cnst}$ and $\hat{L}_{\memb}$
are regular.
Together with the observation that $\pre{L}{\kslice{k}} = h(\hat{L}_{\cnst} \cap \hat{L}_{\memb})$,
it follows that  $\pre{L}{\kslice{k}}$ is regular, as desired.

\myparagraph{Automaton for $\hat{L}_{\cnst}$}
An algorithm that checks membership
of an execution $\hat{\tr}$ in $\hat{L}_{\cnst}$ 
essentially checks if the unique reordering $\hat{\rho} = 
\proj{\hat{\tr}}{1} \proj{\hat{\tr}}{2} \cdots \proj{\hat{\tr}}{k+1}$
is such that $\hat{\tr} \rfeq \hat{\rho}$.
As such, there are complexity-theoretic limits on efficiently solving problems
pertaining the existence of $\rf{}$-equivalent 
reorderings that satisfy even very simple properties~\cite{FarzanMathur2024,Mathur2020b}.
We instead show that, when parametrized by 
a maximum slice-width $k$, this question becomes efficiently checkable, using an automata-theoretic
algorithm, i.e.,
using only as much time and space as is afforded by a DFA.
Our construction, in turn, relies on the following key observation
that outlines how the requirement that `within $\proj{\hat{\tr}}{i}$
the relative order of events does not change' can be leveraged
to efficiently check the consistency of the annotation:
\begin{restatable}{lemma}{consistencyOfSlices}
\lemlabel{consistency}
Let $\tr \in \labs^*$ be a concurrent program execution, $k \in \natsp$,
$\tr_1, \tr_2, \ldots, \tr_{k+1}$ be a partitioning 
of $\tr$ into subsequences, and $\rho = \tr_1 \cdot \tr_2 \cdots \tr_{k+1}$.
We have $\tr \rfeq \rho$ iff 
\begin{enumerate}
	\item $\po{\tr}$ is aligned with respect to the subsequences, 
	i.e., for every $(e, f) \in \po{\tr}$,
	such that $e \in \events{\tr_i}$ and $f \in \events{\tr_j}$, 
	we have $i \leq j$, and,
	
	\item $\rf{\tr}$ is aligned with respect to the subsequences.
	That is, for a read event $e_\rd \in \events{\tr_i}$:
	\begin{enumerate}
		\item If $\rf{\tr}(e_\rd)$ is not defined (i.e., $e_\rd$ is an \emph{orphan read}), then
		for every write event $e'_\wt$ ($\OpOf{e'_\wt} = \wt$ and $\MemOf{e'_\wt} = \MemOf{e_\rd}$)
		with $e'_\wt \in \events{\tr_\ell}$, we have that $\ell \geq i$.

		\item If $e_\wt = \rf{\tr}(e_\rd)$ is defined (with $e_\wt \in \events{\tr_j}$), then
		$j \leq i$ and for every other write event $e_\wt \neq e'_\wt \in \tr_\ell$ such that
		$\OpOf{e'_\wt} = \wt$ and $\MemOf{e'_\wt} = \MemOf{e_\rd}$
		we have 
		(i) $(\ell \leq j \lor \ell \geq i)$, and
		(ii) if $\ell = i \land j < i$ then $e_\rd \trord{\tr} e'_\wt$, and
		(iii) if $\ell = j \land j < i$ then $e'_\wt \trord{\tr} e_\wt$.
	\end{enumerate}
\end{enumerate}
\end{restatable}

Expert readers may already observe that the characterization
of~\lemref{consistency} is FO-definable, given that both $\po{\tr}$
and $\rf{\tr}$ are FO-definable in terms of the total order $\trord{\tr}$.
In the following, we instead describe a DFA $\aut_{\cnst} = (Q_{\cnst}, q^0_{\cnst}, \delta_{\cnst}, F_{\cnst})$
over the alphabet $\hat{\labs}$, directly inspired from
~\lemref{consistency}.

\noindent
\underline{States.}
The states $Q_{\cnst}$ comprise of a unique rejecting state $\bot$, or
is a tuple of the form 
\[
q = (\mathsf{T2S}_q, \mathsf{LastW}_q, \mathsf{SeenW}_q, \mathsf{ForbiddenW}_q) \in Q_{\cnst},
\]
with:
\begin{itemize}
\item $\mathsf{T2S}_q: \threads \to \set{0, 1,\ldots,k+1}$
\item $\mathsf{LastW}_q: \mems \to \set{0,1,\ldots,k+1}$
\item $\mathsf{SeenW}_q: \mems \to \powset{\set{1,\ldots,k+1}}$
\item $\mathsf{ForbiddenW}_q: \mems \to \powset{\set{1,\ldots,k+1}}$
\end{itemize}
Informally, after reading some prefix $\pi$, if the automaton reaches
some state $q \neq \bot$, then $\mathsf{T2S}_q(t)$ stores the
largest slice index seen so far for thread $t$,
$\mathsf{LastW}_q(x)$ stores the slice index of the latest write event on memory location $x$,
$\mathsf{SeenW}_q(x)$ tracks the set of all slices that have, so far, witnessed a write event on $x$
and 
$\mathsf{ForbiddenW}_q(x)$ tracks the set of all slices that must not,
in the future, see a write on $x$.
The initial state is $q^0_{\cnst} = (\lambda t \cdot 0, \lambda x \cdot 0, \lambda x \cdot \emptyset, \lambda x \cdot \emptyset)$.
A state is accepting iff it is not the sink $\bot$; i.e., $F_{\cnst} = Q_{\cnst}\setminus\{\bot\}$.

\noindent
\underline{Transitions.}
The state $\bot$ is a sink, i.e., $\delta_{\cnst}(\bot, (a, i)) = \bot$ for every $(a, i) \in \hat{\labs}$.
Otherwise on input symbol $(e,i)\in \hat{\labs}$ (with $e=\ev{t, op(x)}$),
and on state $p = (\mathsf{T2S}_p, \mathsf{LastW}_p, \mathsf{SeenW}_p, \mathsf{ForbiddenW}_p)$,
the resulting state $q = \delta_{\cnst}(p, (e, i))$ is defined as follows.
If the following hold, then $q = \bot$ (here, $(a,b]$ denotes $\setpred{\ell}{a < \ell \leq b}$):
\begin{align*}
\begin{array}{c}
\big(\mathsf{T2S}_p(t) > i\big) \lor \big(op = \wt \land i \in \mathsf{ForbiddenW}_p(x)\big) \\
\lor \\
\big((op = \rd) \land (\mathsf{LastW}_p(x) > i \lor \mathsf{SeenW}_p(x) \cap (\mathsf{LastW}_p(x), i]  \neq \emptyset)\big)
\end{array}
\end{align*}
Otherwise, we have 
$q = (\mathsf{T2S}_q, \mathsf{LastW}_q, \mathsf{SeenW}_q, \mathsf{ForbiddenW}_q)$, where
$\mathsf{T2S}_q = \mathsf{T2S}_p[t \mapsto i]$, and
\begin{enumerate}
	\item if $op = \wt$, then $\mathsf{LastW}_q = \mathsf{LastW}_p[x \mapsto i]$,
		$\mathsf{SeenW}_q = \mathsf{SeenW}_p[x \mapsto \mathsf{SeenW}_p(x) \cup \set{i}]$
		and $\mathsf{ForbiddenW}_q = \mathsf{ForbiddenW}_p$.
	\item if $op = \rd$, then $\mathsf{LastW}_q = \mathsf{LastW}_p$, $\mathsf{SeenW}_q = \mathsf{SeenW}_p$ 
	and $\mathsf{ForbiddenW}_q = \mathsf{ForbiddenW}_p[x \mapsto \mathsf{ForbiddenW}_p(x) \cup \setpred{\ell}{\mathsf{LastW}_p(x) \leq \ell < i}]$.
\end{enumerate}

\begin{restatable}{lemma}{consistencyRegular}
\lemlabel{consistency-regular}
$L(\autcnst) = \hat{L}_{\cnst}$. Thus, $\hat{L}_{\cnst}$ is regular
\end{restatable}


\myparagraph{Automaton for $\hat{L}_{\memb}$}
We construct a DFA 
$\autmemb = (\statesmemb,$ $\initmemb, \trnsmemb, \accptmemb)$.
that in turn simulates the DFA 
$\aut = (\states, \init, \trns, \accpt)$ for the language $L$,
on each subsequence 
$\proj{\hat{\tr}}{1}, \proj{\hat{\tr}}{2}, \ldots, \proj{\hat{\tr}}{k+1}$.
\figref{out-of-order-run} pictorially illustrates the
challenge that $\autmemb$ addresses --- this automaton
must process events of $\tr$ out-of-order to accurately simulate $\aut$ on the reordered execution $\rho$.
Readers with expertise in automata theory may observe that
one can come up with a $2$-way automaton for this task,
which can then be translated to a DFA~\cite{Vardi1989TwoWay}; here we present a direct construction instead.
Each state $q \in \statesmemb$ is of the form $q \in [\set{1, 2, \ldots, k+1} \times \states \to \states]$,
and tracks, after reading a prefix $\hat{\pi}$ of $\hat{\tr}$,
the state that $\aut$ would result into when having read $h(\proj{\hat{\pi}}{i})$ 
starting from each state $p \in \states$.
The initial state $\initmemb$ is such that for every $i \in \set{1, \ldots, k+1}$
and for every $p \in \states$, we have $q(i, p) = p$.
The transitions are given as follows.
Starting from state $q$ on reading input $(i, a) \in \hat{\labs}$,
the resulting state $q' = \trnsmemb(q, (i, a))$ is given by
$q'(j, p) = \trns(q(j, p), a)$ if $j = i$, and otherwise
it is $q'(j, p) = q(j, p)$.
Finally, the set of final states is those in which, intuitively, the final states for each subsequence
match the initial state of the next subsequence and is an accepting state of $\aut$ for the very last one.
Formally,  a state $q \in \accptmemb$ iff there is a sequence of states
$p_1, p_2, \ldots, p_{k+1} \in \states$ such that
$p_1 = q(1, \init)$,
for every $1 \leq i \leq k$, $p_{i+1} = q(i+1, p_i)$
and finally $p_{k+1} \in \accpt$.
The correctness of the above construction is stated as follows.

\begin{restatable}{lemma}{membershipRegular}
\lemlabel{membership-regular}
$L(\autmemb) = \hat{L}_{\memb}$. Thus, $\hat{L}_{\memb}$ is regular
\end{restatable}

\myparagraph{Putting it together}
Since both $\hat{L}_{\cnst}$ and $\hat{L}_{\memb}$ are shown to be regular,
their intersection $\hat{L}_{\cnst\land \memb} = \hat{L}_{\cnst} \cap \hat{L}_{\memb}$
is also a regular language.
Now, we observe that $\pre{L}{\kslice{k}} = h(\hat{L}_{\cnst\land \memb})$, i.e.,
the set of executions in $\pre{L}{\kslice{k}}$ are precisely those that have a corresponding annotation
which is both consistent and whose reorderings (dictated by the annotation)
belong to $L$. Since regular languages are closed under homomorphism, we have the following:

\begin{restatable}{theorem}{forwardClosureRegular}
\thmlabel{regularity-of-forward-closure}
Let $L$ be a regular language and let $k \in \natsp$. 
The image $\pre{L}{\kslice{k}}$ is regular.
\end{restatable}

In the context of predictive monitoring, we are interested in the complexity-theoretic
aspects of  predictive membership, which follow straightforwardly 
as a consequence of \thmref{regularity-of-forward-closure}

\begin{corollary}
\corlabel{complexity-closure-membership}
Fix a language $L \subseteq \labs^*$ and a constant $k \in \natsp$. 
The predictive membership problem against $L$ modulo $k$-sliced reorderings
can be solved in constant space and linear time.
\end{corollary}

Let us also analyze the actual space usage for the monitoring problem
by counting the number of states in the automaton for $\pre{L}{\kslice{k}}$.
The number of states in the automaton $\autcnst$ 
is $|\statescnst|=O((k+2)^{|\threads|+|\mems|}\cdot2^{2(k+1)\cdot |\mems|})$.
Suppose the automaton for $L$ is a DFA with $m = |\states|$ states.
Then, number of states in $\autmemb$ is $|\statesmemb|=m^{(k+1)\cdot m}$.
Their product automaton has $O\big((k+2)^{|\threads|+|\mems|}\cdot2^{2(k+1)\cdot|\mems|} \cdot m^{(k+1)\cdot m}\big)$ states. Finally, the automaton obtained from the homomorphism is an NFA with the same number
of states.
When monitoring against an NFA, one needs as much memory as the number of states of the NFA.
Thus, a conservative estimate of the
space usage of the predictive monitoring algorithm is $O\big((k+2)^{|\threads|+2|\mems|}\cdot2^{|\mems|(k+1)} \cdot m^{(k+1)\cdot|m|}\big)$, which is constant,
assuming that the alphabet $\labs$ (and thus $|\threads|$ and $\mems$) as well
as the parameter $k$ is constant, i.e., independent of the length of the run being monitored.
This also shows that membership in pre-image of a regular language (with DFA of $m$ states)
is FPT in the parameter $|\labs| + m + k$.


\subsection{Post-image of regular languages under sliced reorderings}
\seclabel{post-image-slices}

Here, we investigate the \emph{dual} predictive membership problem
of checking if there is a $\rho \in L$ such that $(\rho, \tr) \in \kslice{k}$
for a given input execution $\tr$.
Recall that this boils down to the vanilla membership problem in the
image $\post{L}{\kslice{k}}$.
Here, we show that, unlike with pre-images,
the post-image under sliced reorderings  does
not, in general, admit a constant-space linear-time monitoring algorithm, 
and in fact, becomes as hard as predictive monitoring
modulo $\rfeq$~\cite{FarzanMathur2024}, \emph{for every $k \in \natsp$} (even for $k = 1$),
even for very simple regular languages:

\begin{restatable}{theorem}{postImageLinearSpace}
\thmlabel{post-image-liear-space-hardness}
Let $k \in \natsp$.
Let $\alpha = \ev{T_2, \wt(u)}$ and $\beta = \ev{T_1, \wt(u)}$,
where $u \in \mems{}$ and $T_1, T_2 \in \threads{}$ (with $t_1 \neq t_2$),
and let $L$ be the fixed regular language 
$L = \labs^* \beta \cdot \labs^* \cdot \alpha \cdot \labs^*$.
Any algorithm that checks for membership in $\post{L}{\kslice{k}}$
in a streaming fashion must use space at least linear in
 the length of the input execution.
\end{restatable}


\section{Frontier-graph style algorithms}
\seclabel{frontier}

In \secref{pre-image-slices} and \secref{post-image-slices},
we investigated the problem of predictive \emph{monitoring}
modulo slices, where we focused on algorithms (or their provable non-existence)
that work in a streaming fashion while using space that is independent of the
length of the input word.
Here, in turn, we ask if there are alternative algorithms for solving
the predictive membership problem when such a restriction is not imposed.
We show that the classic paradigm of \emph{frontier graph} 
algorithms~\cite{Gibbons1997,Gibbons1994,Mathur2020b,Agarwal2021} can be adopted to the setting of slices
to answer membership in pre- and post-images of regular languages.
In doing so, we also establish new upper and lower bounds for predictive monitoring,
complementing those in \secref{pre-image-slices} and \secref{post-image-slices}.

At a high level, a \emph{frontier} is a subset of events of the trace $\tr$
which is downward closed with respect to some partial 
order, such as the program order $\po{\tr}$.
A frontier graph is then a graph whose nodes are such frontiers and whose
edges represent extension of frontiers by a single event, while ensuring
other constraints, such as, preservation under reads-from, are met.
Intuitively, paths in such a frontier graph represent all possible
(or a precisely defined subset of) reads-from equivalent
executions.
Intuitively, a frontier graph algorithm for, say, predictive membership
against a regular language $L$ annotates each frontier $X$ with set of those states
that correspond to the paths that lead to $X$.
In the following, we show that such algorithms can also be used
to check for membership in pre and post images,
by in turn showing that these state annotations can be computed inductively on the frontier graph.
In turn, these yield algorithms whose running times offer a different tradeoff with respect
to paramters such as $|\threads|$ (number of threads)
or the slice bound $k$.

We first show that the problem of membership in pre-image
can be solved in time that varies polynomially with $k$
and the number of states $m$ in the automaton for $L$,
but exponentially with $|\threads|$:

\begin{restatable}{theorem}{preImageKSliceFrontier}
\thmlabel{pre-image-kslice-frontier}
Fix $k \in \natsp$ and a regular language $L \subseteq \labs^*$
given by an NFA with $m$ states.
There is an algorithm that, given an input run $\tr \in \labs^*$
of length $n$, decides whether $\tr \in \pre{L}{\kslice{k}}$
in time $O\!\left(m \cdot |\threads| \cdot (k+1) \cdot (n+1)^{|\threads|}\right).$
\end{restatable}

\newcommand{\pos}{\textsf{pos}}

\begin{proof}
Let $\tr \in \labs^*$ be the input run of length $n$, and let
$A=(Q,\delta,Q_0,F)$ be the fixed NFA for $L$, with $|Q|=m$.
We describe a frontier-graph algorithm that decides whether
$\tr \in \pre{L}{\kslice{k}}$, i.e., whether there exists a run
$\rho \in L$ such that $\tr \kslice{k} \rho$.

Recall that $\tr \kslice{k} \rho$ holds iff $\rho \rfeq \tr$ and the
permutation $\pi$ mapping positions of $\rho$ to positions of $\tr$
has at most $k$ drop positions.
By \propref{depth-drops}, this is equivalent to requiring that
$\pi$ is a concatenation of at most $(k{+}1)$ increasing runs.

We construct a directed \emph{frontier graph} $G=(V,E)$ whose paths
correspond to such runs $\rho$.
A node of $G$ is a triple 
\[(X,\ell,d),\] where:
(i) $X \subseteq \events{\tr}$ is a set of events that is downward closed
with respect to program order $\po{\tr}$;
(ii) $\ell \in \set{0,1,\ldots,n}$ is the position in $\tr$ of the last
event emitted; $\ell=0$ indicates that no event has yet been
emitted. 
(iii) $d \in \set{0,1,\ldots,k}$ counts the number of drop positions
created so far.
The initial node is $(\emptyset,0,0)$, and terminal nodes are those of
the form $(\events{\tr},\ell,d)$ with $d \le k$.

There is an edge
\[
  (X,\ell,d) \xrightarrow{e} (Y, \ell', d')
\]
iff the following conditions hold.
First, $Y = X \uplus \set{e}$.
In addition, if $e$ is a write on variable $x$, then for every write event $e_w$ on $x$
with $e_w \in X$, we have that $\setpred{e_r \in \events{\tr}}{(e_w, e_r) \in \rf{\tr}} \subseteq X$.

Let $i=\pos_{\tr}(e)$, i.e., $e$ is the $i^{\text{th}}$ event of $\tr$.
We set $\ell' = i$, and update the drop counter as follows:
if $\ell = 0$ or $\ell < i$, then no new drop is created and $d' = d$;
otherwise $\ell > i$, in which case the permutation goes down and we set
$d' = d+1$, requiring that $d' \le k$.
No edge is added if this condition is violated.

By construction, every path in $G$ from $(\emptyset,0,0)$ to a terminal
node spells a permutation $\rho$ of $\tr$ that respects program order,
is reads-from equivalent to $\tr$, and whose associated permutation has
at most $k$ drop positions.
Hence $\tr \kslice{k} \rho$ by \propref{depth-drops}.
Conversely, if there exists $\rho$ such that $\tr \kslice{k} \rho$, then
$\rho$ induces a permutation with at most $k$ drops and corresponds to a
path in $G$ from the initial node to a terminal node.

To enforce the regular constraint $\rho \in L$, we propagate NFA states
over the frontier graph.
For each node $v=(X,\ell,d)$ we maintain a set $P(v)\subseteq Q$ of NFA
states reachable after reading the label sequence of some path from the
initial node to $v$.
We initialize $P(\emptyset,0,0)=Q_0$ and propagate as follows:
for every labeled edge $v \xrightarrow{e} v'$ we add
$\delta(P(v),e)$ to $P(v')$.
A terminal node $v$ is accepting iff $P(v)\cap F \neq \emptyset$.
Thus the algorithm accepts iff there exists $\rho \in L$ such that
$\tr \kslice{k} \rho$.

Let us now analyze the running time.
Let $n_t$ be the number of events of $\tr$ in
thread $t \in \threads$.
The number of possible frontiers is exactly
$\prod_{t\in \threads} (n_t+1) \le (n+1)^{|\threads|}$.
For each frontier, the drop counter $d$ ranges over $\set{0,\ldots,k}$,
and $\ell$ can range over $\set{0, \ldots, n+1}$
Hence the total number of nodes is $O\!\left(k \cdot (n+1)^{|\threads|+1}\right)$.
From each node, there are at most $|\threads|$ outgoing edges corresponding to
enabled events and thus total number of edges is $O\!\left(|\threads| \cdot k \cdot (n+1)^{|\threads|+1}\right)$.
Therefore, the overall running time of the algorithm is
$O\!\left(m \cdot |\threads| \cdot k \cdot (n+1)^{|\threads|+1}\right)$.
\end{proof}

Next, we show that, complementary to the linear space lower bound
of \thmref{post-image-liear-space-hardness}, membership in post image
can be solved with a frontier graph algorithm in time
that whose time varies with a factor of $O(n^{O(k)})$:

\begin{restatable}{theorem}{postImageKSliceFrontier}
\thmlabel{post-image-kslice-frontier}
Fix $k \in \natsp$ and a regular language $L \subseteq \labs^*$
given by an NFA with $m$ states.
There is an algorithm that, given an input run $\tr \in \labs^*$
of length $n$, decides whether $\tr \in \post{L}{\kslice{k}}$
in time
$O\!\left(m \cdot n^{k} \cdot \beta (n+1)^{\beta}\right)$,
where $\beta = \min(k+1, |\threads|)$
\end{restatable}

\begin{proof}
Fix the input run $\tr = a_1a_2\cdots a_n \in \labs^*$.
We first observe that $\ksliceto{k}{\rho}{\tr}$ implies that
$\tr$ can be written as a concatenation of $(k{+}1)$ contiguous blocks
\[
  \tr = \rho_1 \cdot \rho_2 \cdots \rho_{k+1},
\]
where $\rho_1, \rho_2, \ldots, \rho_{k+1}$ are subsequences of $\rho$
and also form a partition of $\rho$.
Our algorithm for checking membership of $\tr$ in $\post{L}{\kslice{k}}$,
therefore, works by enumerating all choices of $k$ cut points in $\tr$
$1 \le c_1 < \cdots < c_k \le n$
and checking, for each such choice, whether there exists $\rho \in L$ such that $\rho$ is a shuffle
(interleaving) of $s_1,\ldots,s_{k+1}$ and also satisfies
feasibility constraints ensuring $\rho \rfeq \tr$;
here for the choice $(c_1, \ldots, c_k)$ of cutpoints,
the blocks $s_1, \ldots, s_{k+1}$ are given by
$s_j = \tr[c_j, \min(c_{j+1}, n)]$. 
We remark that, the number of choices of cut positions is
$\binom{n}{k} \in O(n^k)$.

We now show that, for a fixed decomposition
$\tr = s_1 \cdots s_{k+1}$, the problem of checking whether there exists
a shuffle $\rho$ of $s_1,\ldots,s_{k+1}$ such that
$\rho \in L$ and $\rho \rfeq \tr$ can be solved in time
$O\!\left(m \cdot (n+1)^{\beta}\right)$ using a frontier graph algorithm,
where $\beta = \min(k+1, |\threads|)$.
In the following, we fix this decomposition
$(s_1,\ldots,s_{k+1})$ and construct the associated frontier graph.

We now define the frontier graph associated with the fixed decomposition
$\tr = s_1 \cdots s_{k+1}$.
Let $\prec_{\sf blk}$ denote the union of the total orders induced by the
blocks $s_1,\ldots,s_{k+1}$.
Define $\prec := \prec_{\sf blk} \cup \po{\tr}$; recall that $\po{\tr}$
is the \emph{program order} of $\tr$, and is a union of $|\threads|$ total orders.
A \emph{frontier} is a set $X \subseteq \events{\tr}$ that is downward
closed with respect to $\prec$, i.e.,
for every pair of events $(e, e')$,
if $(e \prec e' \wedge e' \in X)$, then we have $e \in X$.
Intuitively, a frontier represents a prefix of each block and of each
thread.
Let $V$ be the set of all such frontiers.

We define a directed graph $G=(V,E)$ as follows.
There is an edge $X \xrightarrow{e} Y$ iff $Y = X \uplus \set{e}$ and:
(i) all $\prec$-predecessors of $e$ are contained in $X$;
and
(ii) if $e$ is a write on variable $x$, then for every write event $e_w$ on $x$
with $e_w \in X$, we have that $\setpred{e_r \in \events{\tr}}{(e_w, e_r) \in \rf{\tr}} \subseteq X$.
The initial frontier is $X_{\mathsf{init}}=\emptyset$, and the unique
terminal frontier is $X_{\mathsf{fin}}=\events{\tr}$.
By construction, every path in $G$ from $X_{\mathsf{init}}$ to
$X_{\mathsf{fin}}$ spells a word $\rho$ that is a shuffle of
$s_1,\ldots,s_{k+1}$.
Moreover, the write--read closure condition ensures that any such word
$\rho$ is reads-from equivalent to $\tr$.

We now incorporate the regular constraint $\rho \in L$.
Let $\aut = (Q, Q_0, \delta, F)$ be the NFA for $L$.
For each frontier $X \in V$, we associate a set $P(X) \subseteq Q$ of NFA
states, defined as the least mapping satisfying:
\begin{itemize}
	\item $P(X_{\mathsf{init}}) = Q_0$, 
	\item for every edge  $X \xrightarrow{e} Y,\; P(Y) \supseteq \delta(P(X), e)$
\end{itemize}
Equivalently, $P(X)$ is the set of all NFA states reachable after reading
the label sequence of some path from $X_{\mathsf{init}}$ to $X$.
This mapping can be computed by a standard forward worklist algorithm.
There exists a word $\rho \in L$ labeling a path from
$X_{\mathsf{init}}$ to $X_{\mathsf{fin}}$ if and only if
$P(X_{\mathsf{fin}}) \cap F \neq \emptyset$.

Let us now evaluate the running time.
Because every frontier is downward closed under $\prec$, it is uniquely
determined by the number of events it contains from each block and from
each thread.
Consequently, the number of frontiers is bounded by
$|V| \le (n+1)^{\beta}$, where $\beta = \min(k+1, |\threads|)$.
Likewise, $|E| \le \beta \cdot (n+1)^{\beta}$ since
the outdegree of every node is at most $\beta$.
For each frontier $X$, there are at most $\beta$ candidate successors
$Y=X\uplus\set{e}$, and after linear-time preprocessing of $\tr$ (to index
$\prec$-predecessors and reads-from obligations), feasibility of each
candidate can be checked in $O(1)$ time.
Hence the total time spent generating edges is $O(|E|)$, i.e.,
$O(\beta\cdot (n+1)^{\beta})$ for a fixed decomposition.
Likewise, the time to compute the states at each node in the graph
can be upper bounded by $O(\beta\cdot (n+1)^{\beta})$.

Finally, enumerating all $\binom{n}{k} \in O(n^k)$ choices of cutpoints
yields an overall running time of
$O\!\left(m \cdot n^k \cdot \beta \cdot (n+1)^{\beta}\right)$.
This completes the proof.
\end{proof}

Observe that, for the case of pre-image,
the membership problem is in FPT in the parameter $k$ (\thmref{pre-image-kslice-frontier}).
In contrast, the algorithm in~\thmref{post-image-kslice-frontier} does not yield such an FPT algorithm
for post-image.
In the following, we show that this is unavoidable;
we show that the problem is W[1]-hard in the parameter $k$
and thus, under the exponential time hypothesis (ETH), is not in the class FPT:

\begin{restatable}{theorem}{postImageETHHardness}
\thmlabel{post-image-eth-hardness}
The problem of checking,
for given run $\tr \in \labs^*$, language $L\subseteq \labs^*$
and $k \in \natsp$, if $\tr \in \post{L}{\kslice{k}}$,
is W[1]-hard in the parameter $k$.
\end{restatable}




\section{Discussion}
\seclabel{discussion}


\subsection{Parameterizing trace equivalences}
\seclabel{kmaz}

In~\secref{stacking-slices}, we argued how $k$-sliced reorderings
provide a natural way to parameterize the expressive power of
reads-from equivalence.
\emph{Are there meaningful ways in which one can parametrize
other known equivalences? Do they also yield algorithmic benefits in the context
of predictive monitoring a la sliced reorderings~\corref{complexity-closure-membership}?}
Here we entertain these questions in the context of trace equivalence $\mazeq$.
While one may be able to design bespoke parameteric versions of trace equivalence,
a natural parameterization can be obtained by taking from a closer look at 
the swap-based characterization of trace equivalence, and bounding
the number of swaps with a parameter:

\begin{definition}[$k$-Mazurkiewicz reorderings]
Let $\tr$ and $\rho$ be concurrent program runs, and let $k \in \natsp$.
We say that $\rho$ is a $k$-Mazurkiewicz reordering of $\tr$,
denoted $\tr \kmaz{k} \rho$ if $\rho$ can be obtained from $\tr$ by $\leq k$
successive swaps of neighboring independent events determined by $\indrel$.
\end{definition}

As an example, in \figref{k-slice-gradation}, $\tr^{\sf int}$ can be obtained
by $(k+1)\cdot(k+1)$ swaps starting from $\tr^{\sf seq}$.
As a result, $\tr^{\sf int} \kmaz{(k+1)^2} \tr^{\sf seq}$.

\myparagraph{Reflexivity, symmetry and transitivity}
Reflexivity of $\kmaz{k}$ follows because one can chose a swap sequence of length $1$.
Symmetry follows because each individual swap is reversible.
Finally, $\kmaz{k}$ is not transitive because one may need at least $k_1+k_2$
swaps to reach from $\tr$ to $\gamma$ if one can reach $\rho$ from $\tr$ with $k_1 > 0$ swaps
and $\gamma$ from $\rho$ with $k_2 > 0$ swaps.

\begin{proposition}
For every $k$, $\kmaz{k}$ is reflexive and symmetric but not transitive.
\end{proposition}

\myparagraph{Gradation of expressiveness and limit}
As with $k$-sliced reorderings, the above parameterization also observes
strict increase in expressiveness  on increasing the value of the parameter, 
Further, they reach the full trace equivalence in the limit, and
thus remain strictly less expressive than $\rfeq$.

\begin{proposition}
For every $k$, $\kmaz{k} \, \subsetneq \,  \kmaz{k+1}$. Further, $\big(\bigcup_{k \geq 1} \kmaz{k} \big)\, = \,  \mazeq$.
\end{proposition}


\myparagraph{Comparison with sliced reorderings}
Recall that trace equivalence and sliced reorderings are incomparable in their expressive power
(\thmref{grain-subsumption}); the relationship between $\kmaz{k}$ 
and $\kslice{k}$ is that of subsumption: 

\begin{restatable}{proposition}{kmazVskslice}
\proplabel{k-maz-vs-k-slice}
For every $k$, $\kmaz{k} \, \subsetneq \, \kslice{k}$.
\end{restatable}

\begin{proof}
The non-inclusion is easy and follows from \thmref{grain-subsumption}.
Here we focus on the inclusion $\kmaz{k} \subseteq \kslice{k}$.
That is, we prove that $\tr \kmaz{k} \rho \implies \tr \kslice{k} \rho$.


If each of the swaps in the sequence from $\tr$ to $\rho$
occur in well-separated parts of the trace,
the $k{+}1$ slices can be seen directly:
the first slice collects all events up to the first swapped pair and includes
the later event of that pair; the next slice starts from the earlier event of
the first pair and extends to the position just before the second pair,
and so on.
When the swaps are not well-separated—e.g., when the same event moves
left across many others—this simple construction no longer suffices.
Nevertheless, the underlying intuition remains:
if $\rho$ can be obtained from $\tr$ by at most $k$ adjacent
swaps of independent events, then the permutation of events between
$\tr$ and $\rho$ is ``almost increasing''—only a few pairs of events
have reversed their order.
Each such reversal can be absorbed into a slice boundary,
so that concatenating these slices in order yields $\rho$.
We now make this argument precise.

\smallskip
Let $\tr = e_1 e_2 \cdots e_n$ and $\rho = f_1 f_2 \cdots f_n$
be two executions such that $\tr \kmaz{k} \rho$.
Every swap in a $\kmaz{k}$ sequence exchanges a pair of \emph{independent}
events, hence preserves both program order and reads-from edges.
Consequently, $\tr$ and $\rho$ are $\rfeq$-equivalent.

For each event $f_i$ in $\rho$, let $\pi(i)$ denote its position in~$\tr$,
so that $\pi : \set{1, \ldots, n} \to \set{1, \ldots, n}$ is a permutation satisfying
$\tr = e_1 e_2 \cdots e_n$ and $\rho = e_{\pi(1)} e_{\pi(2)} \cdots e_{\pi(n)}$.
We will say that a pair $(i,j)$ with $i<j$ is called an \emph{inversion} if
$\pi(i)>\pi(j)$, and will
denote by $\Inv(\pi)$, the set of all such pairs.
It is well-known that the minimal number of adjacent swaps required to
realize $\pi$ equals $|\Inv(\pi)|$.  Since $\tr \kmaz{k} \rho$
admits a sequence of at most $k$ swaps, we must have $|\Inv(\pi)|\le k$.

Next, consider the \emph{drops} of~$\pi$, i.e.,
positions $i$ where $\pi(i)>\pi(i{+}1)$.
Let $D=\setpred{i\in \set{1, \ldots, n-1}}{\pi(i)>\pi(i{+}1)}$
be the set of all drops of $\pi$.
Each drop contributes at least one inversion—the pair $(i,i{+}1)$,
and thus $|D| \leq |\Inv(\pi)|$.
Let $0=i_0 < i_1 < \cdots < i_r=n$ be
an enumeration of the set  $D \uplus \set{0, n}$, such that
each block $(i_{s-1}{+}1,\dots,i_s)$
is maximal within $\pi$ and is strictly increasing.
Observe that $r-1 = |D| \le |\Inv(\pi)| \le k$.

For each $1\le s\le r$, define the subsequence
\[
P_s = f_{i_{s-1}+1} f_{i_{s-1}+2} \cdots f_{i_s}.
\]
Because $\pi$ is increasing on every block, the events within each $P_s$
appear in $\tr$ in the same relative order as in $\rho$.
Hence every $P_s$ is a subsequence of $\tr$ that preserves $\po{}$ and $\rf{}$.
Moreover,
\[
\rho = P_1 \cdot P_2 \cdots P_r.
\]
Thus, $\rho$ can be obtained from $\tr$ by serially composing
at most $r\le k{+}1$ such slices.
Thus, $\tr \kslice{k} \rho$ since 
$\tr$ and $\rho$ are $\rfeq$-equivalent.
\end{proof}

\myparagraph{Image under $\kmaz{k}$} We final consider the task of predictive monitoring
under parametrized trace equivalence. Much like sliced reorderings,
the image of a regular language under $\kmaz{k}$ is also regular,
giving us a constant space  linear time streaming algorithm for predictive monitoring.
This is in sharp contrast to the case of the full equivalence $\mazeq$ for which
such an algorithm is unlikely under arbitrary regular specifications~\cite{ang2023predictive}.

\begin{restatable}{proposition}{kmazImageRegular}
\proplabel{kmaz-image-regular}
For every $k$ and regular language $L$, 
$\pre{L}{\kmaz{k}} = \post{L}{\kmaz{k}}$ is regular.
\end{restatable}

\begin{proof}
Since $\kmaz{k}$ is symmetric, it follows that
$\pre{L}{\kmaz{k}} = \post{L}{\kmaz{k}}$ for any language $L$.
The image of a regular language $L$ under $\kmaz{k}$ can be shown to be regular as follows.
First, given a DFA $\aut = (\states, \init, \trns, \accpt)$ for $L$, 
one can construct an NFA $\aut' = (\states', \init' = \init, \trns', \accpt' = \accpt)$
for $\post{L}{\kmaz{1}}$ by augmenting states of $\aut$ as
$\states' = \states \uplus \set{(q,b,a)}_{q \in Q, (a,b) \in \indrel}$.
Apart from the transitions of $\aut$,
the automaton $\aut'$ can transition non-deterministically from a state
$q$ on reading $b$ to $(q, b, a)$ and wait to read $a$ next (and fail otherwise).
Next, observe that $\kmaz{k+1} = \kmaz{k} \circ \kmaz{1}$
and thus $\pre{L}{\kmaz{k+1}} = \pre{\pre{L}{\kmaz{k}}}{\kmaz{1}}$
is also regular.
\end{proof}

\subsection{Going beyond trace equivalence}
\seclabel{beyond-trace-equivalence}

Readers may have observed that, while in the limit, sliced reorderings surpass
the expressivity of trace equivalence ($\mazeq \subsetneq \rfeq$), 
each fixed parameter version remains incomparable with $\mazeq$ 
(see \thmref{mslice-expressivity}).
While the goal of a reordering relation that surpasses trace
equivalence in expressive power appears desirable, it is in direct
conflict with the orthogonal goal of a reordering relation that yields an efficient 
(read `constant space, streaming, linear time')
predictive monitoring algorithm against arbitrary regular specifications.
Indeed, for a very simple language $L$ (see \thmref{beyond-trace}),
the closure of $L$ against $\mazeq$ is not even context-free
and cannot admit a sub-linear membership check.
This should,
in fact, hold for every reordering relation $R$ that includes $\mazeq$. Formally:

\begin{restatable}{theorem}{BeyondTraceHardness}
\thmlabel{beyond-trace}
Let $a = \ev{T, \wt(x)}, b = \ev{T, \wt(y)}, \bar{a} = \ev{\bar{T}, 
\bar{x}}, \bar{b} = \ev{\bar{T}, \bar{y}} \in \labs$ for some
distinct $T, \bar{T} \in \threads$ and $x,y,\bar{x},\bar{y} \in \mems$. 
Let $L$ be the regular language $L = (ab + \bar{a}\bar{b})^*$.
Let $R \subseteq \labs^* \times \labs^*$ be an arbitrary sound reordering relation such that
$\mazeq \subseteq R$. 
Any  one pass algorithm for membership in 
$\pre{L}{R}$ and $\post{L}{R}$ must use linear space in the worst case.
Further, the time $T(n)$ and space $S(n)$ usage of any multi-pass algorithm
for solving this problem must satisfy $S(n) \cdot T(n) \in \Omega(n^2)$,
where $|\sigma| = n$.
\end{restatable}

In other words, the absence of subsumption of trace equivalence is, in some sense, 
inevitable if a general framework for deriving efficient predictive monitoring algorithms
against arbitrary regular specifications is desirable.


\section{Related Work}
\seclabel{related}

Predictive monitoring has emerged as a principled means to enhance coverage
of dynamic testing of concurrent programs.
However, the underlying algorithms have largely been one off, mostly catering to 
the prediction of data races and deadlocks and often trying to 
beat the theoretical or practical predictive power of previously proposed
algorithms~\cite{Smaragdakis12,Kini17,Roemer18,Pavlogiannis2020,Mathur21,OSR2024,genc2019,Tunc2023},
while still retaining polynomial time.
In these works, reads-from equivalence, for which complete algorithms are unlikely to be 
tractable~\cite{Mathur21,FarzanMathur2024}, forms the theoretical limit
for predictive power, given that this captures the maximum amount of information
without any knowledge of the program and with knowledge of only addresses of shared memory locations
being accessed in the execution being monitored.
In contrast, ideas borrowed from trace equivalence have
been the guiding principle for more lightweight (and often constant space, streaming)
algorithms for the prediction of specific properties such as data races~\cite{Elmas07,Mathur18}, 
atomicity violations~\cite{FM08CAV,Farzan2006,Mathur2020}, 
pattern languages~\cite{ang2023predictive} or even for detecting robustness violations
for weak memory consistency~\cite{RobustSan2025}.
Motivated by this trend, in this work we ask the question --- \emph{can we design a 
general purpose framework that can be instantiated for prediction against a larger
class of properties?}

Trace equivalence has, in fact, been studied extensively
with regard to this question, and the most prominent result in this space has been
Ochmanski's characterization of star connected regular expressions
and (star-connected) languages that coincide with them~\cite{Ochmanski85}
as the set of those regular languages whose closure under trace equivalence remains 
regular. 
As we point out, even simple languages cease to be in this fragment,
though can encode meaningful classes of bugs in certain cases, and moreover, the
task of determining if a specification falls in this class is undecidable
in general~\cite{undecTraceClosureRegularity1992}.
Boujjani et al~\cite{Bouajjani2007APC} propose the sub-fragment alphabetic pattern constraints
(APCs) as the maximum level in the  Straubing-Th{\'e}rien hierarchy~\cite{STRAUBING198553,STRAUBING1981137,THERIEN1981195} that belong to the class of star-connected languages.
Indeed, as we show in \secref{beyond-trace-equivalence}, any reordering relation (whether equivalence or not)
that coarsens trace equivalence suffers from the downside that it will map
some regular language to a non-regular one under pre and post image.
For arbitrary languages (i.e., those that are not characterized by star connected languages),
the predictive membership problem admits an algorithm that grows with the number
of threads~\cite{Bertoni1989}, and moreover this bound is tight~\cite{ang2023predictive}.

Besides the algorithmic penalty that trace equivalence induces for predictive membership against
arbitrary regular languages, it also significantly limits the expressive power owing
to inability to flip the order of neighboring conflicting memory accesses.
Grain and scattered grain equivalence~\cite{FarzanMathur2024} lift commutativity reasoning
a la trace equivalence to the case of sets of events, but remain less expressive
than reads-from equivalence.
In~\cite{AngMathurPrefixes2024}, prefixes were introduced to enhance trace equivalence,
though they do not allow for arbitrary reorderings of write events.
Our work is in fact inspired from the notion of prefixes, and more precisely
by the observation that prefixes implicitly reason, though in a very limited manner,
about reversals of conflicting events.
In our work, we show that sliced reorderings, and in particular $k$-sliced reorderings 
systematically generalize this
idea and can simulate reads-from equivalence in the limit.
In similar vein, trace equivalence with observers~\cite{AronisDPORWithObservers2018}
moderately weaken trace equivalence by allowing to reorder write events
if they are not observed by any reads, and as such are subsumed by grain based reasoning.

Equivalences also play a central role in dynamic partial order reduction
based model checking, where a concrete notion of equivalence
is often used to determine if the model checker only explores a few 
(or exactly once in case of \emph{optimal} DPOR) representative runs from each
class, and thus coarser equivalences imply fewer exploration~\cite{Kokologiannakis2022,Abdulla2019,SourceSetsDPORAbdulla2017}.
Reads-value equivalence further weakens reads-from equivalence,
in the case when events are annotated with values that variables observe
and has also been employed in a DPOR setting~\cite{ChatterjeePT19,Agarwal2021}.
Nevertheless, like with the case of reads-from equivalence~\cite{Mathur2020b}, 
the predictive monitoring problem remains intractable even for data race prediction
because the underlying consistency checking problem remains hard~\cite{Gibbons1997}.
Approaches based on SMT solving, while theoretically sound, often fail to scale
to real world settings~\cite{Huang14,Kalhauge2018}.

Sliced reorderings offer a complementary take to the various notions
of equivalences and coarsenings that have been proposed in the past,
and reconcile the theoretical hardness of reads-from equivalence
with the desirable goal of obtaining prediction algorithms whose space
complexity can be tamed with a parameter.
Of course, like previous works, we assume that parameters such as number of threads
and memory locations do not grow with the input size (length of the execution),
though the dependence on these parameters can be significant, even for otherwise
fast trace-based algorithms~\cite{Kulkarni2021}.
The notion of reversal-distance of~\cite{Mathur2020b} is another example of
parametrization, but was specifically designed to cater for data races and may not
admit an FPT algorithm for arbitrary regular specifications.
Further, unlike sliced reorderings, the maximum possible value of this parameter
can range upto quadratic in the length of the execution (similar to the parametrization of trace equivalence
we discuss in \secref{kmaz}).
Another form of parametrization that has appeared in the literature is the use of
\emph{windowing}, where one partitions the input run into disjoint windows whose
length is parameterized~\cite{Huang14,Kalhauge2018}, and one employs complete reads-from based
reasoning within each window. 
Unlike sliced reorderings, the notion of reorderings that such a windowing style 
parametrization induces cannot relate executions where far away events are flipped,
and this severely reduces the predictive power of algorithms resorting to such a parametrization~\cite{Kini17,Tunc2023}.
Based on the examples we discuss in our paper, it appears that sliced reorderings
may be able to accurately augment preemption bounding based concurrency
testing and model checking approaches~\cite{preemptionBoundingQadeer2005,preemptionDPORMarmanis2023}, 
where the idea is to only limit exploration
to runs with bounded pre-emptions under the hypothesis that small number of context switches
suffice to expose most bugs; sliced reorderings can allow for exploration
with a smaller bound, if one also additionally employ sliced-reordering based
predictive reasoning.

\section{Conclusion and Future Work}
\seclabel{conclusion}
We propose $\kslice{k}$ as a new parametric {\em predictor} that can be used in predictive monitoring of concurrent programs for regular specification. For any constant $k$ and any regular specification $S$, there exists a streaming-style constant space monitor that, while reading an input program run $\sigma$, {\em soundly} predicts if another program run $\rho$ such that $\sigma \kslice{k} \rho$ satisfies the specification $S$.

\begin{wrapfigure}[9]{r}{0.48\textwidth}\vspace{-10pt}
\includegraphics[scale=0.6]{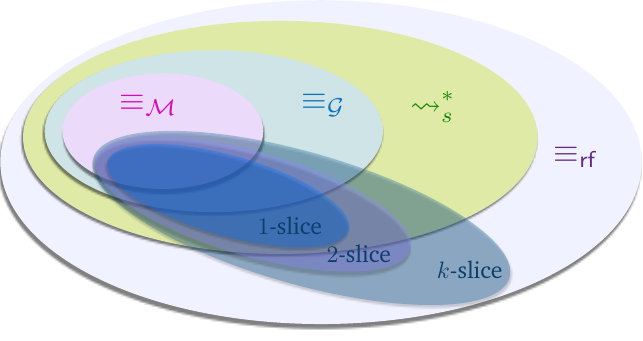}\vspace{-10pt}
\end{wrapfigure}
The Venn diagram on the right compares the expressive power of the existing sound predictors against the newly proposed ones in this paper. In particular, sliced reorderings ($\sliceq$) are strictly better event-based commutativity ($\mazeq$), grain-based commutativity ($\graineq$), but all three are strictly less expressive than {\sf rf}-equivalence ($\rfeq$). $k$-slice reorderings are incomparable against all notions, other than being a strict subset of ($\rfeq$). But, in the limit (not illustrated), they are equivalent to $\rfeq$. It is worth mentioning that since predictive monitors naturally compose (disjunctively), one always has the option of exploiting parallelism and monitoring the run through several monitors simultaneously. Hence, our newly proposed technique complements all other existing techniques in the literature.

We have presented theoretical results that guarantee the independence of the monitor memory from the length of the input run. However, the dependency on the number of threads and the number of shared variables can become a bottleneck in practice for programs with a large number of those entities. Since the design of our monitor is naturally nondeterministic, it would be interesting to explore generic techniques, for instance {\em antichain} methods, to optimize generic monitors rather than having to resort to hand-optimize specific monitors for specific properties. 


\bibliographystyle{ACM-Reference-Format}
\bibliography{references}

\clearpage

\appendix


\section{Proofs from~\secref{slice}}

\LinSpaceHardnessSliceStar*

\begin{proof}
At a high level, we show that the problem of interest has a one-pass constant space
reduction from the problem of membership in the following language, which is known to admit a
linear space lower bound in the streaming setting (here $n\in \nats$):
\[
  L_n
  = \setpred{a_1a_2\cdots a_n \# b_1 b_2 \cdots b_n}
           {\forall i \leq n,\ a_i, b_i \in \set{0,1},\ a_i = b_i}.
\]
The reduction is inspired by an analogous result in~\cite{FarzanMathur2024}
and constructs a run $\sigma$ (of length $O(n)$) starting from a word
$w = \bar{a}\#\bar{b} \in \set{0,1}^*\#\set{0,1}^*$ in a one-pass streaming
fashion using only constant memory such that $w \in L_n$ iff
no repeated sliced reordering of $\sigma$ inverts a certain pair of $u$-events.
Equivalently, $w \notin L_n$ iff there exists a run $\rho$ with $\sliceqto{\sigma}{\rho}$
in which these two events appear in inverted order.

\myparagraph{Construction}
We start from a word $w = a_1a_2\ldots a_n \# b_1 b_2 \ldots b_n$
and construct a run $\sigma$ as follows.
The run $\sigma$ uses two threads $\threads = \set{t_1, t_2}$ and six memory locations
$\mems = \set{x_0, x_1, y_0, y_1, c, u}$.
It has the form
\[
  \sigma \;=\; \pi_1 \cdot \pi_2 \cdots \pi_n \cdot \kappa
              \cdot \eta_1 \cdot \eta_2 \cdots \eta_n \cdot \delta.
\]

The fragments $\pi_i$ encode the prefix $\bar a$, and contain only events
of $t_1$:
\begin{align*}
  \pi_1
    &= \ev{t_1,\wt(x_{\neg a_1})}
       \cdot \ev{t_1,\wt(c)}
       \cdot \ev{t_1,\wt(x_{a_1})},\\[1ex]
  \pi_i
    &= \ev{t_1,\wt(y_{a_i})}
       \cdot \ev{t_1,\rd(c)}
       \cdot \ev{t_1,\wt(c)}
       \cdot \ev{t_1,\rd(y_{a_i})}
       \qquad (2 \le i \le n).
\end{align*}

The fragments $\eta_i$ encode the suffix $\bar b$, and contain only events
of $t_2$:
\begin{align*}
  \eta_1
    &= \ev{t_2,\rd(x_{b_1})}
       \cdot \ev{t_2,\wt(c)},\\[1ex]
  \eta_i
    &= \ev{t_2,\wt(y_{b_i})}
       \cdot \ev{t_2,\wt(c)}
       \qquad (2 \le i \le n).
\end{align*}

The $u$-blocks are
\[
  \kappa
    = \ev{t_1,\wt(u)} \cdot \ev{t_1,\rd(c)} \cdot \ev{t_1,\rd(u)}
  \qquad\text{and}\qquad
  \delta
    = \ev{t_2,\wt(u)} \cdot \ev{t_2,\rd(u)}.
\]
Let $e_1$ be the unique event $\ev{t_1,\rd(u)}$ in $\kappa$, and
let $e_2$ be the unique event $\ev{t_2,\wt(u)}$ in $\delta$.
The transducer that, on input $w$, outputs $\sigma$ simply streams $w$
from left to right and, for each symbol $a_i$ or $b_i$, appends the
corresponding fragment above.
It clearly runs in one pass and uses $O(1)$ working space.

The reads-from mapping $\rf{\sigma}$ is defined in the obvious way:
every read has a unique preceding write on the same location, and
no additional writes to that location appear between them.
In particular, there are no writes to $u$ other than those in
$\kappa$ and $\delta$.

\myparagraph{Case $\bar a = \bar b$}
Assume first that $w \in L_n$, i.e.\ $a_i = b_i$ for all $1 \le i \le n$.
By essentially the same reasoning as in the proof of the reads-from
lower bound of~\cite{FarzanMathur2024}, one shows that the chain
of $x$-, $y$- and $c$-events enforces a causal ordering between $e_1$
and $e_2$: more precisely, in the partial order induced by
$\po{\sigma}$ and $\rf{\sigma}$, we have $e_1$ before $e_2$.
Moreover, this causal order is invariant under reads-from equivalence:
for every run $\rho$ with $\po{\rho} = \po{\sigma}$ and
$\rf{\rho} = \rf{\sigma}$, the corresponding events
$e_1$ and $e_2$ in $\rho$ are still ordered the same way.

By definition of sliced and repeated sliced reordering, every run
$\rho$ with $\sliceqto{\sigma}{\rho}$ is obtained from $\sigma$
by a finite sequence of single sliced reordering steps, and each
step preserves $\rf{}$.
Hence every such $\rho$ is reads-from equivalent to $\sigma$,
and therefore $e_1$ still appears before $e_2$ in $\rho$.
In particular, there is no repeated sliced reordering $\rho$ of $\sigma$
in which $e_2$ precedes $e_1$.

\begin{figure*}[htbp!]
\includegraphics[width=\textwidth]{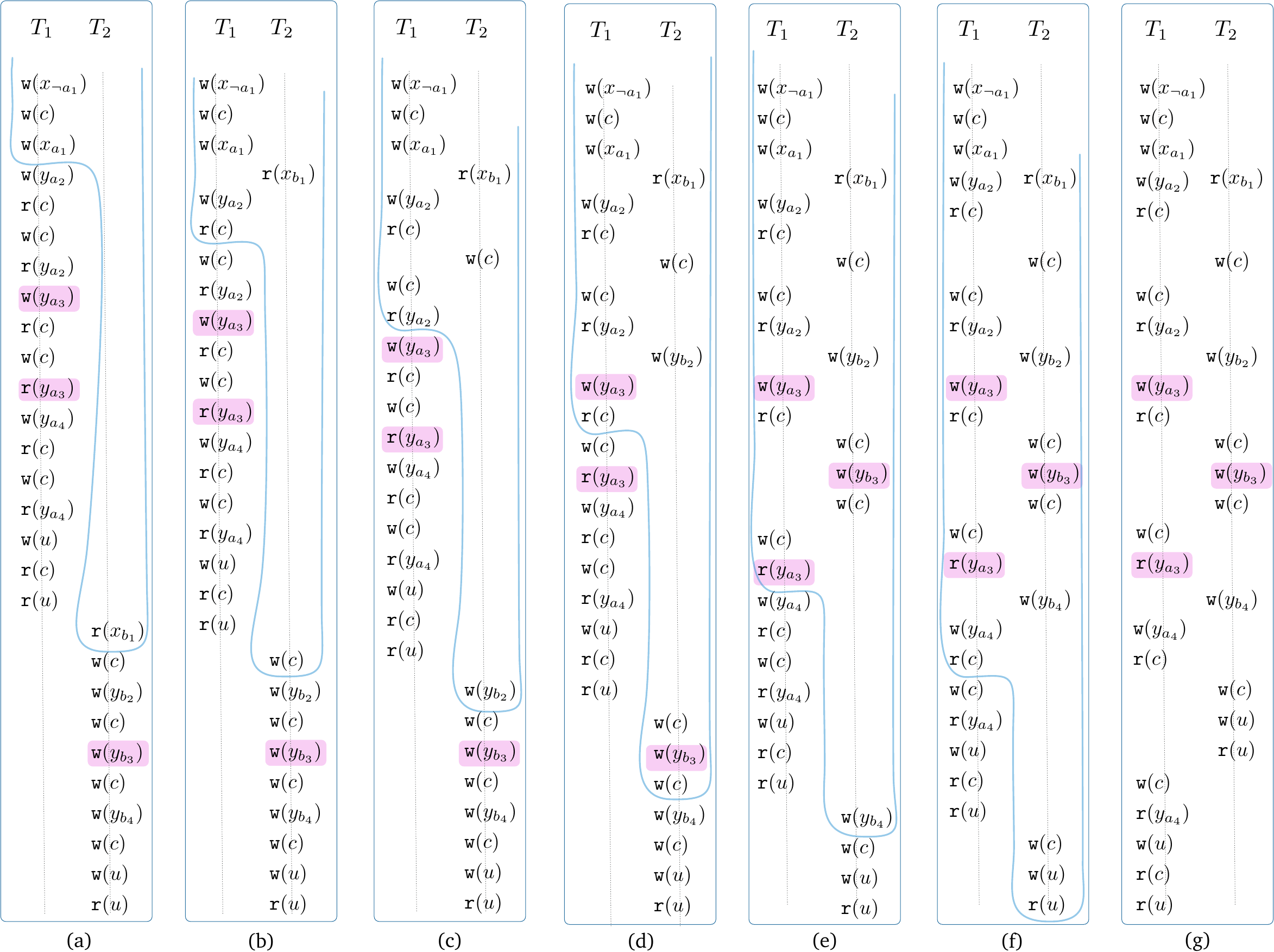}
\caption{Sequence of slice reorderings to transform $\sigma$ (leftmost) to $\rho$ (rightmost) in which the order of events $\ev{t_1, \wt(u)}$ and $\ev{t_2, \wt(u)}$ is flipped.}
\figlabel{repeated-slice-reordering-hardness}
\end{figure*}

\myparagraph{Case $\bar a \neq \bar b$}
Now assume $w \notin L_n$.
Let $i$ be the smallest index such that $a_i \neq b_i$.
We show that there exists a run $\rho$ with $\sliceqto{\sigma}{\rho}$
in which $e_2$ appears before $e_1$.
We sketch the construction; an example for $n=4$ and $i=3$
is depicted in \figref{repeated-slice-reordering-hardness}.

We construct a sequence of runs
\[
  \gamma_0, \gamma_1, \dots, \gamma_{2n}
\]
such that
\[
  \gamma_0 = \sigma,\quad
  \gamma_{2n} = \rho,\quad\text{and}\quad
  \gamma_j \sliceto{}{} \gamma_{j+1}
  \text{ for all } 0 \le j < 2n.
\]

\emph{For each $k < i$}, we use two sliced reordering steps to ``bubble''
the fragment $\eta_k$ upwards:
first slice-move the first event of $\eta_k$ to immediately after the
last event of $\pi_k$, and then slice-move the corresponding $\wt(c)$
of $\eta_k$ to the appropriate position after a $\rd(c)$ in $\pi_{k+1}$.
In each of these steps the moved event is pushed \emph{later} in the
total order, and never across a read on the same location, so the
last-write-before-read on that location is preserved.
Hence $\rf{}$ is unchanged, and both steps are valid sliced reorderings.
Repeating this for all $k < i$ yields $\gamma_{2(i-1)}$.

\emph{At index $i$}, we perform one slightly larger slice:
we take the entire fragment $\eta_i$ together with the $\wt(c)$
of $\eta_{i-1}$ and move this block so that it sits between the
events $\ev{t_1,\wt(y_{a_i})}$ and $\ev{t_1,\rd(y_{a_i})}$ of $\pi_i$.
Because $a_i \neq b_i$, we have $y_{a_i} \neq y_{b_i}$,
so inserting $\wt(y_{b_i})$ between $\wt(y_{a_i})$ and $\rd(y_{a_i})$
does not change the last-write-before-read on $y_{a_i}$.
Similarly, the $c$-events are only shifted within a region where they
do not cross any read that would change their source.
Thus $\rf{}$ is again preserved, and this is a valid sliced reordering,
yielding $\gamma_{2(i-1)+1}$.

\emph{For each $k > i$}, we resume the same two-step bubbling pattern
as for $k < i$, successively moving the events of $\eta_k$ upward
and interleaving them with the $\pi_k$’s.
The argument that each move preserves $\rf{}$ is identical to the
$k < i$ case.
After processing all $k$, we reach $\gamma_{2n-1}$, in which the
$x$-, $y$-, and $c$-events of $t_2$ have been interleaved with
those of $t_1$ in a controlled fashion, while $\rf{}$ remains the same
as in $\sigma$.

Finally, from $\gamma_{2n-1}$ we perform one more sliced reordering
step: we take the $u$-block $\delta = \ev{t_2,\wt(u)}\cdot\ev{t_2,\rd(u)}$
as a slice, and move it so that it appears immediately before the
$u$-block $\kappa$ of $t_1$.
Since there are no writes to $u$ other than in $\kappa$ and $\delta$,
and the two reads of $u$ still see exactly their original writes,
this step also preserves $\rf{}$.
We obtain a run $\rho = \gamma_{2n}$ in which $e_2$ precedes $e_1$
in the total order, and $\sliceqto{\sigma}{\rho}$ holds by construction.

\myparagraph{Conclusion}
We have described a one-pass constant-space transducer that maps
an input word $w$ to a run $\sigma$ such that:
\begin{itemize}
  \item if $w \in L_n$, then in every run $\rho$ with $\sliceqto{\sigma}{\rho}$
        the events $e_1$ and $e_2$ appear in the same order, and
  \item if $w \notin L_n$, then there exists a run $\rho$ with
        $\sliceqto{\sigma}{\rho}$ in which $e_2$ appears before $e_1$.
\end{itemize}
Thus any streaming algorithm that decides whether such a $\rho$ exists
must use $\Omega(n)$ space, by the linear space lower bound for recognizing $L_n$.
This completes the proof.

Next, since $L_n$ admits the space-time tradeoff bound (i.e., the produce of time and space usage of any algorithm for checking membership in $L_n$ be atleast $\Omega(n^2)$), we also get the same bound for our problem.
\end{proof}


\section{Proofs from \secref{stacking-slices}}

Let us now proceed towards the proof of \thmref{sheight}.
Before we do that, we focus on a simple observation:

\depthDrops*

\begin{proof}
Since $\tr \rfeq \rho$ and $|\tr|=|\rho|=n$, we can (and do) fix the permutation
$\pi:[n]\to[n]$ such that the $i^{\text{th}}$ event of $\rho$ is exactly the
$\pi(i)^{\text{th}}$ event of $\tr$.

Recall that $\sheight{\tr}{\rho}$ is the least $k$ for which $\rho$ can be written as
\[
  \rho = \rho_1 \cdot \rho_2 \cdots \rho_k
\]
where each $\rho_j$ is a subsequence of $\tr$, and the $\rho_j$'s are pairwise disjoint
(i.e., they form a partition of the events of $\tr$).

\medskip\noindent
\textbf{Claim 1 (drops force boundaries).}
If $\rho=\rho_1\cdots \rho_k$ is such a $k$-slice decomposition, then for every
drop position $i\in D$ we must have that $i$ is a boundary between two consecutive slices.
Consequently, $k \ge |D|+1$.

\smallskip\noindent
\emph{Proof of Claim 1.}
Fix $i\in D$, so $\pi(i)>\pi(i+1)$.
Suppose for contradiction that the two consecutive events $\rho[i]$ and $\rho[i+1]$
belong to the same slice, say $\rho_j$.
Because $\rho_j$ is a subsequence of $\tr$, the order of its events in $\rho_j$
(and hence in $\rho$) must agree with their order in $\tr$.
Thus the position in $\tr$ of $\rho[i]$ must be strictly smaller than the position in $\tr$
of $\rho[i+1]$, i.e., $\pi(i)<\pi(i+1)$, contradicting $\pi(i)>\pi(i+1)$.
Hence $i$ must be a boundary between slices.
Since distinct drops are distinct boundaries, the number of slices is at least the number of
required boundaries plus one, i.e., $k\ge |D|+1$. \qed

\medskip\noindent
\textbf{Claim 2 (cutting at drops gives a valid decomposition).}
Let $D=\{d_1<d_2<\cdots<d_m\}$ where $m=|D|$, and set $d_0:=0$, $d_{m+1}:=n$.
For each $j\in\{1,\ldots,m+1\}$ define $\rho_j$ to be the contiguous block of $\rho$
\[
  \rho_j := \rho[d_{j-1}+1\,..\,d_j].
\]
Then each $\rho_j$ is a subsequence of $\tr$, the blocks are pairwise disjoint, and
$\rho=\rho_1\cdots\rho_{m+1}$. Hence $\sheight{\tr}{\rho}\le |D|+1$.

\smallskip\noindent
\emph{Proof of Claim 2.}
By construction the blocks $\rho_1,\ldots,\rho_{m+1}$ partition $\rho$, so they are disjoint
and concatenate to $\rho$.

It remains to show each $\rho_j$ is a subsequence of $\tr$.
Fix $j$, and consider any consecutive indices $p,p+1$ within the interval
$\{d_{j-1}+1,\ldots,d_j\}$.
By definition of $D$, there is no drop inside this interval; hence for all such $p$,
$\pi(p) < \pi(p+1)$.
By transitivity this implies that along the entire block we have a strictly increasing chain
\[
  \pi(d_{j-1}+1) < \pi(d_{j-1}+2) < \cdots < \pi(d_j).
\]
Therefore, if we look at $\tr$ and pick exactly the events at positions
$\pi(d_{j-1}+1),\pi(d_{j-1}+2),\ldots,\pi(d_j)$, they appear in $\tr$ in that same order,
and they are precisely the events of $\rho_j$ in order.
Thus $\rho_j$ is a subsequence of $\tr$. \qed

\medskip
Combining Claim 1 and Claim 2, we get
\[
  |D|+1 \le \sheight{\tr}{\rho} \le |D|+1,
\]
and hence $\sheight{\tr}{\rho}=|D|+1$.
\end{proof}

The \thmref{sheight} now follows:

\kSliceDepthLinearTime*

\begin{proof}
Follows from \propref{depth-drops}, and the fact that the number of drops can be computed in linear time, as well as reads-from equivalence (i.e., whether $\sigma \rfeq \rho$?) can be checked in linear time.
\end{proof}


\section{Detailed construction and proofs from \secref{algo}}
\applabel{app-algo}

This section serves as a companion for \secref{algo}
where we present 
present proofs of its correctness, and finally present the proof of the hardness result \thmref{post-image-liear-space-hardness}.


\subsection{Proofs from \secref{pre-image-slices}}
\applabel{app-algo-proofs}

\consistencyOfSlices*

\begin{proof}
($\Rightarrow$) 
Consider a pair $(e_i, e_j) \in \po{\tr}$ such that $e_i \in \tr_i$ and $e_j \in \tr_j$.
Since $\rho \rfeq \tr$, we must have $\po{\rho} = \po{\tr}$ and thus $(e_i, e_j) \in \trord{\rho}$. 
This means either they belong to the same subsequence (i.e, $i = j$), and if not,
then $e_i$ appears in an earlier partition, i.e., $i < j$.
Now consider $(e_i, e_j) \in \rf{\tr}$ with $e_i \in \tr_i$ and $e_j \in \tr_j$.
Clearly $(e_i, e_j) \in \rf{\rho}$ and thus $i \leq j$.
Further, consider a conflicting write $e'$ (with $\OpOf{e'} = \wt$ and $\MemOf{e'} = \MemOf{e_i}$)
belonging to subsequence $\tr_\ell$.
First, it cannot be that $i < \ell < j$, otherwise $e'$ intervenes $e_i$ and $e_j$ in $\rho$.
So we must have either $\ell \leq i$ or $j \leq \ell$.
In the former case, if additionally $\ell = i$, then again to ensure that $e'$ does not intervene $e_i$ and $e_j$,
$e'$ must appear before $e_i$ (since within $\tr_i$, the relative order of events does not change), i.e.,
$e' \trord{\tr} e_i$.
In the latter case, a similar reasoning tells us that if $\ell = j$, then 
$e_j \trord{\tr} e'$.

($\Leftarrow$)
First, by construction, $\events{\rho} = \events{\tr}$, so we only have to establish that the order of events in $\rho$
is in accordance to $\po{\tr}$ and $\rf{\tr}$.
Consider two events $e_i, e_j$ such that $(e_i, e_j \in) \po{\tr}$ with $e_i \in \tr_i, e_j \in \tr_j$. 
By condition (1), we have $i \leq j$. If $i = j$, then $e_i$ 
appears earlier than $e_j$ in $\tr_i = \tr_j$ and since the relative order of events does not 
change within $\tr_i$, we have $e_i \trord{\rho} e_j$ and thus $(e_i, e_j) \in \po{\rho}$. 
If $i < j$, then by construction of $\rho = \tr_1 \cdot \tr_2 \cdots \tr_{k+1}$, 
all events in $\tr_i$ appear before all events in $\tr_j$ in $\rho$, so 
$e_i \trord{\rho} e_j$ and $(e_i, e_j) \in \po{\rho}$.
Conversely, consider $(e_i, e_j) \in \po{\rho}$ with $e_i \in \tr_i, e_j \in \tr_j$. 
This means $e_i \trord{\rho} e_j$ and thus by construction of $\rho$, we must have $i \leq j$.
Also $\ThreadOf{e_i} = \ThreadOf{e_j}$ and thus either $(e_i, e_j) \in \po{\tr}$ or $(e_j, e_i) \in \po{\tr}$.
If $i = j$, then since $\ThreadOf{e_i} = \ThreadOf{e_j}$ and the relative order within 
$\tr_i$ is preserved from $\tr$, we have $(e_i, e_j) \in \po{\tr}$. 
If $i < j$, then we cannot have $(e_j, e_i) \in \po{\tr}$ as otherwise condition (1) will be violated, so we must have $(e_i, e_j) \in \po{\tr}$.

Let us now establish $\rf{\tr} = \rf{\rho}$.
Consider a read event $e_\rd \in \events{\tr_i}$. 
We need to show that $\rf{\rho}(e_\rd) = \rf{\tr}(e_\rd)$.
First, consider the case when $\rf{\tr}(e_\rd)$ is not defined.
By condition (2a), for every write event $e'_\wt$ with $\OpOf{e'_\wt} = \wt$ and $\MemOf{e'_\wt} = \MemOf{e_\rd}$ and $e'_\wt \in \events{\tr_\ell}$, we have $\ell \geq i$. 
In $\rho$, since all events from $\tr_j$ with $j < i$ appear before all events from $\tr_i$, and condition (2a) ensures no such writes exist in $\tr_j$ for $j < i$, there is also no write to $\MemOf{e_\rd}$ before $e_\rd$ in $\rho$. Therefore, $\rf{\rho}(e_\rd)$ is also undefined.
Now consider the case when $\rf{\tr}(e_\rd) = e_\wt$ is defined with $e_\wt \in \events{\tr_j}$.
By condition (2b), we have $j \leq i$.
We need to show that $e_\wt$ is the last write to $\MemOf{e_\rd}$ before $e_\rd$ in $\rho$.
First, if $j < i$, then $e_\wt \trord{\rho} e_\rd$ by construction and if $j = i$ then since
order of events inside a given subsequence does not change, yet again we have $e_\wt \trord{\rho} e_\rd$.
Now consider any other write $e'_\wt$ such that
$e_\wt \neq e'_\wt, \OpOf{e'_\wt} = \wt, \MemOf{e'_\wt} = \MemOf{e_\rd}$
with $e'_\wt \in \events{\tr_\ell}$.
We have by condition (2b) that $\ell \leq i \lor \ell \geq j$.
If $\ell < j$ or $\ell > i$, then $e'_\wt$ cannot be in between $e_\wt$ and $e_\rd$ in $\rho$.
So we have two remaining cases:
\begin{itemize} 
    \item If $\ell = i$: By condition (2b)(ii), $e_\rd \trord{\tr} e'_\wt$. 
    Since the relative order within $\tr_i$ is preserved in $\rho$, we have
    $e_\rd \trord{\rho} e'_\wt$ and thus $e'_\wt$ cannot be in between $e_\wt$ and $e_\rd$ in $\rho$.

    \item If $\ell = j$: By condition (2b)(iii), $e'_\wt \trord{\tr} e_\wt$. 
    Since the relative order within $\tr_j$ is preserved in $\rho$, 
    we have  $e'_\wt \trord{\rho} e_\wt$ and thus $e'_\wt$ cannot be in between $e_\wt$ and $e_\rd$ in $\rho$.s
\end{itemize}
Therefore, $e_\wt$ is indeed the last write to $\MemOf{e_\rd}$ before $e_\rd$ in $\rho$, so 
$\rf{\rho}(e_\rd) = e_w = \rf{\tr}(e_\rd)$.
\end{proof}

\begin{proposition}[Prefix-closedness of consistency]
\proplabel{prefix-closed-consistency}
If $\hat{\sigma} \in \hat{L}_{\cnst}$, then for every prefix
$\hat{\gamma} \preceq \hat{\sigma}$ we also have
$\hat{\gamma} \in \hat{L}_{\cnst}$.
\end{proposition}

\begin{proof}
Let $\hat{\sigma} \in \hat{L}_{\cnst}$ and let
$\hat{\gamma} \preceq \hat{\sigma}$ be a prefix.
Let $\sigma = h(\hat{\sigma})$ and $\gamma = h(\hat{\gamma})$.
By definition of $\hat{L}_{\cnst}$, we have
\[
\sigma \rfeq \rho := h(\proj{\hat{\sigma}}{1}) \cdots h(\proj{\hat{\sigma}}{k+1}).
\]

Define $\gamma_i := h(\proj{\hat{\gamma}}{i})$ and
$\rho_\gamma := \gamma_1 \cdots \gamma_{k+1}$.
We show that $\gamma \rfeq \rho_\gamma$ by verifying the two conditions
of \lemref{consistency}.

\emph{Program order.}
Let $(e,f) \in \po{\gamma}$. Then $(e,f) \in \po{\sigma}$ since $\gamma$
is a prefix of $\sigma$.
If $e \in \gamma_i$ and $f \in \gamma_j$, then also
$e \in \proj{\hat{\sigma}}{i}$ and $f \in \proj{\hat{\sigma}}{j}$.
Since $\sigma \rfeq \rho$, \lemref{consistency}(1) yields $i \le j$.

\emph{Reads-from.}
Let $e_{\rd} \in \gamma_i$ be a read event.

If $e_{\rd}$ is orphan in $\gamma$, then it is also orphan in $\sigma$,
since orphanhood depends only on preceding events.
By \lemref{consistency}(2a) for $\sigma$, every write to the same memory
location lies in a slice $\ell \ge i$, and hence the same holds for
writes occurring in $\gamma$.

If $e_{\rd}$ is non-orphan in $\gamma$, let $e_{\wt}$ be its rf-source
in $\gamma$, with $e_{\wt} \in \gamma_j$.
Then $e_{\wt}$ is also the rf-source of $e_{\rd}$ in $\sigma$.
Applying \lemref{consistency}(2b) to $\sigma$ and restricting attention
to events occurring in $\gamma$ yields the same inequalities and
order constraints for $\gamma$.

Thus both conditions of \lemref{consistency} hold for $\gamma$, and hence
$\gamma \rfeq \rho_\gamma$, i.e.\ $\hat{\gamma} \in \hat{L}_{\cnst}$.
\end{proof}

\consistencyRegular*

\begin{proof}
Fix an annotated word $\hat{\sigma}\in\hat{\labs}^*$.
Write $\sigma := h(\hat{\sigma})$.
For each $i\in\{1,\dots,k+1\}$, let $\sigma_i := h(\proj{\hat{\sigma}}{i})$.
By \lemref{consistency}, we have $\hat{\sigma}\in\hat{L}_{\cnst}$ iff the partition
$\{\sigma_i\}_{1\le i\le k+1}$ satisfies \lemref{consistency}(1) and \lemref{consistency}(2).
Therefore, it suffices to show that $\aut_{\cnst}$ accepts $\hat{\sigma}$ iff
$\{\sigma_i\}_{1\le i\le k+1}$ satisfies \lemref{consistency}(1) and \lemref{consistency}(2).

\medskip
\noindent
We prove both directions by induction over prefixes. Throughout, for a prefix $\hat{\pi}$
we denote by $q_{\hat{\pi}}$ the (unique) state reached by $\aut_{\cnst}$ after reading $\hat{\pi}$.
We also use (standard) interval notation: $[a,b) = \{\ell \mid a \le \ell < b\}$ and
$(a,b] = \{\ell \mid a < \ell \le b\}$.

\medskip
\noindent
\textbf{Inductive state invariant.}
For any prefix $\hat{\pi}$ such that $q_{\hat{\pi}}\neq\bot$, writing
$q_{\hat{\pi}} = (\mathsf{T2S}, \mathsf{LastW}, \mathsf{SeenW}, \mathsf{ForbiddenW})$,
we maintain the following invariants:

\begin{enumerate}
\item[(I1)] For each thread $t$, $\mathsf{T2S}(t)$ equals the maximum slice index among
events of thread $t$ occurring in $\hat{\pi}$ (or $0$ if none).

\item[(I2)] For each location $x$, $\mathsf{LastW}(x)$ equals the slice index of the
\emph{last write to $x$} occurring in $\hat{\pi}$ (or $0$ if none).

\item[(I3)] For each location $x$, $\mathsf{SeenW}(x)$ equals the set of slice indices
in which a write to $x$ occurs in $\hat{\pi}$.

\item[(I4)] For each location $x$,
$\mathsf{ForbiddenW}(x)$ equals the union of all slice-intervals contributed by reads of $x$
already seen in $\hat{\pi}$, as follows.
For every read event $r$ on $x$ in $\hat{\pi}$ that is annotated with slice $i_r$,
let $j_r$ be the slice index of the rf-source of $r$ \emph{within the prefix $\hat{\pi}$}
(where $j_r=0$ if $r$ is orphan within $\hat{\pi}$). Then $r$ contributes:
\[
\begin{cases}
[1,i_r) & \text{if } j_r = 0 \text{ (orphan read)}\\
[j_r,i_r) & \text{if } j_r > 0 \text{ (non-orphan read)}
\end{cases}
\]
and $\mathsf{ForbiddenW}(x)$ is the union of all such contributions.
\end{enumerate}

\medskip
\noindent
\textbf{Consistent $\Rightarrow$ Accepted.}
We prove by induction on $|\hat{\pi}|$ the statement:
\[
\hat{\pi}\in \hat{L}_{\cnst}
\implies
q_{\hat{\pi}}\neq\bot \text{ and } q_{\hat{\pi}} \text{ satisfies (I1)--(I4)}.
\]

\smallskip
\noindent
\emph{Base.}
For $\hat{\pi}=\epsilon$, we have $q_{\hat{\pi}} = q^0_{\cnst}\neq\bot$ and (I1)--(I4) hold trivially.

\smallskip
\noindent
\emph{Step.}
Let $\hat{\pi}' = \hat{\pi}\cdot (e,i)$ where $e=\ev{t,op(x)}$, and assume $\hat{\pi}'\in\hat{L}_{\cnst}$.
Since $\hat{L}_{\cnst}$ is prefix closed, we have $\hat{\pi}\in\hat{L}_{\cnst}$.
By IH, $q_{\hat{\pi}}=p\neq\bot$ and $p$ satisfies (I1)--(I4).
We show $\delta_{\cnst}(p,(e,i))\neq\bot$, and that the resulting state satisfies (I1)--(I4).

Write $p=(\mathsf{T2S}_p,\mathsf{LastW}_p,\mathsf{SeenW}_p,\mathsf{ForbiddenW}_p)$.

\emph{(a) Thread monotonicity.}
Suppose for contradiction that $\mathsf{T2S}_p(t) > i$, so the automaton would reject.
By (I1), $\mathsf{T2S}_p(t)$ is the maximum slice index of thread $t$ in the prefix $\hat{\pi}$.
Thus $\hat{\pi}$ contains an event of thread $t$ annotated with slice $>i$ that precedes $(e,i)$ in program order,
violating \lemref{consistency}(1) for the consistent prefix $\hat{\pi}'$. Contradiction.

\emph{(b) Write case.}
Assume $op=\wt$. Suppose for contradiction that $i\in \mathsf{ForbiddenW}_p(x)$, so the automaton would reject.
By (I4), there exists a read $r$ of $x$ already in $\hat{\pi}$ with slice $i_r$
such that $i$ lies in the interval contributed by $r$.

If $r$ is orphan in $\hat{\pi}$, then it contributed $[1,i_r)$ and hence $i<i_r$.
But then $\hat{\pi}'$ contains a write to $x$ in slice $i<i_r$, contradicting
\lemref{consistency}(2a) for the read $r$ in the consistent prefix $\hat{\pi}'$.

If $r$ is non-orphan in $\hat{\pi}$ with rf-source slice $j_r>0$, then it contributed $[j_r,i_r)$,
so $j_r\le i < i_r$.
The new write in slice $i$ is neither in a slice $\le j_r$ nor in a slice $\ge i_r$,
contradicting \lemref{consistency}(2b)(i) for the read $r$ in the consistent prefix $\hat{\pi}'$.

In both cases we contradict consistency of $\hat{\pi}'$, hence $i\notin\mathsf{ForbiddenW}_p(x)$ and no rejection occurs.

\emph{(c) Read case.}
Assume $op=\rd$.
Let $j := \mathsf{LastW}_p(x)$ (the slice of the last write to $x$ in $\hat{\pi}$, by (I2)).
Since $\hat{\pi}'$ is consistent, \lemref{consistency}(2) for this read implies:
(i) the rf-source slice satisfies $j\le i$ (with $j=0$ allowed for orphan),
and (ii) there is no write to $x$ in any slice strictly between $j$ and $i$,
and also no write to $x$ in slice $i$ preceding this read when $j<i$.
By (I3), this is exactly the condition
$\mathsf{SeenW}_p(x)\cap (j,i] = \emptyset$.
Thus neither disjunct in the automaton's read-rejection test can hold, and no rejection occurs.

\smallskip
Having shown $\delta_{\cnst}(p,(e,i))\neq\bot$, the update rules of the automaton immediately preserve (I1)--(I3).
For (I4), note that the only update to $\mathsf{ForbiddenW}$ occurs on reads, and it adds exactly
the interval $[j,i)$ where $j=\mathsf{LastW}_p(x)$ (with $j=0$ corresponding to adding $[1,i)$),
which is precisely the contribution required by the new read in $\hat{\pi}'$.
Hence (I4) is preserved as well.

This completes the induction. In particular, taking $\hat{\pi}=\hat{\sigma}$,
we conclude that if $\hat{\sigma}\in\hat{L}_{\cnst}$ then $\aut_{\cnst}$ accepts $\hat{\sigma}$.

\medskip
\noindent
\textbf{Accepted $\Rightarrow$ Consistent.}
We prove by induction on $|\hat{\pi}|$ the statement:
\[
q_{\hat{\pi}}\neq\bot
\implies
\hat{\pi}\in \hat{L}_{\cnst} \text{ and } q_{\hat{\pi}} \text{ satisfies (I1)--(I4)}.
\]

\smallskip
\noindent
\emph{Base.}
For $\hat{\pi}=\epsilon$, we have $q_{\hat{\pi}}=q^0_{\cnst}\neq\bot$ and $\epsilon\in\hat{L}_{\cnst}$, and (I1)--(I4) hold.

\smallskip
\noindent
\emph{Step.}
Let $\hat{\pi}'=\hat{\pi}\cdot(e,i)$ with $e=\ev{t,op(x)}$, and suppose $q_{\hat{\pi}'}\neq\bot$.
Then also $q_{\hat{\pi}}\neq\bot$ (since $\bot$ is a sink).
By IH, $\hat{\pi}\in\hat{L}_{\cnst}$ and $p:=q_{\hat{\pi}}$ satisfies (I1)--(I4).
We show $\hat{\pi}'\in\hat{L}_{\cnst}$ by checking \lemref{consistency}(1) and \lemref{consistency}(2)
for the partition induced by slice annotations on the prefix $\hat{\pi}'$.

\emph{(1) PO alignment.}
Since $q_{\hat{\pi}'}\neq\bot$, the transition did not reject on the thread check,
hence $\mathsf{T2S}_p(t)\le i$.
By (I1), this means the new event's slice index does not decrease along thread $t$,
so \lemref{consistency}(1) continues to hold in $\hat{\pi}'$.

\emph{(2) RF alignment for the new event.}
If $op=\wt$, we must show that appending this write does not break \lemref{consistency}(2)
for any earlier read in $\hat{\pi}$.
Since $q_{\hat{\pi}'}\neq\bot$, the transition did not reject on the write check, so
$i\notin \mathsf{ForbiddenW}_p(x)$.
By (I4), this means that for every earlier read $r$ of $x$ in $\hat{\pi}$,
the slice $i$ does not lie in the interval forbidden by $r$.
Equivalently, appending this write in slice $i$ cannot violate \lemref{consistency}(2a) (if $r$ is orphan)
or \lemref{consistency}(2b) (if $r$ is non-orphan), for any such $r$.
Thus \lemref{consistency}(2) remains true for all reads already in the prefix.

If $op=\rd$, let $j:=\mathsf{LastW}_p(x)$.
Since $q_{\hat{\pi}'}\neq\bot$, the transition did not reject on the read check, so
$j\le i$ and $\mathsf{SeenW}_p(x)\cap (j,i]=\emptyset$.
By (I2)--(I3), this exactly enforces the constraints required by \lemref{consistency}(2)
for this new read in the prefix $\hat{\pi}'$ (orphan when $j=0$, and non-orphan when $j>0$).
Hence \lemref{consistency}(2) holds for the new read as well.

Therefore both \lemref{consistency}(1) and \lemref{consistency}(2) hold for $\hat{\pi}'$,
so $\hat{\pi}'\in\hat{L}_{\cnst}$.

Finally, preservation of (I1)--(I4) follows exactly as in Direction~1, by inspecting the updates.

This completes the induction. Taking $\hat{\pi}=\hat{\sigma}$,
if $\aut_{\cnst}$ accepts $\hat{\sigma}$ then $\hat{\sigma}\in\hat{L}_{\cnst}$.

\medskip
Combining both directions yields $L(\aut_{\cnst})=\hat{L}_{\cnst}$.
\end{proof}

\membershipRegular*

\begin{proof}
The proof of correctness follows from an inductive argument that establishes the following
invariant after each prefix of the word seen so far:

\begin{claim}
\claimlabel{membership-automaton-invariant}
Let $\hat{\tr} \in \hat{\labs}^*$ be an annotated execution and let $\hat{\pi}$ be a prefix of $\hat{\tr}$. 
Let $q_{\hat{\pi}} = (\trnsmemb)^*(\initmemb, \hat{\pi})$ be the state reached after reading prefix $\hat{\pi}$.

For every slice index $i \in \{1, 2, \ldots, k+1\}$ and every state $p \in \states$:
\[q_{\hat{\pi}}(i, p) = \trns^*(p, h(\proj{\hat{\pi}}{i}))\]

That is, $q_{\hat{\pi}}(i, p)$ is precisely the state that the original automaton $\aut$ reaches when starting from state $p$ and reading the string $h(\proj{\hat{\pi}}{i})$ (the $i$-th projection of the prefix $\hat{\pi}$ with annotations removed).

Furthermore, $\hat{\tr} \in \hat{L}_{\memb}$ if and only if the final state $q_{\hat{\tr}}$ satisfies the acceptance condition: there exists a sequence of states $p_1, p_2, \ldots, p_{k+1} \in \states$ such that $p_1 = q_{\hat{\tr}}(1, \init)$, for every $1 \leq i \leq k$, $p_{i+1} = q_{\hat{\tr}}(i+1, p_i)$, and $p_{k+1} \in \accpt$.
\end{claim}
The above invariant can be established through a straightforward induction proof and is skipped.

\end{proof}

\begin{claim}
\claimlabel{forward-closure-intersection-homomorphism}
Let $L$ be a regular language and $k \in \natsp$.
Let $\hat{L}_{\cnst\land \memb}$ be as defined in \secref{pre-image-slices}.
We have $\pre{L}{\kslice{k}} = h(\hat{L}_{\cnst\land \memb})$.
\end{claim}

\begin{proof}
We show both directions of the equality.

\noindent ($\subseteq$) Let $\tr \in \pre{L}{\kslice{k}}$. By definition, there exists $\rho$ such that $\tr \kslice{k} \rho$ and $\rho \in L$. By Definition~\ref{def:k-slice-reordering}, there exist disjoint subsequences $\tr_1, \tr_2, \ldots, \tr_{k+1}$ of $\tr$ such that $\rho = \tr_1 \cdot \tr_2 \cdots \tr_{k+1}$ and $\tr \rfeq \rho$.

Define an annotation $\hat{\tr}$ by assigning each event $e \in \events{\tr_i}$ the annotation $i$. Then:
\begin{itemize}
\item $h(\hat{\tr}) = \tr$ (removing annotations gives back the original execution)
\item $\hat{\tr} \in \hat{L}_{\cnst}$ because the subsequences $\tr_1, \ldots, \tr_{k+1}$ satisfy the alignment conditions of Lemma~\ref{lem:consistency} (since $\tr \rfeq \rho$)
\item $\hat{\tr} \in \hat{L}_{\memb}$ because $h(\proj{\hat{\tr}}{1}) h(\proj{\hat{\tr}}{2}) \cdots h(\proj{\hat{\tr}}{k+1}) = \tr_1 \cdot \tr_2 \cdots \tr_{k+1} = \rho \in L$
\end{itemize}

Therefore $\hat{\tr} \in \hat{L}_{\cnst} \cap \hat{L}_{\memb}$ and $\tr = h(\hat{\tr}) \in h(\hat{L}_{\cnst} \cap \hat{L}_{\memb})$.

\noindent ($\supseteq$) Let $\tr \in h(\hat{L}_{\cnst} \cap \hat{L}_{\memb})$. Then there exists $\hat{\tr} \in \hat{L}_{\cnst} \cap \hat{L}_{\memb}$ such that $h(\hat{\tr}) = \tr$.

Since $\hat{\tr} \in \hat{L}_{\cnst}$, by Lemma~\ref{lem:consistency}, the subsequences $\tr_1 = h(\proj{\hat{\tr}}{1}), \ldots, \tr_{k+1} = h(\proj{\hat{\tr}}{k+1})$ satisfy the alignment conditions, which means $\tr \rfeq \tr_1 \cdot \tr_2 \cdots \tr_{k+1}$.

Since $\hat{\tr} \in \hat{L}_{\memb}$, we have $h(\proj{\hat{\tr}}{1}) h(\proj{\hat{\tr}}{2}) \cdots h(\proj{\hat{\tr}}{k+1}) = \tr_1 \cdot \tr_2 \cdots \tr_{k+1} \in L$.

Let $\rho = \tr_1 \cdot \tr_2 \cdots \tr_{k+1}$. Then $\tr \kslice{k} \rho$ and $\rho \in L$, so $\tr \in \pre{L}{\kslice{k}}$.
\end{proof}

\forwardClosureRegular*

\begin{proof}
Follows from \lemref{consistency-regular}, \lemref{membership-regular}, \claimref{forward-closure-intersection-homomorphism}, and the fact that regular languages are closed under homomorphism.
\end{proof}


\subsection{Proofs from \secref{post-image-slices}}

\postImageLinearSpace*

\begin{proof}
The proof of \thmref{post-image-liear-space-hardness}
can be derived from the proof of linear space hardness result in context of the causal
concurrency question~\cite[Theorem 3.1]{FarzanMathur2024} under reads-from equivalence.
The input to this causal concurrency question is an execution $\tr$ together
with two distinctly marked events $\alpha$ and $\beta$ in it with $\alpha \trord{\tr} \beta$, 
and the output is YES iff there is a run $\rho \rfeq \tr$ such that, in $\rho$, 
the relative order of $\alpha$ and $\beta$ is flipped, i.e., $\beta \trord{\rho} \alpha$.
We leverage the same proof for our purposes.
More precisely, the reduction in~\cite[Theorem 3.1]{FarzanMathur2024} is from the linear-space hard
language:
\begin{align*}
L_n = \setpred{\vect{a}\#\vect{b}}{\vect{a}, \vect{b} \in \set{0,1}^n \text{ and } \vect{a} = \vect{b}}
\end{align*} 
Given a word of the form $w = a_1a_2 \ldots a_n \# b_1 b_2 \ldots b_n \in (0+1)^n\#(0+1)^n$, the reduction,
in constant space, constructs a run $\tr$ that contains exactly two threads $T_1$ and $T_2$ and contains
two events $\alpha = \ev{T_1, \rd(u)}$ and $\beta = \ev{T_2, \wt(u)}$ satisfying $\alpha \trord{\tr} \beta$.
The reduction is such that $w \in L_n$ iff there is a $\rho \rfeq \tr$ such that $\beta \trord{\tr} \alpha$,
or in other words $\rho \in L_{\sf OV}^{\alpha,\beta}$.
We observe that in fact, any such $\rho$ (if one exists) is such that $\sliceto{\rho}{\tr}$, i.e.,
$\tr$ is a slice reordering of $\rho$ because in $\tr$, all events of $T_1$ are ordered before all events of $T_2$, making $\tr$ a sliced reordering of $\rho$.
In other words, we also have $w \in L_n$ iff $\tr \in \post{L}{\slice}$.
Since $\slice \subseteq \kslice{1} \subseteq \kslice{k}$ for every $k$, we also have the more general result:
$w \in L_n$ iff $\tr \in \post{L}{\kslice{k}}$.
But since membership checking in $L_n$ cannot be done by algorithm that uses sub-linear space and
works in a streaming fashion, we have the desired lower bound.
\end{proof}

\section{Proofs from \secref{frontier}}

\postImageETHHardness*

\begin{proof}
To show the parametrized hardness result, we show a reduction
from INDEPENDENT-SET(c), which is known to be W[1]-hard 
for the parameter $c$ (size of the independent set).

\myparagraph{INDEPENDENT-SET(c)} Given an undirected graph $G = (V, E)$,
check if $G$ has an independent set of size  $c$, i.e., whether
there is a subset $S \subseteq V$ s.t. $|S| = c$ and for every $(u, v) \in E$,
$\set{u, v} \nsubseteq S$.

\myparagraph{Reduction}
We start with a graph $G = (V, E)$ and parameter $c$
and construct a run $\tr$ with $O(|V|+|E|)$ events
over $O(c)$ threads $\threads = \set{t_1, t_2, \ldots, t_{2c+2}}$ 
and $O(c)$ memory locations.
Our alphabet $\labs$ is thus of size $O(c)$ as well.
The language whose post-image we are interested in is simple:
$$L = \labs^* \ev{t_{2c+1}, \wt(x)} \ev{t_{2c+2}, \rd(x)} \labs^*.$$
As such, the construction of the run $\tr$ is similar to that in the proof of the W[1]-hardness
result for~\cite[Theorem 2.3]{Mathur2020b}, with a few 
differences.
First, in our setting, we do not have locks as part of the alphabet.
Instead, we rely on memory locations to simulate thread-local critical sections.
For the purpose of this reduction (where reorderings preserve all the events),
it suffices to replace an acquire event ${\tt acq}(\lk)$
(resp. release event ${\tt rel}(\lk)$) over lock $\lk$
with $\wt(x_{\lk})$ (resp. $\rd(x_{\lk})$), where 
$x_{\lk}$ is a distinguished memory location corresponding to the lock $\lk$.
Further, we require that the corresponding read/write pairs thus introduced
are related by reads-from; this ensures mutual exclusion on critical sections
induced by the same lock.
Since this translation is straightforward, in the rest of the description,
we will avoid emphasizing this technicality and instead work with acquire and release events anyway.
The second, and more crucial difference is in the additional 
reasoning we will need to do to ensure that the hardness is in the parameter $k$
(slice height) and not just the number of threads (as in the original reduction).
We cater to this by showing that if $(G, c)$ is a YES instance, then
there is a $\rho \in L$ for which $\ksliceto{k}{\rho}{\tr}$ for $k = 3c + 2 \in \Theta(c)$.
Next, if $(G, c)$ is a NO instance, then there is no $\rho \rfeq \tr$ for which $\rho \in L$;
this is already established in~\cite[Theorem 2.3]{Mathur2020b}, so we will skip
repeating this part of the proof.

While the run $\tr$ we construct is identical to that in~\cite[Theorem 2.3]{Mathur2020b},
for the sake of completeness we give some details here.
Overall $\tr$ has the following form:
\[
	\tr = \tau_{2c+1} \cdot \gamma_{1} \cdot \gamma_{2} \cdots \gamma_{c} \cdot \tau_{2c+2}
\]
Here, $\tau_{2c+1}$ and $\tau_{2c+2}$ respectively contain events of threads
$t_{2c+1}$ and $t_{2c+2}$.
For each $i \in \set{1, \ldots, c}$, the sequence $\gamma_{i}$ is an interleaving of
events of threads $t_{i}$ and $t_{c + i}$.
Informally, the per-thread sequence is identical for the threads $t_1, t_2, \ldots, t_{c}$;
we will denote these by $\tau_{1}, \tau_{2} \ldots \tau_{c}$.
Likewise, the per-thread sequence is identical for the threads 
$t_{c+1}, t_{c+2}, \ldots, t_{2c}$; 
we will denote these by $\tau_{c+1}, \ldots, \tau_{2c}$.
The thread $t_{c+i}$ serves as an auxiliary thread for the main thread $t_i$ 
(for every $i \in \set{1, \ldots, c}$).
Let us now describe each of these components in detail, and omit thread identifiers
when clear from context:
\begin{itemize}
	\item The sequence $\tau_{2c+1}$ is a singleton sequence that writes to $x$:
	\[
		\tau_{2c+1} = \wt(x)
	\]
	\item The sequence $\tau_{2c+2}$ contains a sequence of reads, followed by a nested critical section that contains a read of $x$:
	\[
		\tau_{2c+2} = \rd(s_1) \cdot \rd(s_2) \cdots \rd(s_c) \cdot \acq(\lk_1) \cdots \acq(\lk_c) \cdot \rd(x) \cdot \rel(\lk_c) \cdots \rel(\lk_1)
	\]
	\item The sequence $\gamma_i$ (comprising of events of $t_i$ and $t_{c+i}$)
	together encode the graph.
	Informally, the sequence $\tau_i$ of events of thread $t_i$ is obtained by
	concatenating $n = |V|$ smaller sequences, one for each vertex:
	\[
		\tau_i = \tau_i^1 \cdot \tau_i^2 \cdots \tau_i^n
	\]
	where $\tau_i^j$ encodes the $j^\text{th}$ vertex
	as a critical section on locks $\set{\lk_{\set{j, v}}}_{v \in E_j}$, where 
	$E_j$ denotes the set of neighbors of $j$.
	Inside the critical section of $\tau_i^j$, there are two events:
	a $\wt(y_i)$ and a $\rd(z_i)$, except for $j=1$ and $j=n$.
	That is, for $1 < j < n$, we have:
	\[
		\tau_i^j = \acq_i(\lk_{j, v_1}) \cdots \acq_i(\lk_{j, v_d}) \cdot \wt^j(y_i) \cdot \rd^j(z_i) \cdot \rel_i(\lk_{j, v_d}) \cdots \rel_i(\lk_{j, v_1})
	\] 
	where $v_1, \ldots, v_d$ is an enumeration of $E_j$
	and the subscript $i$ in $\acq_i/\rel_i$ and the superscript $j$ in $\rd^j/\wt^j$
	are just for ease of refering.
	For $j=1$, we have:
	\[
		\tau_i^1 = \acq_i(\lk_{1, v_1}) \cdots \acq_i(\lk_{1, v_d}) \cdot \wt(s_i) \cdot \rd^1(z_i) \cdot \rel_i(\lk_{1, v_d}) \cdots \rel_i(\lk_{1, v_1})
	\] 
	For $j=n$, we have:
	\[
		\tau_i^n = \acq_i(\lk_{n, v_1}) \cdots \acq_i(\lk_{n, v_d}) \cdot \wt^n(y_i) \cdot \rd_i(x) \cdot \rel_i(\lk_{1, v_d}) \cdots \rel_i(\lk_{1, v_1})
	\]
	In thread $t_{c+i}$, we have $n-1$ partitions, where its $j^\text{th}$ partition interleaves in $\gamma_i$ with both $\tau_i^j$ and $\tau_i^{j+1}$.
	The sequence of its events is $\tau_{c+i} = \tau_{c+i}^1 \cdot \tau_{c+i}^2 \cdots \tau_{c+i}^{n-1}$,
	where for every $1 \leq j \leq n-1$, we have:
	\[
		\tau_{c+i}^j = \acq^j(\lk_i) \cdot \wt^j(z_i) \cdot \rd^{j+1}(y_i) \cdot \rel^j(\lk_i) 
	\] 
	where the superscript $j$ in $\acq^j/\rel^j/\rd^j/\wt^j$ are
	for ease of referring. 
	The interleaving $\gamma_i$ is obtained so that,
	for each $i \leq j < n$, the following reads-from constraints
	are met with fewest possible context switching:
	\begin{align*}
	\begin{array}{rcl}
	(\ev{t_{c+i}, \wt^j(z_i)}, \ev{t_{i}, \rd^j(z_i)}) &\in& \rf{\tr} \\
	(\ev{t_{i}, \wt^{j+1}(y_i)}, \ev{t_{c+i}, \rd^{j+1}(y_i)}) &\in& \rf{\tr}
	\end{array}
	\end{align*}
	\end{itemize}
Observe that the number of events is $O(c \cdot (|V| + |E|))$.
The number of threads is $O(c)$, the number of memory locations (which includes locks) is 
$O(|V| + |E|)$.

\myparagraph{Correctness of reduction}
Since membership in the language $L$ essentially reduces to the existence
of a predictive data race on $x$, we omit the proof of the direction
`if $G$ does not have a $c$-sized independent set, then there is no
$\rho \in L$ such that $\rho \rfeq \tr$'.
This stronger statement was already established in~\cite[Theorem~2.3]{Mathur2020b}.

We therefore focus on the more interesting direction.
Assume that $G$ has an independent set
\[
  S = \set{v_1, v_2, \ldots, v_c} \subseteq V
\]
of size $c$, where $v_i$ denotes the vertex selected in the $i$-th copy
of the gadget.
We show that there exists a run $\rho \in L$ such that
$\ksliceto{k}{\rho}{\tr}$ for $k = 4c+2$.

\myparagraph{Decomposition of the blocks $\gamma_i$}
Recall that
\[
  \tr \;=\; \tau_{2c+1} \cdot \gamma_1 \cdot \gamma_2 \cdots \gamma_c \cdot \tau_{2c+2},
\]
where for each $i \in \set{1,\ldots,c}$, the block $\gamma_i$ is a fixed
interleaving of the events of threads $t_i$ and $t_{c+i}$.
Fix $i \in \set{1,\ldots,c}$ and let $v_i$ be the $i$-th vertex in the
chosen independent set $S$.
Write $n = |V|$ and identify vertices of $G$ with $\set{1,\ldots,n}$.

We decompose $\gamma_i$ into three (possibly empty) contiguous substrings
\[
  \gamma_i \;=\; \gamma_i^{\mathsf{start}} \cdot \gamma_i^{\mathsf{mid}} \cdot \gamma_i^{\mathsf{end}},
\]
using cut points defined by distinguished events that always exist.

\smallskip
\noindent\emph{Distinguished events.}
Let $e_i^{\mathsf{main}}$ denote the unique ``selection write'' event in
thread $t_i$ corresponding to vertex $v_i$:
\[
  e_i^{\mathsf{main}} \;=\;
  \begin{cases}
    \ev{t_i,\wt(s_i)} & \text{if } v_i = 1,\\[1mm]
    \ev{t_i,\wt^{v_i}(y_i)} & \text{if } 2 \le v_i \le n.
  \end{cases}
\]
Let $e_i^{\mathsf{aux}}$ denote the auxiliary write in thread $t_{c+i}$
associated with the gadget for $v_i$:
\[
  e_i^{\mathsf{aux}} \;=\;
  \begin{cases}
    \ev{t_{c+i},\wt^{v_i}(z_i)} & \text{if } 1 \le v_i \le n-1,\\[1mm]
    \ev{t_{c+i},\wt^{n-1}(z_i)} & \text{if } v_i = n.
  \end{cases}
\]
Finally, let $e_i^{x}$ denote the unique read of $x$ in thread $t_i$:
\[
  e_i^{x} \;=\; \ev{t_i,\rd_i(x)},
\]
which occurs in the last vertex gadget of $\tau_i$.

\smallskip
\noindent\emph{The three substrings.}
\begin{itemize}
  \item $\gamma_i^{\mathsf{start}}$ is the \emph{unique shortest prefix} of
  $\gamma_i$ that contains both $e_i^{\mathsf{main}}$ and $e_i^{\mathsf{aux}}$.
  Equivalently, it ends at the later of these two events in the total order
  of $\gamma_i$.

  \item $\gamma_i^{\mathsf{mid}}$ is the (possibly empty) substring of
  $\gamma_i$ that begins immediately after $\gamma_i^{\mathsf{start}}$
  and ends at $e_i^{x}$ (inclusive). Equivalently,
  $\gamma_i^{\mathsf{start}} \cdot \gamma_i^{\mathsf{mid}}$ is the shortest
  prefix of $\gamma_i$ containing $e_i^{x}$.

  \item $\gamma_i^{\mathsf{end}}$ is the (possibly empty) suffix of
  $\gamma_i$ consisting of all events occurring strictly after $e_i^{x}$.
\end{itemize}

By construction, the three substrings are uniquely defined, preserve
program order within each thread, and satisfy
$\gamma_i = \gamma_i^{\mathsf{start}} \cdot \gamma_i^{\mathsf{mid}} \cdot \gamma_i^{\mathsf{end}}$.

\myparagraph{Construction of the witness run $\rho$}
We now define the run $\rho$ as the following concatenation:
\[
\begin{aligned}
  \rho \;=\;&
  \gamma_1^{\mathsf{start}} \cdot \gamma_2^{\mathsf{start}} \cdots \gamma_c^{\mathsf{start}}
  \cdot
  \gamma_1^{\mathsf{mid}} \cdot \gamma_2^{\mathsf{mid}} \cdots \gamma_c^{\mathsf{mid}}
  \cdot
  \tau_{2c+1} \cdot \tau_{2c+2}^{\mathsf{init}} \\
  &\cdot
  \gamma_1^{\mathsf{end}} \cdot \gamma_2^{\mathsf{end}} \cdots \gamma_c^{\mathsf{end}}
  \cdot
  \tau_{2c+2}^{\mathsf{end}},
\end{aligned}
\]
where $\tau_{2c+2}^{\mathsf{init}}$ is the prefix of $\tau_{2c+2}$ up to
(and including) the event $\rd(x)$, and $\tau_{2c+2}^{\mathsf{end}}$ is the
remaining suffix.

The ordering above preserves program order within each thread and all
reads-from constraints induced by the construction.
In particular, the write $\wt(x)$ in $\tau_{2c+1}$ precedes the read
$\rd(x)$ in $\tau_{2c+2}$, and no other write to $x$ exists.
Hence $\rho \rfeq \tr$ and $\rho \in L$.

\myparagraph{Bounding the slice height}
Each block $\gamma_i$ contributes at most three subsequences
$\gamma_i^{\mathsf{start}}, \gamma_i^{\mathsf{mid}}, \gamma_i^{\mathsf{end}}$.
Thus the events of $\gamma_1,\ldots,\gamma_c$ contribute at most $3c$
subsequences.
In addition, $\tau_{2c+1}$ contributes one subsequence, and $\tau_{2c+2}$
is split into two subsequences $\tau_{2c+2}^{\mathsf{init}}$ and
$\tau_{2c+2}^{\mathsf{end}}$.
Altogether, $\tr$ is partitioned into at most $3c+3$ subsequences whose
concatenation yields $\rho$.
Therefore,
\[
  \ksliceto{k}{\rho}{\tr}
  \quad\text{holds for}\quad
  k = (3c+3)-1 = 3c+2.
\]

\myparagraph{Time complexity of reduction}
The construction of the run $\tr$ takes time $O(|G|\cdot c)$ since each $\gamma_i$
has size $O(|V|+|E|)$, because for each vertex, we has as many critical
sections as its neighbors.

\end{proof}

\section{Proofs from \secref{discussion}}

\begin{lemma}[Trace and Reads-From Closure Coincide]
\lemlabel{trace-rf-coincide}
Let $L = (ab + \bar{a}\bar{b})^*$
where $a,b,\bar a,\bar b$ are as in \thmref{beyond-trace}.
Then $[L]_{\mazeq} = [L]_{\rfeq}.$
Moreover, for any reordering relation $R$ such that
$\mazeq \subseteq R \subseteq \rfeq$, we have $\pre{L}{R} = \post{L}{R} = [L]_{\mazeq}.$
\end{lemma}

\begin{proof}
Let $\Gamma = \set{a,b,\bar a,\bar b}$ and note that $L \subseteq \Gamma^*$. 
Observe that if $\rho \in \Gamma^*$ and if $\sigma \sim \rho$ (for any sound reordering relation $\sim$),
then $\sigma \in \Gamma^*$. So it suffices to focus on just the sub-alphabet $\Gamma$.
Indeed, $[L]_{\mazeq} \subseteq \Gamma^*$ and $[L]_{\rfeq} \subseteq \Gamma^*$.

Fix $\sigma,\rho \in \Gamma^*$.
Since all events in $\Gamma$ are writes, both $\sigma$ and $\rho$ contain no reads,
so $\rf{\sigma}=\rf{\rho}=\emptyset$.
Therefore, $\sigma \rfeq \rho$ reduces to equality of the program-order projections:
\[
  \sigma \rfeq \rho
  \quad\Longleftrightarrow\quad
  \sigma\!\downharpoonright_T = \rho\!\downharpoonright_T
  \ \text{ and }\
  \sigma\!\downharpoonright_{\bar T} = \rho\!\downharpoonright_{\bar T}.
\]
On the other hand, under the Mazurkiewicz trace model with the standard
dependence relation, two actions from distinct threads are dependent only if
they conflict on a shared location.
Here, the four writes target pairwise distinct locations $x,y,\bar x,\bar y$,
so there are \emph{no} cross-thread dependences between $T$ and $\bar T$.
Hence trace equivalence on $\Gamma^*$ is also exactly equality of per-thread
projections:
\[
  \sigma \mazeq \rho
  \quad\Longleftrightarrow\quad
  \sigma\!\downharpoonright_T = \rho\!\downharpoonright_T
  \ \text{ and }\
  \sigma\!\downharpoonright_{\bar T} = \rho\!\downharpoonright_{\bar T}.
\]
Combining the two characterizations yields
\[
  \sigma \mazeq \rho \quad\Longleftrightarrow\quad \sigma \rfeq \rho
  \qquad\text{for all } \sigma,\rho \in \Gamma^*.
\]
This immediately gives $[L]_{\mazeq} = [L]_{\rfeq}$.

Now consider a sound equivalence $R$ that subsumes $\mazeq$,
i.e., it satisfies $\mazeq \subseteq R \subseteq \rfeq$.
It is clear that $\pre{L}{R}, \post{L}{R} \subseteq [L]_{\rfeq}$.
First, consider  a $\sigma \in [L]_{\mazeq}$.
Then $\exists \rho \in L$ with
$\sigma \mazeq \rho$, and since $\mazeq \subseteq R$ we get $(\sigma,\rho)\in R$,
hence $\sigma \in \pre{L}{R}$.
Therefore $[L]_{\mazeq} \subseteq \pre{L}{R} \subseteq [L]_{\rfeq}$. 
Symmetrically, consider a $\sigma \in [L]_{\mazeq}$.
Then $\exists \rho \in L$ with
$\sigma \mazeq \rho$, and since $\mazeq \subseteq R$ we get $(\rho,\sigma)\in R$,
hence $\sigma \in \post{L}{R}$.
Therefore $[L]_{\mazeq} \subseteq \post{L}{R} \subseteq [L]_{\rfeq}$.
Together, we have $\pre{L}{R} = \post{L}{R} = [L]_{\mazeq} = [L]_{\rfeq}$.
\end{proof}

Let us now move to the proof of \thmref{beyond-trace}:

\BeyondTraceHardness*

\begin{proof}
At a high level, we show that the membership problem for
$\pre{L}{R}$ (equivalently $\post{L}{R}$) has a one-pass constant-space
reduction from the problem of membership in the following language,
which is known to admit a linear space lower bound in the streaming
setting (here $n \in \nats$):
\[
  L_n
  = \setpred{a_1a_2\cdots a_n \# b_1 b_2 \cdots b_n}
           {\forall i \leq n,\ a_i, b_i \in \set{0,1},\ a_i = b_i}.
\]

By \lemref{trace-rf-coincide}, for any sound reordering relation
$R$ with $\mazeq \subseteq R$, we have
\[
  \pre{L}{R} = \post{L}{R} = [L]_{\mazeq}.
\]
Thus it suffices to prove a streaming lower bound for membership in
$[L]_{\mazeq}$.

\myparagraph{Reduction}
We describe a one-pass constant-space transducer that maps a word
\[
  w = a_1a_2\cdots a_n \# b_1b_2\cdots b_n \in \set{0,1}^n \# \set{0,1}^n
\]
to a word $\sigma \in \labs^*$ such that
\[
  w \in L_n \iff \sigma \in [L]{\mazeq}.
\]

Before the symbol $\#$, the transducer outputs $a$ for each $0$ and $b$
for each $1$; after $\#$, it outputs $\bar a$ for each $0$ and $\bar b$
for each $1$.
Thus
\[
  \sigma \in \{a,b\}^n \{\bar a,\bar b\}^n,
\]
and the transducer clearly operates in one pass using $O(1)$ working
space.

\myparagraph{Case $\bar a = \bar b$}
Assume $w \in L_n$, i.e.\ $a_i = b_i$ for all $i$.
Then $\sigma$ is a concatenation of $n$ blocks, each equal to either
$ab$ or $\bar a \bar b$, and hence $\sigma \in L \subseteq [L]_{\mazeq}$.

\myparagraph{Case $\bar a \neq \bar b$}
Assume $w \notin L_n$, and let $i$ be the smallest index such that
$a_i \neq b_i$.
Then the projection of $\sigma$ to thread $T$ differs, at position $i$,
from the projection to thread $\bar T$ under the correspondence induced
by $L$.
Since trace equivalence preserves per-thread projections, it follows
that $\sigma \notin [L]_{\mazeq}$.

\myparagraph{One-pass space lower bound}
We have shown that the transducer maps $w$ to $\sigma$ such that
\[
  w \in L_n \iff \sigma \in \pre{L}{R}.
\]
Since deterministic one-pass streaming membership for $L_n$ requires
$\Omega(n)$ space, the same lower bound applies to membership in
$\pre{L}{R}$ and $\post{L}{R}$.

\myparagraph{Time--space tradeoff}
The reduction described above is streaming, length-preserving up to
constant factors, and uses constant additional space.
Therefore, any multi-pass streaming algorithm for membership in
$\pre{L}{R}$ (or $\post{L}{R}$) with time $T(n)$ and space $S(n)$
can be composed with this reduction to obtain a multi-pass streaming
algorithm for $L_n$ with asymptotically the same resource bounds.

Since membership in $L_n$ admits the time--space tradeoff lower bound
\[
  S(n) \cdot T(n) \in \Omega(n^2),
\]
the same bound holds for membership in $\pre{L}{R}$ and $\post{L}{R}$.
This completes the proof.
\end{proof}

\end{document}